\documentclass[11pt]{article}
\usepackage{amssymb,amsfonts,amsmath,amsthm,amscd,dsfont,mathrsfs,bm}
\usepackage{graphicx,float,psfrag,epsfig}
\usepackage{wrapfig}
\usepackage{hyperref}
\usepackage{algorithm,algorithmic}
\usepackage{booktabs}
\usepackage{color}
\usepackage{xcolor}

\DeclareMathAlphabet{\mathpzc}{OT1}{pzc}{m}{it}

\newtheorem{propo}{Proposition}[]
\newtheorem{lemma}{Lemma}

\newtheorem{coro}[propo]{Corollary}
\newtheorem{thm}[propo]{Theorem}

\DeclareMathOperator*{\argmin}{arg\,min}

\def\cE{{\cal E}}

\def\cM{{\cal M}}

\def\tD{{\tilde D}}

\def\hp{{\hat p}}
\def\hq{{\hat q}}

\def\hd{{\widehat d}}

\def\var{{\rm Var}}
\def\Poi{{\rm Poi}}

\def\cM{{\cal M}}

\def\cQ{{\cal Q}}

\def\reals{{\mathbb R}}
\def\prob{{\mathbb P}}
\def\Z{\mathbb Z}
\def\E{\mathbb E}
\def\Var{{\rm Var}\,}
\def\Poi{{\rm Poi}\,}

\def\ind{{\mathbb I}}

\def\TV{{\rm TV}}
\def\Bmu{{\bm{ \mu}}}

\def\tnu{{\tilde \nu}}
\def\cD{{\cal D}}
\def\cQ{{\cal Q}}

\begin{document}
\title{Minimax Rates of Estimating Approximate Differential Privacy}
\author{Xiyang Liu and Sewoong Oh\thanks{Emails: xiyangl@cs.washington.edu and sewoong@cs.washington.edu}\\
Allen School of Computer Science and Engineering, University of Washington}
\date{}

\maketitle

%
%
\begin{abstract} 

Differential privacy has become a  widely accepted notion of privacy, 
leading to the introduction and deployment of numerous 
privatization mechanisms.   
However, ensuring the privacy guarantee is an error-prone process, 
both in designing mechanisms and in implementing those mechanisms.  
Both types of errors will be greatly reduced, 
if we have a data-driven approach to verify privacy guarantees, 
from a black-box access to a mechanism.  
We pose it as a property estimation problem, 
and study the fundamental trade-offs involved in 
the accuracy in estimated privacy guarantees and the number of samples required. 
We introduce a novel estimator that uses polynomial approximation of a carefully chosen degree to 
optimally trade-off bias and variance. 
With $n$ samples, 
we show that this estimator achieves performance of 
a straightforward plug-in estimator with $n \ln n$ samples, a phenomenon referred to as effective sample size amplification. 
The minimax optimality of the proposed estimator is proved by comparing it to 
a matching fundamental lower bound. 

\end{abstract} 

%
%


\section{Introduction}
\label{sec:intro}

Differential privacy is gaining popularity as an agreed upon measure of privacy leakage, widely used by 
the government to publish Census statistics \cite{census}, 
Google to aggregate user's choices in web-browser features \cite{google1,google2}, 
Apple to aggregate mobile user data \cite{apple1}, 
and smart meters in telemetry \cite{telemetry}.
As increasing number of  privatization mechanisms are introduced and deployed in the wild, 
it is critical to have countermeasures to check the fidelity of those mechanisms. 
Such techniques will allow us to hold accountable the deployment of privatization mechanisms 
if the claimed privacy guarantees are not met, and 
find and fix bugs in implementations of those mechanisms. 

A user-friendly tool for checking privacy guarantees    
is necessary for several reasons. 
Writing a program for a privatization mechanism is error-prone, 
as it involves complex probabilistic computations. 
Even with  customized languages for 
differential privacy, checking the end-to-end 
privacy guarantee of an implementation remains challenging \cite{BKO12,WDW19}. 
Furthermore, even when the implementation is error-free, 
there have been several cases where the mechanism designers have made errors in calculating the 
privacy guarantees, and falsely reported higher level of privacy \cite{LSL17,CM15}. 
This is evidence of an alarming issue that 
analytically checking the proof of a privacy guarantee 
is a challenging process even for an expert. 
An automated and data-driven algorithm for checking privacy guarantees will 
significantly reduce such errors in the implementation and the design. 
On other cases,  we are given very limited information about how 
the mechanism works, like Apple's white paper \cite{apple1}. 
The users are left to trust the claimed privacy guarantees. 

To address these issues, we propose 
a data-driven approach to 
estimate how much privacy is guaranteed, 
from a black-box access to a purportedly private mechanism. 
Our approach is based on 
an optimal polynomial approximation 
that gracefully trades off  bias and variance. 
We study the fundamental limit of how many samples are 
necessary to achieve a desired level of accuracy in the estimation, 
and show that the proposed approach achieves this fundamental bound.  









\bigskip\noindent{\bf Problem formulation.}
Differential privacy (DP) introduced in \cite{Dwo11} is a formal mathematical notion of privacy that 
is widespread, due to several key advantages.  
It gives one of the strongest guarantees, 
allows for precise mathematical analyses, 
and is intuitive to explain even to non-technical end-users. 
When accessing a database through a query, we say the query output is private if 
the output did not reveal whether a particular person's entry is in the database or not. 
Formally, we say two databases are {\em neighboring} if they only differ in one entry (one row in a table, for example).
Let $P_{\cQ,\cD}$ denote the distribution of the randomized output to a query $\cQ$ on a database $\cD$. 
We consider discrete valued mechanisms taking one of $S$ values, 
i.e.~the response to a query is in $[S]=\{1,\ldots,S\}$ for some integer $S$.
We say a mechanism guarantees $(\varepsilon,\delta)$-DP, \cite{Dwo11}, if the following holds 
\begin{eqnarray}
	\label{eq:defDP}
	P_{\cQ,\cD}(E) \;\leq\; e^\varepsilon P_{\cQ,\cD'}(E) +\delta \;, 
\end{eqnarray}
for some $\varepsilon \geq 0$, $\delta\in[0,1]$, and all subset $E\subseteq [S]$ and 
for all neighboring databases $\cD$ and $\cD'$. 
When $\delta=0$, $(\varepsilon,0)$-DP is referred to as (pure) differential privacy, and 
the general case of $\delta\geq0$ is referred to as {\em approximate} differential privacy. 
For pure DP, 
the above condition can be relaxed as 
\begin{eqnarray}
	\label{eq:defDP_pure}
	P_{\cQ,\cD}(x) \;\leq\; e^\varepsilon P_{\cQ,\cD'}(x) \;, 
\end{eqnarray}
for all {\em output symbol} $x \in [S]$, and 
for all neighboring databases $\cD$ and $\cD'$. 
This condition can now be checked, one symbol $x$ at a time from $[S]$, without having to enumerate all subsets $E\subseteq [S]$. 
This naturally leads to the following  algorithm. 

For a query $\cQ$ and two neighboring databases $\cD$ and $\cD'$ of interest, 
we need to verify the condition in Eq.~\eqref{eq:defDP_pure}. 
As we only have a black-box access to the mechanism, 
we collect $n$ responses from the mechanism on the two databases.  
We check the condition on the empirical distribution of those collected samples, for each $x\in[S]$.  
If it is violated for any $x$, we assert the mechanism to be not $(\varepsilon,0)$-DP and present $x$ as an evidence. 
Focusing only on pure DP, \cite{DWW18} proposed an approach similar to this, 
where they also give guidelines for choosing the databases  $\cD$ and $\cD'$ to test. 
However, their approach is only evaluated empirically, 
 no statistical analysis is provided, 
and a more general case of 
approximate DP is left as an open question, as  
the condition in Eq.~\eqref{eq:defDP} 
cannot be decoupled like Eq.~\eqref{eq:defDP_pure} when $\delta>0$. 

We propose an alternative approach from first principles 
to check the general approximate DP guarantees, and prove its minimax optimality.  
%
Given two probability measures $P=[p_1, \ldots,p_S]$ and $Q=[q_1,\ldots,q_S]$ over $[S] = \{1,\ldots,S\}$, 
we define the following {\em approximate DP divergence} with respect to $\varepsilon$ as 
\begin{eqnarray}
	d_\varepsilon(P\| Q) \;\; \triangleq \;\;  \sum_{i=1}^S [p_i - e^\varepsilon q_i]^+ \;\; = \;\; \E_{x\sim P}\Big[  \,\Big[ 1 - e^\varepsilon \frac{q_x}{p_x} \Big]^+ \, \Big] \;. 
	\label{eq:def_d}
\end{eqnarray} 
where $[x]^+=\max\{x,0\}$. 
The last representation indicates that this metric falls under a broader class of 
metrics known as $f$-divergences, with a special choice of $f(x)=[ 1 - e^\varepsilon x]^+$. 
From the definition of DP, it follows that a mechanism is $(\varepsilon,\delta)$-DP if and only if 
$d_\varepsilon( P_{\cQ,\cD}\|P_{\cQ,\cD'})  \leq \delta $ for all neighboring databases $\cD$ and $\cD'$. 
We propose estimating this divergence 
$d_\varepsilon(P_{\cQ,\cD}\|P_{\cQ,\cD'})$ from samples, and comparing it to the target $\delta$. 
This only requires number of operations scaling as $S\ln n$ where $n$ is the sample size.






In this paper, we suppose there is a specific query $\cQ$ of interest, 
and two neighboring databases $\cD$ and $\cD'$ have been already selected 
either by a statistician who has some side information on the structure of the mechanism 
or by some algorithm, such as those from  \cite{DWW18}. 
Without exploiting the structure 
(such as symmetry, exchangeability, or invariance to the entries of the database conditioned on the true output of the query), 
one cannot avoid having to check all possible combinations of neighboring databases. 
As a remedy, \cite{GM18} proposes checking randomly selected databases. 
This in turn ensures  a relaxed notion of privacy known as random differential privacy. 
Similarly, \cite{DJR13} proposed checking the {\em typical} databases, 
assuming there we have access to a prior  distribution over the databases. 
Our framework can be seamlessly incorporated with such higher-level routines to select databases. 
 
\bigskip
\noindent{\bf Contributions.} 
We study the problem of estimating the approximate differential privacy guaranteed by 
a mechanism, from a black-box access where we can sample from 
the mechanism output given a query $\cQ$, a database $\cD$, and a target $(\varepsilon,\delta)$. 
We first show that a straightforward  plug-in estimator of $d_\varepsilon(P\|Q)$ 
achieves mean squared error scaling as $(e^\varepsilon S )/ n$, where 
$S$ is the size of the alphabet and  $n$ is the number of samples used (Section~\ref{sec:Qmle}). 

In the regime where we fix $S$ and increase the sample size, this achieves the parametric rate of $1/n$, and cannot be improved upon. 
However, in many cases of practical interest where $S$ is comparable to $n$, 
we show that this can be improve upon with a more sophisticated estimator.
To this end, we introduce a novel estimator of $d_\varepsilon(P\|Q)$. 
The main idea is to identify the regimes of non-smoothness in $[p_i-e^\varepsilon q_i]^+$
where the plug-in  estimator has a large bias. 
We replace it by the uniformly best polynomial approximation of the non-smooth regime of the function, 
and estimate those polynomial from samples. 
By selecting appropriate  degree of the polynomial, 
we can optimally trade off the bias and variance. 
We provide an upper bound  on the error scaling as $(e^\varepsilon S )/ (n \ln n)$, 
when $S$ and $n$ are comparable. 
We prove that this is the best one can hope for, by providing a minimax lower bound that matches. 

We first show this for the case when we know $P$ and sample from $Q$ in Section~\ref{sec:Q}, 
to lay out the main technical insights while maintaining simple exposition. 
Then, we consider the practical scenario where both $P$ and $Q$ are accessed via samples, and 
provide an minimax optimal estimator in Section~\ref{sec:main}. 
This phenomenon is referred to as {\em effective sample size amplification}; 
one can achieve with $n$ samples a desired error rate, that would require $n \ln n$  samples for 
a plug-in estimator. 
We present numerical experiments supporting our theoretical predictions in 
Section~\ref{sec:exp}, and 
use our estimator to identify those real-world and artificial mechanisms that have 
incorrectly reported the privacy guarantees.

\bigskip
\noindent{\bf Related work.}
Formally guaranteeing differential privacy is a challenging and error-prone task. 
Principled approaches have been introduced to derive the end-to-end privacy loss 
of a software program that uses multiple differentially private accesses to the data, 
including SQL-based languages \cite{Mcs09}, 
higher-order functional languages \cite{RP10}, and 
imperative languages \cite{CGL11}. 
A unifying recipe of these approaches is to 
combine standard mechanisms like Laplacian, exponential, and Gaussian mechanisms 
\cite{Dwo11,MT07,GKOV15},
and calculate the end-to-end privacy loss using the composition theorem \cite{KOV17}. 
To extend these approaches to 
a more general notion of approximate differential privacy 
and allow non-standard components, principled approaches have been proposed 
in \cite{BKO12,WDW19}. The main idea is to perform fine-grained 
reasoning about  the complex probabilistic computations involved. 
These existing approaches have been preemptive measures 
 aimed at automated verification of privacy of the source code. 
Instead, we seek a post-hoc measure of 
estimating the privacy guarantee, given a black-box access to a privatization mechanism and not the source code. 

 Several published mechanisms have mistakes in the privacy guarantees. 
 These are variations of a popular mechanism known as Sparse Vector Technique (SVT). 
The original SVT first generates a random threshold. 
To answer a sequence of queries, it adds a random noise to each query output, 
and respond whether this is higher than the threshold (true) or not (false). 
This process continues until either all queries are answered, each with a privatized Boolean response or 
if the number of trues meets a pre-defined bound. 
Several variations violate claimed DP guarantees. 
\cite{SCM14} proposes 
variant of SVT with no noise adding and no bound on the number of trues, 
which does not satisfy DP for any values of $\varepsilon$. 
\cite{CXZ15} makes a similar false claim, proposing a version of SVT with no bound on the number of trues.
\cite{LC14} adds a smaller noise independent of the bound on the trues. 
\cite{Rot11} outputs  the actual noisy vector instead of the Boolean vector.

A formal investigation into verifying DP guarantees of a given mechanism was addressed in \cite{DJR13}. 
DP condition is translated into a certain 
Lipschitz condition on $P_{\cQ,\cD}$ over the databases $\cD$, 
and a Lipschitz tester is proposed to check the conditions. 
However, this approach is not data driven, as it requires the knowledge of the  distribution $P_{\cQ,\cD}$ 
and no sampling of the mechanism outputs is involved. 
 \cite{GM18} analyzes tradeoffs involved in testing DP guarantees. 
It is shown that one cannot get accurate testing without sacrificing the privacy of the databases used in the testing. 
Hence, when testing DP guarantees, one should not use databases that contain sensitive data. 
We compare some of the techniques involved in Section~\ref{sec:Qmle}.

Our techniques are inspired by 
a long line of research in property estimation of a distribution from samples. 
In particular, there has been significant recent advances 
for high-dimensional estimation problems, 
starting from entropy estimation for discrete random variables in 
\cite{VV13,JKYT15,WY16}. 
 The general recipe is to identify the regime where the property to be estimated is not smooth, 
 and use functional approximation to estimate a smoothed version of the property. 
 This has been widely successful in support recovery \cite{WY19}, 
 density estimation with $\ell_1$ loss \cite{HJW15}, 
and estimating Renyi entropy \cite{AOS17}. 
more recently, this technique has been applied to 
estimate certain divergences between two unknown distributions, for 
Kullback-Leibler divergence \cite{HJW16}, 
total variation distance \cite{JYT18}, and 
identity testing \cite{DKW18}. 
With carefully designed estimators, these approximation-based approaches can 
achieve improvement over typical parametric rate of $1/n$ error rate, 
sometimes referred to as 
{\em effective sample size amplification}. 

\bigskip\noindent{\bf Notations.} 
We let the alphabet of a discrete distribution be 
$[S] =\{1,\ldots, S\}$ for some positive integer $S$ denoting the size of the alphabet. 
We let $\cM_S$ denote the set of probability distributions over $[S]$. 
We use $f(n) \gtrsim g(n)$ to denote that $\sup_n f(n)/g(n) \geq C$ for some constant $C$, and  
 $f(n) \lesssim g(n)$ is analogously defined.  $f(n)\asymp g(n)$ denotes that $f(n) \gtrsim g(n)$ and 
 $f(n) \lesssim g(n)$.


%
%
\section{Estimating differential privacy guarantees from samples}
\label{sec:main1}

We want to estimate $d_\varepsilon(P\|Q)$ from a blackbox access to 
the mechanism outputs accessing two databases, i.e.~$P=P_{\cQ,\cD}$ and $Q=P_{\cQ,\cD'}$. 
We first consider a simpler case, where $P=[p_1,\ldots,p_S]$ is known and we observe samples from 
an unknown distribution $Q=[q_1,\ldots,q_S]$ in Section~\ref{sec:Q}. 
We cover this simpler case first to 
demonstrate  the main ideas on the algorithm design and 
analysis technique 
while maintaining the exposition simple. 
This paves the way for our main algorithmic and theoretical results in Section~\ref{sec:main}, 
where we only have access to samples from  both $P$ and $Q$.

%
%
\subsection{Estimating $d_\varepsilon(P\| Q)$ with known P}
\label{sec:Q}

For a given budget $n$, representing an upper bound on the expected number of samples we can collect,  
we propose sampling a random number $N$ of samples 
from Poisson distribution with mean $n$, i.e.~$N\sim \Poi(n)$.  
Then, each sample $X_j \in[S]$ is drawn from $Q$ for $j\in\{1,\ldots,N\}$, 
and we let $Q_n=[\hq_1,\ldots,\hq_S]$ denote the 
resulting histogram divided by $n$, 
such that 
$\hq_i \triangleq |\{j\in[N]: X_i=j\}|/n$. 

Note that $Q_n$ is not the standard  {\em empirical distribution}, as $\sum_i \hq_i \neq 1$ with high probability. 
However, in this paper we refer to $Q_n$ as empirical distribution of the samples. 
The empirical distribution would have been divided by $N$ instead of $n$. 
Instead, $Q_n$ is the maximum likelihood estimate of the true distribution $Q$.   
This Poisson sampling, together with the MLE construction of $Q_n$, 
ensures independence among $\{\hq_i\}_{i=1}^S$, 
making the analysis simpler. 
We first analyze the sample complexity of a simple plug-in estimator $d_\varepsilon(P\| Q_n )$ in Section~\ref{sec:Qmle}. 
This is well-defined regardless of whether $Q_n$ sums to one or not. 
We next show that we can 
  significantly improve the sample complexity by using a novel estimator, Algorithm~\ref{Qalgo},  in Section~\ref{sec:Qub}.  
We show minimax optimality of the proposed estimator by 
proving a matching  lower bound in Section~\ref{sec:Qlb}.

\subsubsection{Performance of the plug-in estimator}
\label{sec:Qmle}

The following result shows that it is necessary and sufficient to have $n \approx e^\varepsilon \, S$ samples 
to achieve an arbitrary desired error rate, if we use this plug-in estimator $d_\varepsilon(P\|Q_n)$, under the worst-case  $P$ and $Q$.  
Some assumption on $(P,Q)$ is inevitable as 
it is trivial to achieve zero error for any sample size, 
for example if $P$ and $Q$ have disjoint supports. 
Both $d_\varepsilon(P\| Q)$ and $d_\varepsilon(P\|Q_n)$ are 1 with probability one. 
We provide a proof in Section~\ref{sec:proof_Qmle}.
The bound in Eq.~\eqref{eq:Qmle_bound} also holds for $d_\varepsilon(P_n\|Q)$. 
The proof technique is analogous, and we omit the proof here. 

\begin{thm}
	\label{thm:Qmle}
	For any $\varepsilon\geq 0$, support size $S\in\Z^+$, and  distribution $P\in\cM_S$, 
	the plug-in estimator satisfies
	\begin{eqnarray}
	\sup_{Q \in \cM_S } \E_{Q} \big[\,|d_\varepsilon(P\| Q_n)-d_\varepsilon(P \| Q)|^2\,\big] 
	\;\; \lesssim \;\;	 \Big(\, \sum_{i=1}^S \, p_i \wedge \sqrt{\frac{e^\varepsilon p_i}{n }}\,  \Big)^2 +\frac{e^\varepsilon}{n}\;,
	\end{eqnarray}
	with expected number of samples $n$.
	If $S\geq 2$, we can also lower bound the worst case mean squared error as 
	\begin{eqnarray}
		\sup_{Q \in \cM_S } \E_{Q} \big[\,|d_\varepsilon(P\| Q_n)-d_\varepsilon(P \| Q)|^2\,\big] 
	\;\; \gtrsim \;\;	 \Big(\, \sum_{i=1}^S \, p_i \wedge \sqrt{\frac{e^\varepsilon p_i}{n }} \, \Big)^2\;.
	\end{eqnarray}
\end{thm}
When $n\geq e^\varepsilon S$, it follows that right-hand side of the lower bound is at least $e^\varepsilon S / n $, 
and the matching upper and lower bounds can be simplified as follows, for the worst case $P$ and $Q$. 
\begin{coro}
	\label{coro:Qmle}
	If $n\geq e^{\varepsilon}S$ and $S\geq 2$, we have
	\begin{eqnarray}
		\label{eq:Qmle_bound}
	\sup_{P,Q \in \cM_S } \E_{Q} \big[\,|d_\varepsilon(P \| Q_n )-d_\varepsilon(P \| Q)|^2\,\big] 
	\;\; \asymp \;\;	 \frac{e^{\varepsilon}S   \, }{n}\;.
	\end{eqnarray}
\end{coro}

A similar analysis was done in \cite{GM18}, 
which gives an upper bound scaling as $e^{2\varepsilon} S / n$. 
We tighten the analysis by a factor of $e^\varepsilon$, and provide a matching lower bound.


\subsubsection{Achieving optimal sample complexity with a polynomial approximation}
\label{sec:Qub}

We construct a minimax optimal estimator using techniques first introduced in \cite{WY16,JKYT15}
and adopted in several property estimation problems including \cite{HJW15,HJW16,AOS17,JYT18,DKW18,WY19}.

\floatname{algorithm}{Algorithm}
\renewcommand{\algorithmicrequire}{\textbf{Input:}}
\renewcommand{\algorithmicensure}{\textbf{Output:}}
\begin{algorithm}[h]
	\begin{algorithmic}
	\REQUIRE{target privacy $\varepsilon\in\reals^+$, query $\cQ$, neighboring databases $(\cD,\cD')$, pmf of $P_{\cQ,\cD}$ \\
	samples from $P_{\cQ,\cD'}$, 
	degree $K\in{\mathbb Z}^+$, constants $c_1,c_2 \in\reals^+$, expected sample size $2n$}
	\ENSURE{estimate $\hd_{\varepsilon,K,c_1,c_2}(P \| Q_n)$ of $d_\varepsilon( P_{\cQ,\cD}\| P_{\cQ,\cD'})$}
	\STATE{$P\leftarrow P_{\cQ,\cD}$}
	\STATE{Draw two independent sample sizes: $N_1 \leftarrow \Poi(n)$ and $N_2 \leftarrow \Poi(n)$}
	\STATE{Sample from $P_{\cQ,\cD'}$: $\{ X_{i,1} \}_{i=1}^{N_1} \in [S]^{N_1}$ and $\{ X_{i,2} \}_{i=1}^{N_2} \in [S]^{N_2}$}
	\STATE{$\hq_{i,j} \leftarrow \frac{|\{\ell\in [N_j]: X_{\ell,j} = i\}|}{n}$ for all $i\in[S]$ and $j\in\{1,2\}$}
	\STATE{$Q_{n,1} \leftarrow [\hq_{1,1},\ldots,\hq_{S,1}]$ and $Q_{n,2} \leftarrow [\hq_{1,2},\ldots,\hq_{S,2}]$ }
	\FOR{$i=1$ \TO $S$}
	\STATE{
	\hspace{2cm}$
	\delta_i \gets  \left\{ \begin{array}{rl}
	0 	&,\;\text{ if }  \hq_{i,1} > U(p_i;c_1,c_2) \\ 
	\tD_K(\hq_{i,2}; p_i)	&,\;\text{ if }  \hq_{i,1} \in U(p_i;c_1,c_2) \\
	{[p_i-e^\varepsilon \hq_{i,2}]^+} 	&,\;\text{ if }  \hq_{i,1} < U(p_i;c_1,c_2) 
	\end{array} \right. \label{eq:Qopt_estimator}
	$
	}
	\ENDFOR
	\STATE{$\hd_{\varepsilon,K,c_1,c_2}(P \| Q_n ) \gets  0 \vee ( 1 \wedge \sum_{i=1}^S \delta_i $})
	\end{algorithmic}
	\caption{Differential Privacy (DP) estimator with known $P$}
	\label{Qalgo}
\end{algorithm}

To simplify the analysis, 
we split the samples randomly into two partitions, each having an independent and identical distribution 
of $\Poi(n)$ samples from the multinomial distribution $Q$. 
We let $Q_{n,1}=[\hq_{1,1},\ldots,\hq_{S,1}]$ denote the count of the first set of $N_1\sim\Poi(n)$ samples (normalized by $n$), 
and $Q_{n,2}=[\hq_{1,1},\ldots, \hq_{S,1} ]$ the second set of $N_2\sim\Poi(n)$ samples. 
See Algorithm~\ref{Qalgo} for a formal definition. 
Note that for the analysis we are collecting $2n$ samples in total on average.
In all the experiments, however, we apply our estimator without  partitioning the samples. 
A major challenge  in achieving the minimax optimality is 
in handling the non-smoothness of the function $f(\hq_i;p_i) \triangleq [p_i-e^\varepsilon \hq_i]^+$ 
at $p_i \simeq e^\varepsilon \hq_i$. 
We use one set of samples to identify whether an outcome $i \in [S]$ is in 
the smooth regime ($\hq_{i,1} \notin U(p;c_1,c_2)$) or not ($\hq_{i,1} \in U(p;c_1,c_2)$), with an appropriately defined set function: 
\begin{eqnarray}
	U(p;c_1,c_2) \;\; \triangleq \;\; \left\{ 
	\begin{array}{rl}
		{[0,\frac{(c_1+c_2) \ln n }{n}]} &,\; \text{ if }p \leq\frac{c_1 e^{\varepsilon} \ln n}{n}\;, \\
		{\Big[e^{-\varepsilon}p-\sqrt{\frac{c_2 e^{-\varepsilon} p \ln n}{n}} , e^{-\varepsilon} p + \sqrt{\frac{c_2 e^{-\varepsilon} p \ln n}{n}}  \Big]}  & ,\; \text{ otherwise}, \\
	\end{array}
	\right.
	\label{eq:Qopt_region1}
\end{eqnarray}
for $c_1\geq c_2>0$ and $p \in[0,1]$. 
The scaling of the interval  is chosen carefully such that 
$(a)$ it is large enough for the probability of making a mistake on the which regime $(p_i,q_i)$ falls into 
to vanishes (Lemma~\ref{lem:Qopt_error}); and 
$(b)$ it is small enough for the variance of the polynomial approximation in the non-smooth regime to match that of the other regimes 
(Lemma~\ref{lem:Qopt_key}).
In the smooth regime, we use the plug-in estimator. 
In the non-smooth regime, we can improve the estimation error by  
using the best polynomial approximation of $f(x;p_i)=[p_i-e^\varepsilon\,x]^+$, which has a smaller bias: 
\begin{eqnarray}
	\label{eq:Qopt_poly}
	D_{K}(x;p) \;\; \triangleq \;\; \argmin_{P \in {\rm poly_K}} \max_{\tilde{x} \in U(p;c_1,c_1)} \big|\, [p-e^\varepsilon \tilde{x} ]^+-P(\tilde{x} ) \,\big| \;,
\end{eqnarray}
where ${\rm poly_K}$ is the set of polynomial functions of degree at most $K$, 
and we approximate $f(x;p)$ in an interval $U(p;c_1,c_1) \supset U(p;c_1,c_2)$ for any $c_1>c_2$. 
Having this slack of $c_1>c_2$ in the approximation allows us to guarantee the approximation quality, 
even if the actual $q_i$ is not exactly in the non-smooth regime 
$U(p;c_1,c_2)$. 
Once we have the polynomial approximation, we estimate this polynomial function $D_K(q_i;p_i)$ 
from samples, using 
 the {\em uniformly minimum variance unbiased estimator (MVUE)}. 

There are several advantages 
that makes this two-step process attractive. 
As we use an unbiased estimate of the polynomial, 
the {\em bias} is exactly the 
polynomial approximation error of $D_K(x;p_i)$, 
which scales as $(1/K)\sqrt{(p_i \ln n )/ n}$. Larger degree $K$ reduces the approximation error, 
and larger $n$ reduces the support of the domain we apply the approximation to in $U(p;c_1,c_1)$ 
(Lemma~\ref{lem:Qopt_key}). 
The {\em variance} is due to the sample estimation of the polynomial $D_K(x;p_i)$, 
which scales as $(B^K p_i \ln n )/ n$ for some universal constant $B$ (Lemma~\ref{lem:Qopt_key}).
Larger degree $K$ increases the variance. 
We prescribe choosing $K=c_3 \ln n$ for appropriate constant $c_3$ 
to optimize the bias-variance tradeoff in Algorithm~\ref{Qalgo}.
%
The resulting two-step polynomial estimator, 
has two characterizations, depending on where $p_i$ is. 
Let $\Delta = (c_1 \ln n )/n$. 

\vspace{0.2cm}
\noindent{\bf Case 1: $p_i \leq e^\varepsilon\Delta$ and $\hq_i\in U(p_i;c_1,c_1) = [0,2\Delta]$.}
We consider the function
$g(y)= [p_i - e^\varepsilon 2\Delta y ]^+$ by substituting $2\Delta y = x$ into $f(p_i,x)=[p_i-e^\varepsilon x ]^+$. 
Let $H_K(y)$ be the best polynomial approximation of $g(y)\in C[0,1]$ with order $K$, i.e. $H_K(y) = \argmin_{P \in {\rm poly_K}} \max_{y'\in [0,1]} \big|\, g(y')-P(y') \,\big| $ and denote it as $H_K(y)=\sum_{j=0}^Ka_jy^j$. Then $D_K(x;p_i) = H_K({x}/({2\Delta})) = \sum_{j=0}^K a_j (2\Delta)^{-j} x^j$. 
Once we have the polynomial approximation, we estimate with
the uniformly minimum variance unbiased estimator (MVUE) to estimate $D_K(\hq_i ; p_i)$.
\begin{eqnarray}
\tD_K(\hq_i ; p_i ) \;\; = \;\; \sum_{j=0}^K a_j (2\Delta)^{-j} \prod_{k=0}^{j-1} (\hq_i - \frac{k}{n}) \;.
\label{eq:Qopt_mvue}
\end{eqnarray}
Computing the $a_j$'s can be challenging, and we discuss this for the general case when $P$ is not known in Section~\ref{sec:main}. 

\vspace{0.2cm}
\noindent{\bf Case 2: $p_i > e^\varepsilon\Delta$ and $\hq_i \in [e^{-\varepsilon}p_i - \sqrt{e^{-\varepsilon} p_i \Delta}, e^{-\varepsilon}p_i + \sqrt{e^{-\varepsilon}p_i \Delta}]$.}
In this regime, the best polynomial approximation $D_K(x;p_i)$ of
$[p_i-e^\varepsilon x]^+$ is given by
\begin{eqnarray*}
D_K(x;p_i) \;\; =\;\; \frac{e^\varepsilon}{2}\sum_{j=0}^Kr_j \big(\sqrt{e^{-\varepsilon}p_i \Delta}\big)^{-j+1}(x-e^{-\varepsilon}p_i )^j+\frac{p_i -e^\varepsilon x}{2} \;,
\end{eqnarray*}
where $r_j$'s are defined from the best polynomial approximation $R_K(y)$ of $g(y) = |y|$ on $[-1,1]$ with order $K$:
$R_K(y) = \sum_{j=0}^K r_j y^j$.
The unique uniformly minimum variance unbiased estimator (MVUE) for $(q_i - p_i )^j$ is
\begin{eqnarray*}
g_{j, p_i}(\hq_i ) \;\; \triangleq \;\; \sum_{k=0}^{j}{j \choose k}(-p_i)^{j-k} \prod_{h=0}^{k-1}\left(\hq_i - \frac{h}{n}\right)\;,
\end{eqnarray*}
shown in  Lemma~\ref{lem:Qopt_moment}. 
Hence,
\begin{eqnarray*}
\tD_K(x;p_i) = \frac{e^\varepsilon}{2}\Big (\sum_{j=0}^Kr_j \big(\sqrt{e^{-\varepsilon}p_i \Delta}\big)^{-j+1}g_{j,e^{-\varepsilon}p_i }(\hq_i )+ g_{1,e^{-\varepsilon}p_i }(\hq_i )\Big)\;.
\end{eqnarray*}
The coefficients $r_j$'s only depend on $K$ and can be pre-computed and stored in a table. 

Choosing an appropriate $K$ scaling as $\ln n$, 
we can prove an upper bound on the error scaling as  $1/ (n \ln n)$.
This proves that the proposed estimator 
achieves the minimax optimal performance. 

\begin{thm} 
	\label{thm:Qopt}
	For any $P$, suppose  $ c \ln S \leq \ln n \leq C \ln(e^{-\varepsilon} \sum_{i=1}^S \sqrt{e^\varepsilon p_i} \wedge p_i \sqrt{n \ln n}) $ for some constants $c$ and $C$, then there exist constants $c_1, c_2$ and $c_3$ that only depends on $c$, $C$ and $\varepsilon$ such that 
	\begin{eqnarray}
		\label{eq:Qopt}
		\sup_{Q \in \cM_S}  \E_Q \big[ \, \big|\, \hd_{\varepsilon,K,c_1,c_2}( P \| Q_n )  - d_\varepsilon(P\| Q) \,\big|^2 \, \big] 
		\;\; \lesssim \;\;  \Big( \sum_{i=1}^S \, p_i \wedge \sqrt{\frac{e^\varepsilon p_i}{n \ln n}}  \Big)^2 \;,
	\end{eqnarray}
	for $K=c_3 \ln n $ and where  $\hd_{\varepsilon,K,c_1,c_2}$ is defined in Algorithm~\ref{Qalgo}.
\end{thm} 
We provide a proof in Section \ref{sec:proof_Qopt}, and a matching lower bound in Theorem~\ref{thm:Qlb}. 
A simplified upper bound follows from Jensen's inequality.
\begin{coro}
	\label{coro:Qopt}
	For worst case $P$ and $Q$, 
	if $ \ln n \leq C' (\ln S-\varepsilon)$ 
	for some $C'$, then 
	there exist constants $c_1, c_2$ and $c_3$, 
	such that 	for $K=c_3 \ln n $,
	\begin{eqnarray}
		\label{eq:Qopt2}
		\sup_{P,Q \in \cM_S}  \E_Q \big[ \, \big|\, \hd_{\varepsilon,K,c_1,c_2}( P \| Q_n )  - d_\varepsilon(P\| Q) \,\big|^2 \, \big] 
		\;\; \lesssim \;\;  \frac{e^\varepsilon S}{n \ln n} \;.
	\end{eqnarray}
\end{coro}

Note that the plug-in estimator in Corollary~\ref{coro:Qmle} achieves  
the parametric rate of $1/n$. 
In the low-dimensional regime, where we fix $S$ and grow $n$, this cannot be improved upon. 
To go beyond the parametric rate, we need to consider a high-dimensional regime, 
where $S$ grows with $n$.  
Hence, a condition similar to $\ln n \leq C' \ln S$ is necessary, although it might be possible to further relax it.

%
\subsubsection{Matching minimax lower bound}
\label{sec:Qlb}

In the high-dimensional regime, where $S$ grows with $n$ sufficiently fast, 
we can get a tighter lower bound then Theorem~\ref{thm:Qmle}, 
that matches the upper bound in Theorem~\ref{thm:Qopt}. 
Again, supremum over $Q$ is necessary as there exists $(P,Q)$ 
where it is trivial to achieve zero error, for any sample size (see Section~\ref{sec:Qmle} for an example). 
For any given $P$ we provide a minimax lower bound in the following. 

\begin{thm}
	\label{thm:Qlb}
	Suppose  $S\geq 2$ and  there exists a constant $C>0$ such that $\ln n\geq C\ln S$. 
	Then for any $P$, there exists a constant $C'$ that only depends on $C$ such that if 
	$\sum_{j=1}^{S}p_j \wedge \sqrt{{e^{\varepsilon} p_j}/({n \ln n })}\geq C'\left(\sqrt{{(e^{\varepsilon}\ln n)}/{n}}+{(e^\varepsilon\sqrt{S}\ln n)}/{n}\right)$, then
	\begin{eqnarray}
		\inf_{\hd_\varepsilon(P \| Q_n )} \sup_{Q \in \cM_S}  \E_Q \big[ \, \big|\, \hd_{\varepsilon}( P \| Q_n)  - d_\varepsilon(P\| Q) \,\big|^2 \, \big] 
		\;\; \gtrsim \;\;  \Big( \, \sum_{i=1}^S \, p_i \wedge \sqrt{\frac{e^\varepsilon p_i}{n \ln n}} \, \Big)^2\;,
	\end{eqnarray}
	where the infimum is taken over all possible estimator. 
\end{thm}

A proof is provided in Section~\ref{sec:proof_Qlb}. 
When $ \ln n \leq C\ln S $ and $n\geq {c(e^\varepsilon S)}/{\ln S}$, 
the condition 
$\sum_{j=1}^{S}p_j \wedge \sqrt{{e^{\varepsilon} p_j}/({n \ln n })}\geq C'\left(\sqrt{{(e^{\varepsilon}\ln n)}/{n}}+{(e^\varepsilon\sqrt{S}\ln n)}/{n}\right)$ is satisfied 
for uniform distribution $P$, and we can simplify the right-hand side of the lower bound. 
Together, we get the following minimax lower bound, matching the upper bound in 
Corollary~\ref{coro:Qopt}.
\begin{coro}
\label{coro:Qlb}
	If there exist constants $c,C>0$, such that
	$n\geq c{(e^\varepsilon S)}/{\ln S}$ and $\ln n \leq C\ln S$, then
	\begin{eqnarray}
		\sup_{P\in \cM_S}\inf_{\hd_\varepsilon(P \| Q_n )} \sup_{Q \in \cM_S}  \E_Q \big[ \, \big|\, \hd_{\varepsilon}( P \| Q_n )  - d_\varepsilon(P\| Q) \,\big|^2 \, \big] 
		\;\; \gtrsim \;\;
		 \frac{e^\varepsilon S}{n\ln n}\;.
	\end{eqnarray}

\end{coro}

\subsection{Estimating $d_\varepsilon(P\|Q)$  from samples} 
\label{sec:main}

We now consider the general case where $P=P_{\cQ,\cD}$ and $Q=P_{\cQ,\cD'}$ are both unknown, 
and we access them through samples. 
We propose sampling a random number of samples 
$N_1\sim\Poi(n)$ and $N_2\sim\Poi(n)$ from each distribution, respectively. 
Define the empirical distributions $P_n=[\hp_1,\ldots,\hp_S]$ and $Q_n=[\hq_1,\ldots,\hq_S]$ as in the previous section. 
From the proof of Theorem~\ref{thm:Qmle}, we get the same sample complexity for the plug-in estimator: 
If $n\geq e^{\varepsilon}S$ and $S\geq 2$, we have
\begin{eqnarray}
\sup_{P,Q \in \cM_S } \E_{Q} \big[\,|d_\varepsilon(P_n \| Q_n )-d_\varepsilon(P \| Q)|^2\,\big] 
\;\; \asymp \;\;	 \frac{e^{\varepsilon}S   \, }{n}\;. \label{eq:mle}
\end{eqnarray}
Using the same two-step process, we construct an estimator that improves upon this parametric rate of plug-in estimator.

%
\subsubsection{Estimator for $d_\varepsilon(P\| Q)$ }
\label{sec:algo}

We present an estimator using similar techniques as in Algorithm~\ref{Qalgo}, but there are 
several challenges in moving to a multivariate case.   
The multivariate function $f(x,y) = [x-e^\varepsilon y ]^+$ 
in non-smooth in a region $ x=e^\varepsilon y$. 
We first 
define a two-dimensional non-smooth set $U(c_1,c_2)\subset [0,1] \times [0,e^\varepsilon]$ as
\begin{eqnarray}
	U(c_1,c_2) \;\;=\;\; \left\{( p , e^\varepsilon q): | p - e^\varepsilon q\,|\, \leq \sqrt{\frac{ (c_1+c_2)\ln n}{n}}(\sqrt{ p}+\sqrt{e^\varepsilon q}),\; p\in [0,1], \;q\in [0,1]  \right\}\;,
\end{eqnarray}
where $0<c_2<c_1$. 
As before, the plug-in estimator is good enough in the smooth regime, i.e.~$(p,e^\varepsilon q)\notin U(c_1,c_2)$. 

We construct a polynomial approximation of this function with order $K$, in this non-smooth regime.
We will set $K = c_3\ln n$ to achieve the optimal tradeoff. 
We split the samples randomly into four partitions, each having an independent and identical distribution 
of $\Poi(n)$ samples, two from the multinomial distributions $P$ and other two from $Q$.  
See Algorithm~\ref{algo} for a formal definition. We use one set of samples to identify the regime, and the other for estimation.

\begin{algorithm}[h]
\begin{algorithmic}
\REQUIRE{target privacy $\varepsilon\in\reals^+$, query $\cQ$, neighboring databases $(\cD,\cD')$, \\
samples from $P_{\cQ,\cD}$ and $P_{\cQ,\cD'}$,
degree $K\in{\mathbb Z}^+$, constants $c_1,c_2 \in\reals^+$, expected sample size $2n$}
\ENSURE{estimate $\hd_{\varepsilon,K,c_1,c_2}(P_n \| Q_n)$ of $d_\varepsilon( P_{\cQ,\cD}\| P_{\cQ,\cD'})$}
\STATE{$P\leftarrow P_{\cQ,\cD}$, $Q\leftarrow P_{\cQ,\cD'}$}
\STATE{Draw four independent sample sizes: $N_{1,1},N_{1,2},N_{2,1},N_{2,2}\sim \Poi(n)$}
\STATE{Sample from $P_{\cQ,\cD}$: $\{ X_{i,1} \}_{i=1}^{N_{1,1}} \in [S]^{N_{1,1}}$ and $\{ X_{i,2} \}_{i=1}^{N_{1,2}} \in [S]^{N_{1,2}}$}
\STATE{Sample from $P_{\cQ,\cD'}$: $\{Y_{i,1} \}_{i=1}^{N_{2,1}} \in [S]^{N_{2,1}}$ and $\{ Y_{i,2} \}_{i=1}^{N_{2,2}} \in [S]^{N_{2,2}}$}

\STATE{$\hp_{i,j} \leftarrow \frac{|\{\ell\in [N_{1,j}]: X_{\ell,j} = i\}|}{n}$ and $\hq_{i,j} \leftarrow \frac{|\{\ell\in [N_{2,j}]: Y_{\ell,j} = i\}|}{n}$ for all $i\in[S]$ and $j\in\{1,2\}$}

\STATE{$P_{n,1} \leftarrow [\hp_{1,1},\ldots,\hp_{S,1}]$, $P_{n,2} \leftarrow [\hp_{1,2},\ldots,\hp_{S,2}]$, $Q_{n,1} \leftarrow [\hq_{1,1},\ldots,\hq_{S,1}]$ and $Q_{n,2} \leftarrow [\hq_{1,2},\ldots,\hq_{S,2}]$ }

\FOR{$i=1$ \TO $S$}
\STATE{
\hspace{2cm}$\delta_i \gets \left\{ \begin{array}{rl}
0 &,\;\text{ if } \hp_{i,1}- e^{\varepsilon}\hq_{i,1} < -\sqrt{\frac{(c_1+c_2)\ln n}{n}}(\sqrt{\hp_{i,1}}+\sqrt{e^\varepsilon\hq_{i,1}}) \\

\hp_{i,2}-e^\varepsilon \hq_{i,2} &,\;\text{ if }\hp_{i,1}- e^{\varepsilon}\hq_{i,1} > \sqrt{\frac{(c_1+c_2)\ln n}{n}}(\sqrt{\hp_{i,1}}+\sqrt{e^\varepsilon\hq_{i,1}})\\

\tD^{(1)}_K(\hp_{i,2}, \hq_{i,2}) &,\;\text{ if } \hp_{i,1}+e^\varepsilon\hq_{i,1} < \frac{c_1\ln n}{n}\\

\tD^{(2)}_K(\hp_{i,2}, \hq_{i,2};\hp_{i,1}, \hq_{i,1}) &,\;\text{ if } (\hp_{i,1},e^{\varepsilon}\hq_{i,1}) \in U(c_1,c_2),\; \hp_{i,1}+e^\varepsilon\hq_{i,1} \geq \frac{c_1\ln n}{n}

\end{array} \right. \label{eq:opt_estimator}$}
\ENDFOR
\STATE{$\hd_{\varepsilon,K,c_1,c_2}(P_n \| Q_n ) \gets 0 \vee ( 1 \wedge \sum_{i=1}^S \delta_i $})
\end{algorithmic}
\caption{Differential Privacy (DP) estimator}
\label{algo}
\end{algorithm}

\vspace{0.2cm}
\noindent
 {\bf case 1: 
For $(x,  e^{\varepsilon}  y)\in \left[0, ({2c_1\ln n})/{n}\right]^2$.}
 
A straightforward polynomial approximation of 
$[x-e^\varepsilon y]^+$ on $[0,(2c_1 \ln n) / n]^2$ cannot achieve approximation error smaller than 
$(1/K)((2c_1\ln n) / n )$. As $K=c_3 \ln n$, this gives a bias of $1/n$ for each symbol in $[S]$, resulting in 
total bias of $S/n$. This requires $n\gg S$ to achieve arbitrary small error, as opposed to $n\gg S/\ln S$ which is what we are targeting. 
This is due to the fact that we are requiring multivariate approximation, and 
the bias is dominated by the worst case $y$ for each $x$. 
If $y$ is fixed, as in the case of univariate approximation in Lemma~\ref{lem:Qopt_key}, 
the bias would have been $(1/K)\sqrt{(e^\varepsilon y 2c_1 \ln n)/n}$, with $y=q_i$, 
where total bias scales as $\sqrt{S/n \ln n}$ when summed over all symbols $i$. 

Our strategy is to use the decomposition  $[x - e^\varepsilon y ]^+ =  (\sqrt{x}+\sqrt{e^{\varepsilon}y})\left[\sqrt{x}-\sqrt{ e^{\varepsilon}y }\right]^+$. 
Each function can be approximated up to a bias of $ (1/K)\sqrt{(\ln n) / n}$, and the dominant term in the bias becomes 
$(1/K)\sqrt{(e^\varepsilon q_i \ln n) / n}$. This gives the desired bias. 
Concretely, we use two bivariate polynomials $u_K(x,y)$ and $v_K(x,y)$ to  approximate $\sqrt{x}+\sqrt{y}$ and $[\sqrt{x}-\sqrt{y}]^+$ in $[0,1]^2$, respectively. 
Namely,
\begin{eqnarray}
	\sup_{(x,y)\in [0,1]^2}\left|u_K(x,y)- (\sqrt{x}+\sqrt{y})\right| \;\; = \;\; \inf_{P\in {\rm poly}_K^2}\sup_{(x',y')\in [0,1]^2}\left|P(x',y')- (\sqrt{ x'}+\sqrt{y'})\right|\;, \text{ and }
\end{eqnarray}
\begin{eqnarray}
	\sup_{(x,y)\in [0,1]^2}\left|v_K(x,y)- [\sqrt{ x}-\sqrt{y}]^+\right| \;\; = \;\; \inf_{P\in {\rm poly}_K^2}\sup_{(x',y')\in [0,1]^2}\left|P(x',y')- [\sqrt{ x'}-\sqrt{y'}]^+\right|\;.
\end{eqnarray}

Then denote $h_{2K}(x,y) = u_K(x,y)v_K(x,y)-u_K(0,0)v_K(0,0)$. For $(x, e^{\varepsilon} y)\in \left[0, ({2c_1\ln n})/{n}\right]^2$, define 
\begin{eqnarray}
	D_K^{(1)}(x,y) \;\;=\;\; \frac{2c_1 \ln n}{n} h_{2K}\left(\frac{ x n }{2c_1\ln n}, \frac{ e^{\varepsilon} y n}{2c_1\ln n}\right)\;.\label{eq:Opt1_dk}
\end{eqnarray}
In practice, one can use the best Chebyshev polynomial expansion to achieve the same uniform error rate, efficiently \cite{MNS13}.

\vspace{0.2cm}
\noindent
 {\bf case 2: For $( x , e^{\varepsilon} y)\in U(c_1,c_1)$ and $ x + e^{\varepsilon} y\geq ({c_1\ln n})/{2n}$.} 

We utilize the best polynomial approximation of $|t|$ on $[-1,1]$ with order $K$. Denote it as $R_K(t) = \sum_{j=0}^K r_j t^j$. Define
\begin{eqnarray}
	D^{(2)}_K\left(x,y; \hp_{i,1}, \hq_{i,1}\right) \;\;=\;\;\frac{1}{2}\sum_{j=0}^Kr_j W^{-j+1}(e^{\varepsilon}y-x)^j+\frac{x-e^\varepsilon y }{2}\;,
\end{eqnarray}
where $W = \sqrt{ ({8c_1\ln n})/{n} }\left(\sqrt{ \hp_{i,1}+ e^{\varepsilon} \hq_{i,1}}\right)$. 
Finally, we use second part of samples to construct unbiased estimator for $D_K^{(1)}(x,y)$ and 
$D_K^{(2)}(x,y; \hp_{i,1},\hq_{i,1})$ by Lemma~\ref{lem:unbiased_estimate_pq}. Namely, 
\begin{eqnarray}
	\E\left[\tD_K^{(1)}(\hp_{i,2},\hq_{i,2})\right] \;\;=\;\; D_K^{(1)}(p,q)\;,
\end{eqnarray}
and
\begin{eqnarray}
	\E\left[ \tD_K^{(2)}(\hp_{i,2},\hq_{i,2}; \hp_{i,1},\hq_{i,1})|\hp_{i,1},\hq_{i,1}\right] \;\;=\;\; D_K^{(2)}(p,q; \hp_{i,1},\hq_{i,1})\;.
\end{eqnarray}

These unbiased estimator are easy to construct by using following lemma.

\begin{lemma}\cite[Exercise~2.8]{withers1987bias}
\label{lem:unbiased_estimate_pq}
	For any $r,s\geq 1$, $r,s \in \Z$, $(n\hp,n\hq)~\Poi(np)\times\Poi(nq)$, we have
	\begin{eqnarray}
		\mathbb{E}\left[\prod_{i=0}^{r-1}\left(\hat{p}-\frac{i}{n}\right) \prod_{j=0}^{s-1}\left(\hat{q}-\frac{j}{n}\right)\right]\;\;=\;\;p^{r} q^{s}\;.
	\end{eqnarray}	
\end{lemma}
The explicit formula for the unbiased estimators 
can be found in Eq.~\eqref{eq:opt1_estimator} and Eq.~\eqref{eq:opt2_estimator}. 

%
%
\subsubsection{Minimax optimal upper bound}
\label{sec:test}

We provide an upper bound on the error achieved by the proposed estimator. 
The analysis uses similar techniques as the proof of Theorem~\ref{thm:Qopt}. 
However, it is not an immediate corollary, and we provide a proof in Section~\ref{sec:proof_opt}. 

\begin{thm} 
	\label{thm:opt}
	Suppose  there exists a constant $C>0$ such that $\ln n \leq C \ln S $. Then there exist constants $c_1, c_2$ and $c_3$ that only depends on $C$  and $\varepsilon$ such that 
	\begin{eqnarray}
		\label{eq:opt}
		\sup_{P,Q \in \cM_S}  \E_{P\times Q} \big[ \, \big|\, \hd_{\varepsilon,K,c_1,c_2}( P_n \| Q_n)  - d_\varepsilon(P\| Q) \,\big|^2 \, \big] 
		\;\; \lesssim \;\;   \frac{e^\varepsilon S}{n\ln n} \;,
	\end{eqnarray}
	for $K=c_3 \ln n $ where  $\hd_{\varepsilon,K,c_1,c_2}$ is defined in Algorithm~\ref{algo}.
\end{thm} 

It follows from the proof of Theorem~\ref{thm:Qlb} that 
\begin{eqnarray*}
		\inf_{\hd_\varepsilon(P_n \| Q_n )} \sup_{P,Q \in \cM_S}  \E_{P\times Q} \big[ \, \big|\, \hd_{\varepsilon}( P_n \| Q_n)  - d_\varepsilon(P\| Q) \,\big|^2 \, \big] 
		\;\; \gtrsim \;\;  \frac{e^\varepsilon S}{n \ln n}\;.
\end{eqnarray*}
 Together, the above upper and lower bounds 
 prove that the proposed estimator in Algorithm~\ref{algo} is minimax optimal and cannot be improved upon in terms of sample complexity. 
 We want to emphasize that we do not require to know the size of the support $S$, 
 as opposed to exiting methods in \cite{DWW18}, which 
 requires collecting enough samples to identify the support. 
Comparing it to the  error rate of plug-in estimator in Theorem~\ref{thm:Qmle}, 
this minimax rate of $e^\varepsilon S /( n \ln n)$ demonstrates the  
 {\em effective sample size amplification} holds;
with $n$ samples,  
a sophisticated estimator can achieve the error rate 
equivalent to a plug-in estimator with $n\ln n $ samples.


%
%

%
%
\section{Experiments}
\label{sec:exp}
We implemented both plug-in estimator and Algorithm~\ref{algo}. Note that the coefficients of bivariate polynomial $u_K(x,y)$, $v_K(x,y)$ on $[0,1]^2$ and $R_K(t)$ on $[-1,1]$ are independent of data and $\varepsilon$. 
 We pre-compute the coefficients and look them up from a table when running our estimators. Using numerical computation provided by Chebfun toolbox \cite{Driscoll2014}, we obtained the coefficients of $R_K$ by Remez algorithm, and the coefficients of bivariate polynomial $u_K$, $v_K$ by lowpass filtered Chebyshev expansion \cite{MNS13}. The time complexity of Algorithm~\ref{algo} is $O(n\ln^2 n)$.
We compare the proposed Algorithm~\ref{algo} and the plug-in estimator on synthetic data (Section~\ref{sec:syn}), 
and demonstrate the $(\varepsilon, \delta)$ regions for some of the popular differential privacy mechanisms and their variations 
(Section~\ref{sec:real}).
All experiments are done on a Macbook Pro  with Intel\textsuperscript{\textregistered} Core\textsuperscript{TM} i5 processor and $8$ GB memory. Our estimator is implemented in Python 3.6. We provide the code as a supplementary material.

\subsection{Synthetic experiments}
\label{sec:syn}

Although in the analysis we make conservative choices of the constants $c_1$, $c_2$ and $c_3$, 
Algorithm~\ref{algo} is not sensitive to those choices in practice and 
we fix them to be $c_1 = 4$, $c_2 = 0.1$ and $c_3 = 1.5$ in the experiments. 
The degree of polynomial approximation is chosen as $K = \lfloor c_3\ln(n)\rfloor$.
Figure~\ref{fig:synthetic} illustrates the 
Mean Square Error (MSE) for estimating $d_\varepsilon(P\|Q)$ between uniform distribution $P$ and Zipf distribution $Q$, where the support size is fixed to be $S = 100$, ${\rm Zipf}(\alpha)\propto 1/i^\alpha$, and $\alpha = -0.6$ for $i\in [S]$. The $\varepsilon$ is fixed to be $\varepsilon=0.4$. Each data point represents $100$ random trials, with standard error (SE) error bars smaller than the plot marker. 
This suggests that the Algorithm 2 consistently improves upon the plug-in estimator, as predicted by Theorem~\ref{thm:opt}.

\begin{figure}[h]
	\centering
	\includegraphics[width=0.4\linewidth]{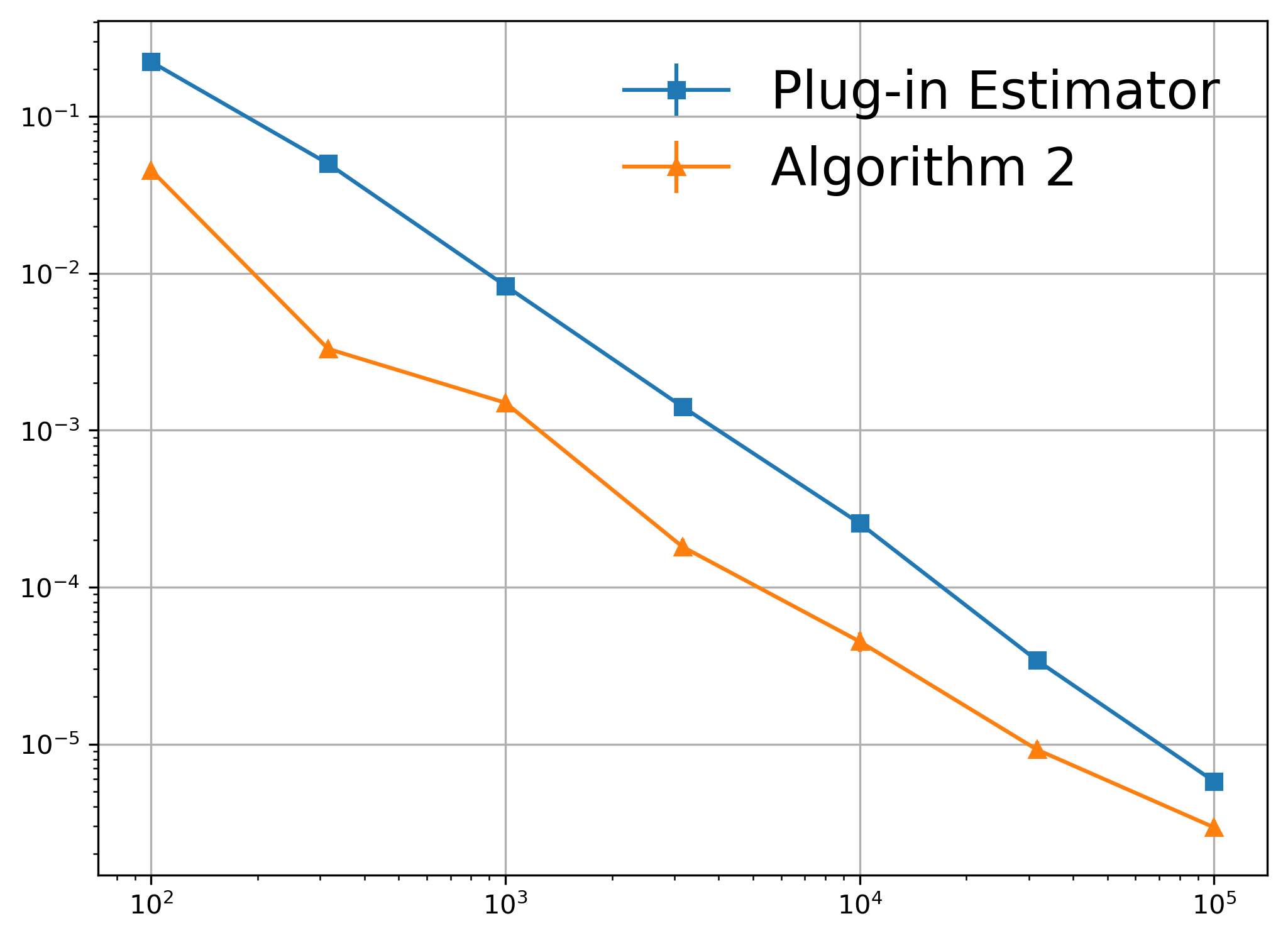}
	\put(-215,60){MSE}
	\put(-110,-10){sample size $n$}
	\caption{The proposed minimax optimal estimator in Algorithm~\ref{algo} consistently improves upon the plug-in estimator
	 on synthetic data.}
	\label{fig:synthetic}
\end{figure}

\subsection{Detecting violation of differential privacy with Algorithm~\ref{algo}}
\label{sec:real}

We demonstrate how we can use Algorithm~\ref{algo} 
to detect mechanisms with false claim of DP guarantees on four types of mechanisms: 
Report Noisy Max \cite{dwork2014algorithmic}, Histogram \cite{Dwork2006},  
Sparse Vector Technique \cite{LSL17} and Mixture of Truncated Geometric Mechanism. 
Following the experimental set-up of \cite{DWW18},  the test query and databases 
defining $(\cQ,\cD,\cD')$   are chosen by some heuristics, shown in Table~\ref{tab:databases}.  
 However, unlike the approach from \cite{DWW18}, 
  we do not require to know the size of the support $S$, 
  we don't have to specify candidate bad events $E\subseteq [S]$, 
  and we can estimate general approximate DP with $\delta>0$.
Throughout all the experiments, we fix $c_1 = 4$, $c_2 = 0.1$, $c_3 = 0.9$, and the mean of number of samples $n = 100000$. 
 For the examples in 
 Report Noisy Max, Histogram, and Sparse Vector Technique, we compose 5 and 10  queries together to form a one giant query. 
There are several categories  of queries and databases to be tested, 
each represented by the true answer of the 5 queries in the table. 
 When testing, 
 we test all categories, and report the largest estimate $\hat{\delta}$ for each given $\varepsilon$. For those faulty mechanisms, Algorithm~\ref{algo} also provides a certificate in the form of a set $T\subseteq[S]$ such that $[P(T)-e^\varepsilon_0 Q(T)]^+ -\delta_0 > 0$.

\begin{table}[h]
\centering
\caption{Database categories and samples \cite{DWW18}}
\label{tab:databases}
\begin{tabular}{@{}lll@{}}
\toprule
Category              & $[\cQ_1(\cD),\ldots,\cQ_5(\cD)]$      & $[\cQ_1(\cD'),\ldots,\cQ_5(\cD')]$     \\ \midrule
One Above             & [1, 1, 1, 1, 1] & [2, 1, 1, 1, 1] \\
One Below             & [1, 1, 1, 1, 1] & [0, 1, 1, 1, 1] \\
One Above Rest Below  & [1, 1, 1, 1, 1] & [2, 0, 0, 0, 0] \\
One Below Rest Above  & [1, 1, 1, 1, 1] & [0, 2, 2, 2, 2] \\
Half Half             & [1, 1, 1, 1, 1] & [0, 0, 0, 2, 2] \\
All Above \& All Below & [1, 1, 1, 1, 1] & [2, 2, 2, 2, 2] \\
X Shape               & [1, 1, 1, 1, 1] & [0, 0, 1, 1, 1] \\ \bottomrule
\end{tabular}
\end{table}

\noindent
{\bf Report Noisy Max.} For privacy budget $\varepsilon_0$, {\em Report Noisy Argmax with Laplace noise (RNA+Lap)} adds independent $Lap(2/\varepsilon_0)$ noise to query answers $\cQ(\cD)$ and return the index of the largest noisy query answer. {\em Report Noisy Argmax with Exponential noise (RNA+Exp) } adds $Exp(2/\varepsilon_0)$ noise instead of Laplace noise. 
Their claimed level of $(\varepsilon_0)$-DP is correctly guaranteed  \cite[Claim~3.9 and Theorem~3.10]{dwork2014algorithmic}.  
On the other hand, {\em Report Noisy Max with Laplace noise (RNM+Lap)} or {\em Report Noisy Max with Exponential noise (RNM+Exp)} return the largest noisy answer itself, instead of its index. 
This reveals more information than intended, leading to violation of claimed $(\varepsilon_0,0)$-DP.
Figure~\ref{fig:noisy_max} shows $(\varepsilon,\delta)$ regions for the above variations of noisy max mechanisms. 
As expected, for each $\varepsilon_0$, 
 {\em RNA+Lap} and {\em RNA+Exp} satisfy $(\varepsilon_0,0)$-DP, 
 whereas  {\em RNM+Lap} and {\em RNM+Exp} have $\hat\delta>0$. 


\begin{figure}[h]
	\centering
	\begin{tabular}{ccccc}
\includegraphics[width=.35\textwidth]{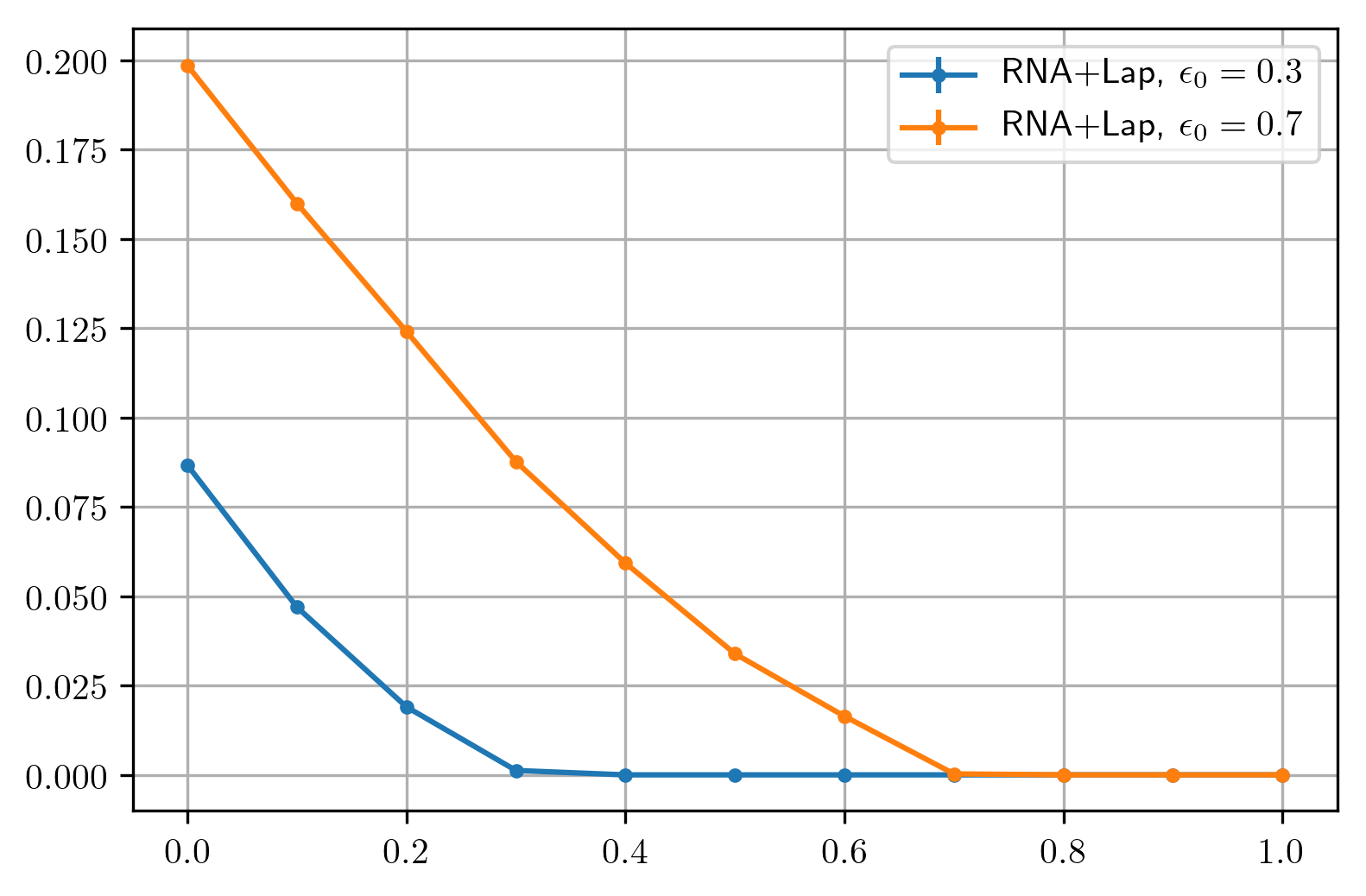}&
\put(-80,-5){$\varepsilon$}
\put(-180,50){$\hat{\delta}$}
\hspace{0.1cm}
\includegraphics[width=.35\textwidth]{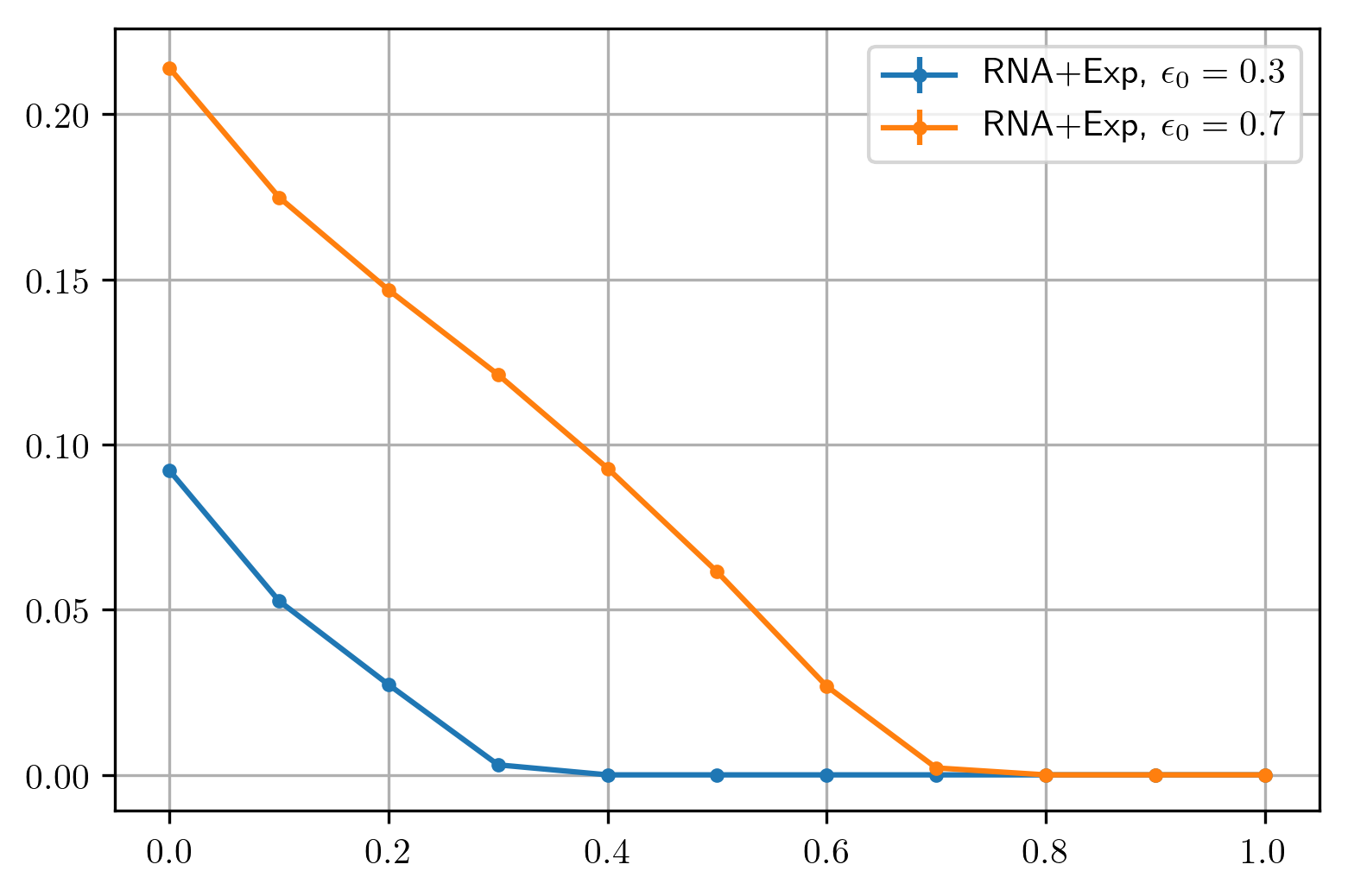}
\put(-80,-5){$\varepsilon$}
\put(-173,50){$\hat{\delta}$}
\\

\includegraphics[width=.35\textwidth]{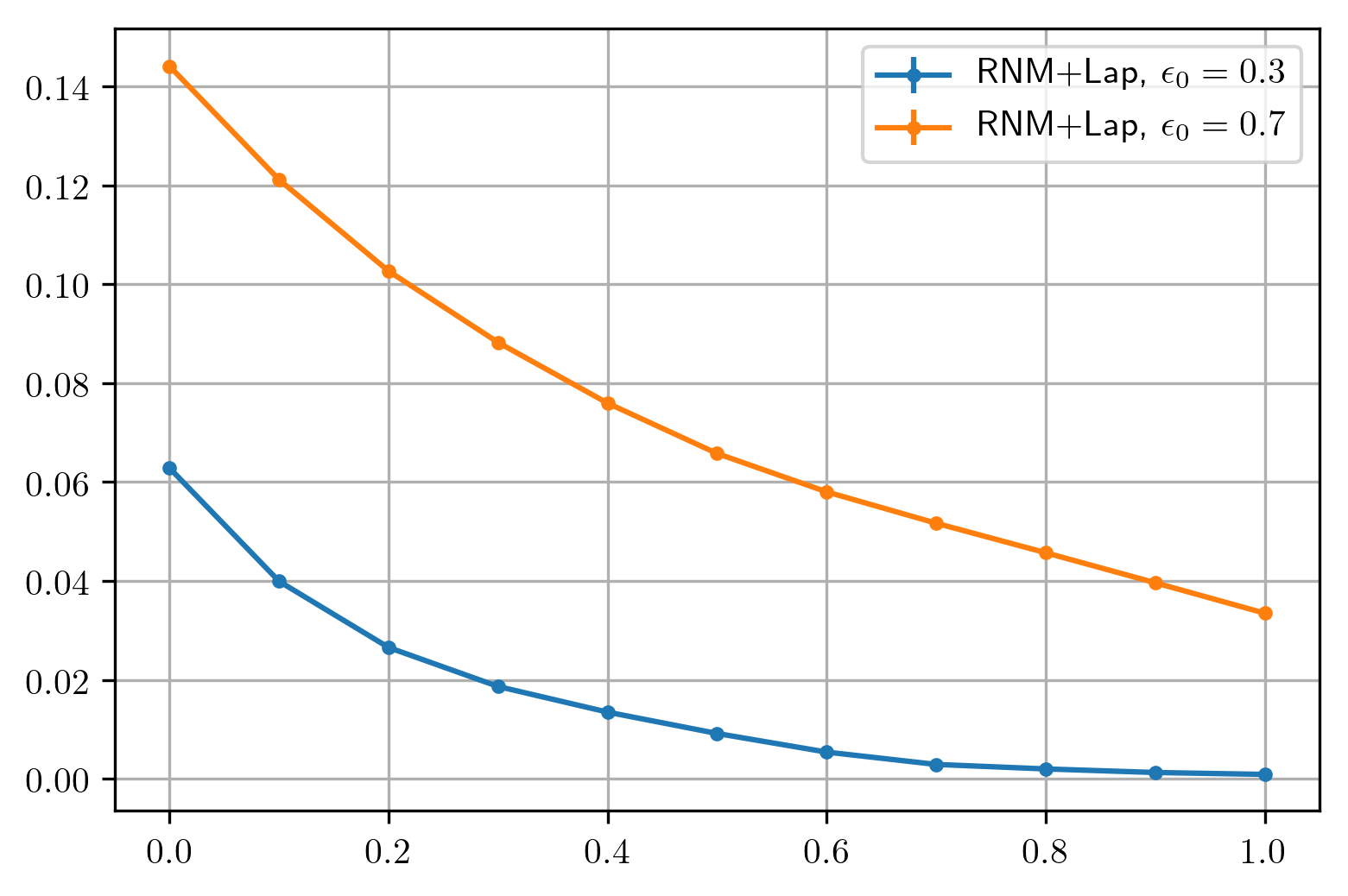}&
\put(-80,-5){$\varepsilon$}
\put(-180,50){$\hat{\delta}$}
\hspace{0.1cm}

\includegraphics[width=.35\textwidth]{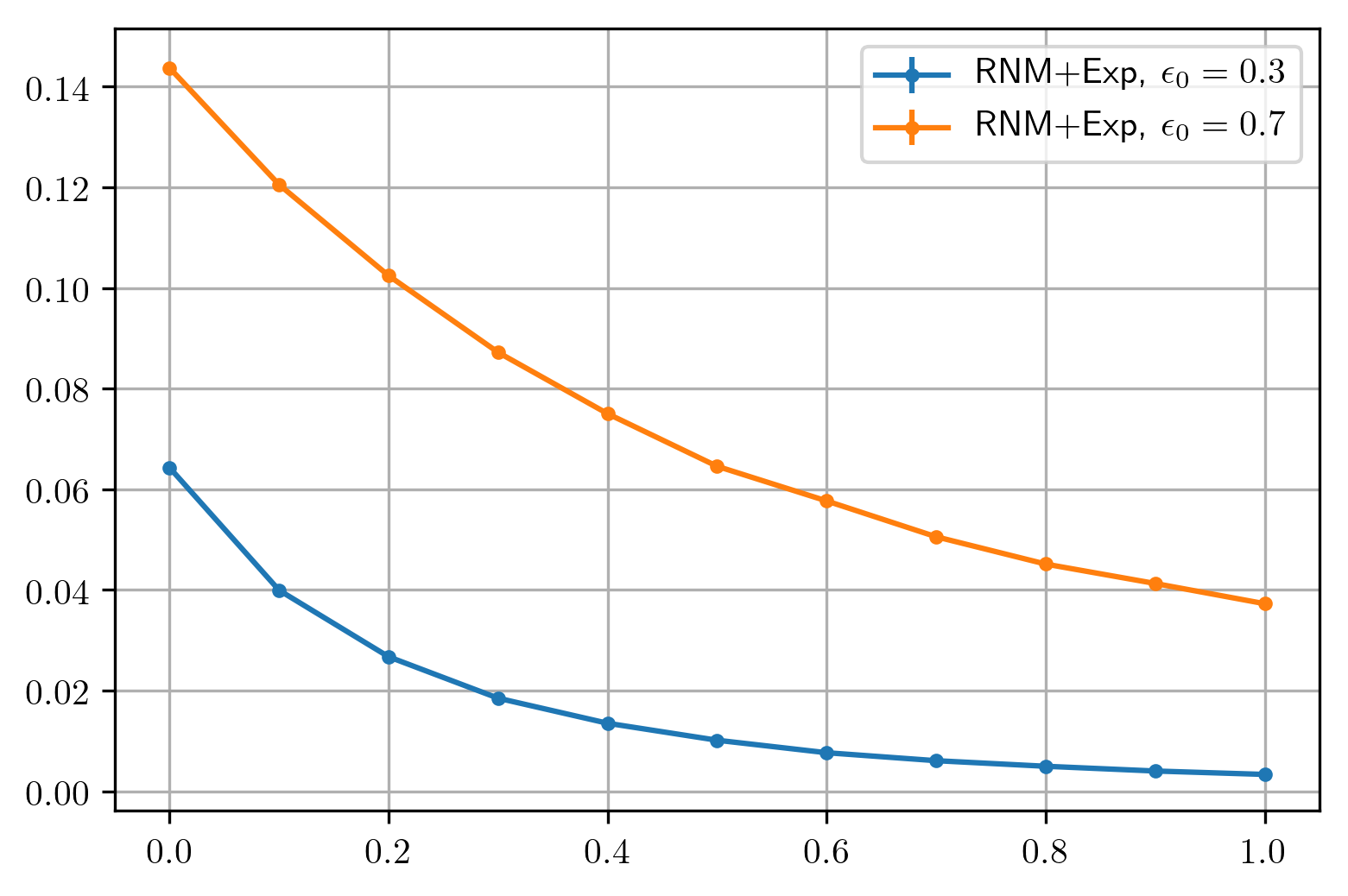}
\put(-80,-5){$\varepsilon$}
\put(-173,50){$\hat{\delta}$}

\end{tabular}

	\caption{Estimate $\hat{\delta}$ from Algorithm~\ref{algo} of $\delta$ given $\varepsilon$, and privacy budget $\varepsilon_0$ for noisy max mechanisms. Each point is showing an average over  $10$ random trials with standard error.}
	\label{fig:noisy_max}
\end{figure}
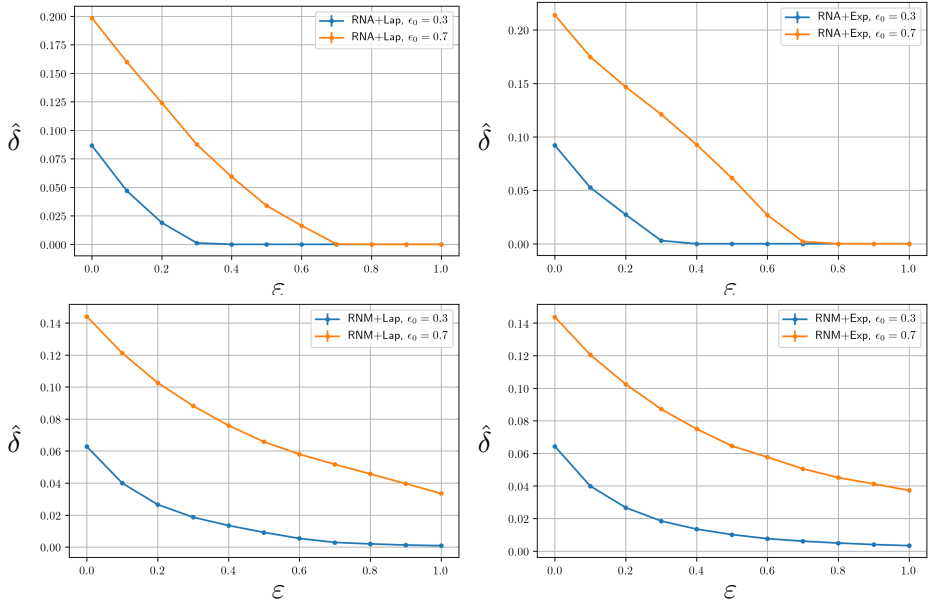

\bigskip
\noindent{\bf Histogram.}
For privacy budget $\varepsilon_0$, {\em Histogram} takes histogram queries as input, adds independent $Lap(1/\varepsilon_0)$ noise to each query answers, and output the randomized query answers directly, which is proved to be $(\varepsilon_0,0)$-DP \cite{dwork2014algorithmic}. As a comparison, {\em Histogram with incorrect noise} adds incorrect noise $Lap(\varepsilon_0)$. Note that as we use histogram queries, we require $\cQ(\cD)$ and $\cQ(\cD')$ to be different in at most one element, which is tested on One Above and One Below samples as shown in Table~\ref{tab:databases}.
With the setting of privacy budget $\varepsilon_0 = 0.5$, Figure~\ref{fig:histogram} left panel shows that the incorrect histogram is likely to be $(1/\varepsilon_0,0)$-DP. Both mechanisms claim $(0.5,0)$-DP, but the figure shows that the incorrect mechanism ensures $(1/0.5,0)$-DP instead. 

\begin{figure}[h]
	\centering
	\includegraphics[width=0.31\linewidth]{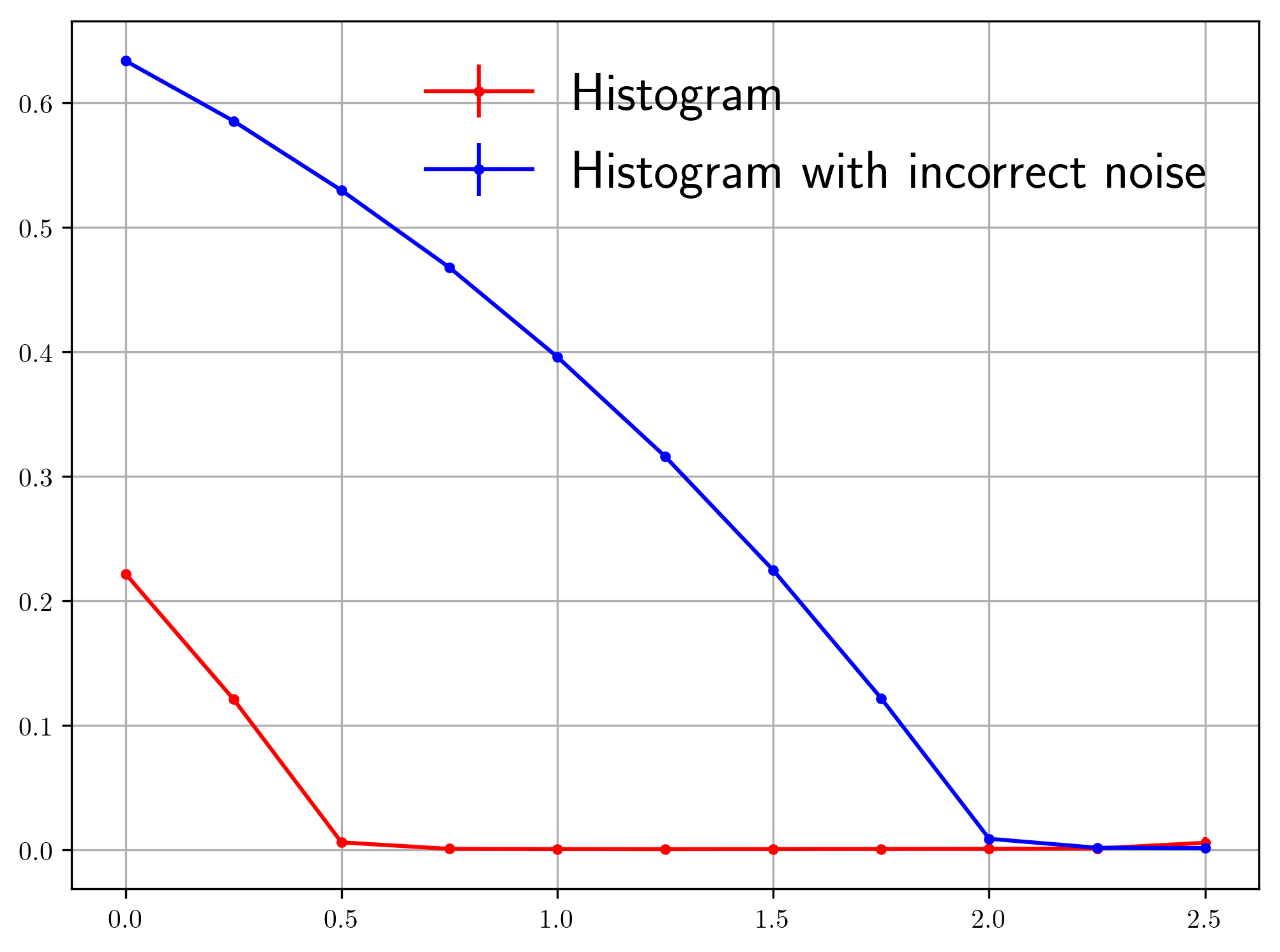}
	\put(-70,-5){$\varepsilon$}
	\put(-155,50){$\hat{\delta}$}
	\hspace{0.3cm}
	\includegraphics[width=0.31\linewidth]{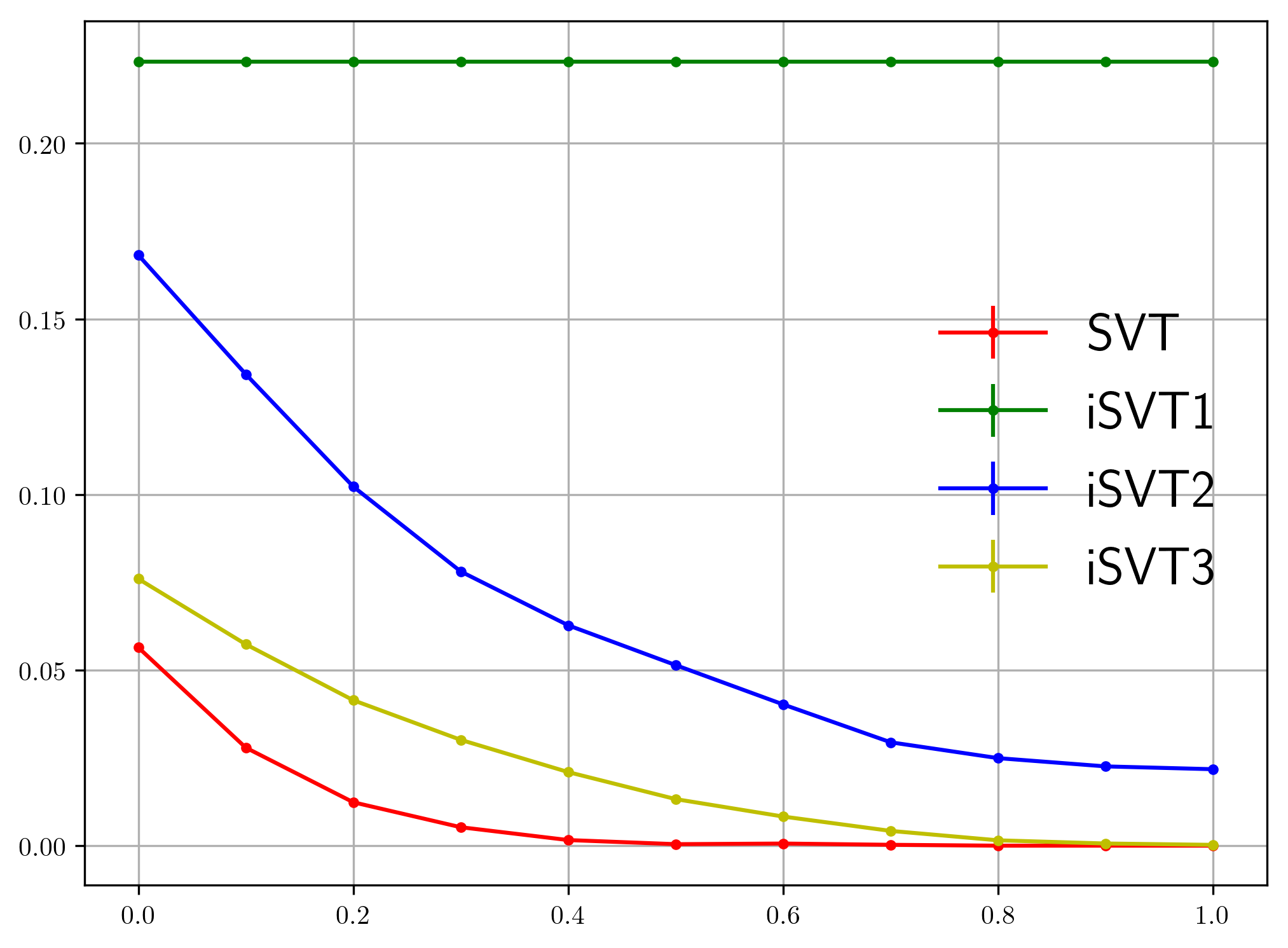}
	\put(-70,-5){$\varepsilon$}
	\put(-155,50){$\hat{\delta}$}
	\hspace{0.3cm}
	\includegraphics[width=0.31 \linewidth]{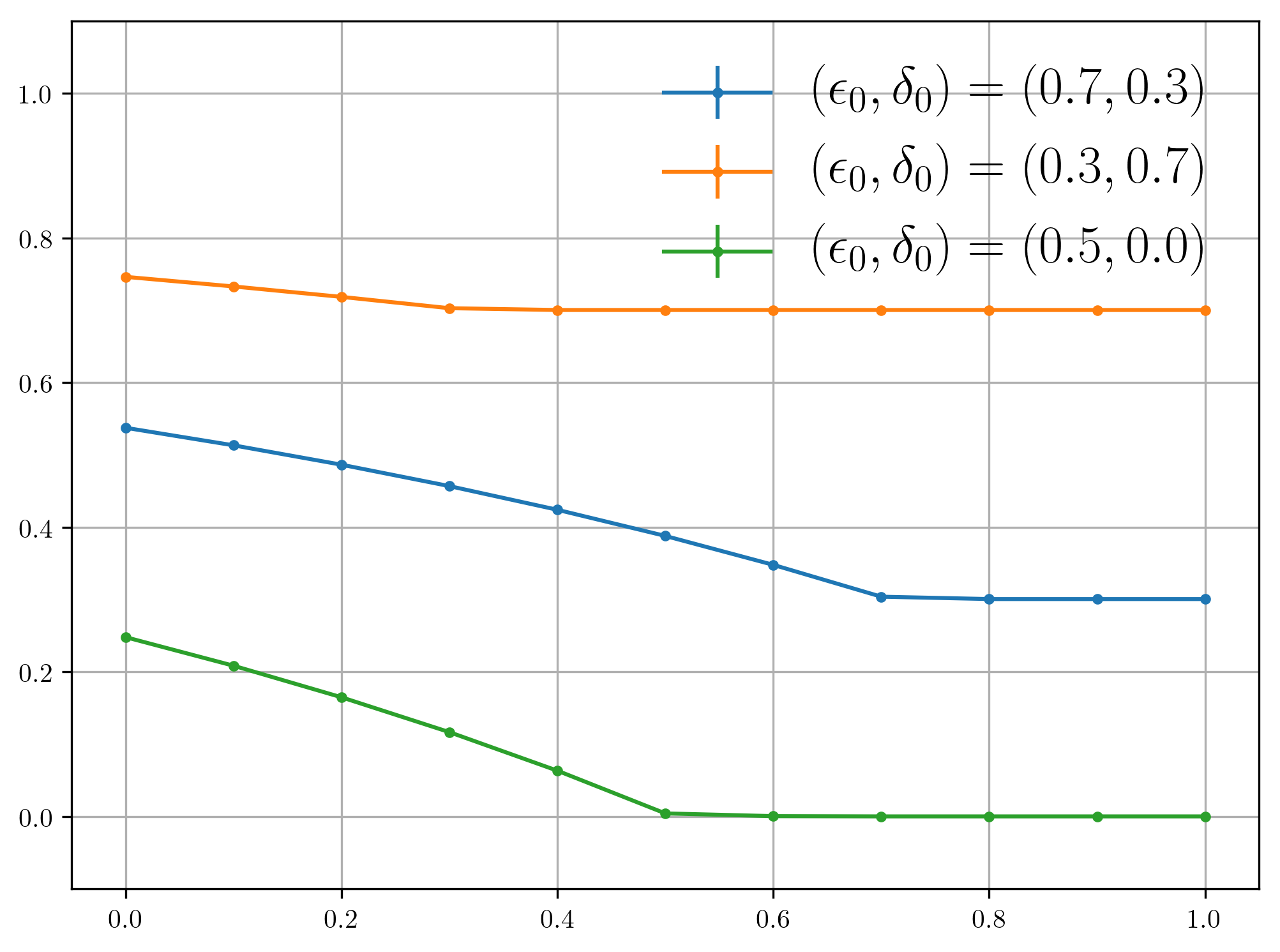}
	\put(-70,-5){$\varepsilon$}
	\put(-155,50){$\hat{\delta}$}
	\caption{Estimates of $\delta$ given $\varepsilon$, and privacy budget $\varepsilon_0$ for histogram (left),  
	 SVT mechanisms (middle),   Mixture of Truncated Geometric Mechanisms (right). 
	Each point represents $10$ random trials with standard error bars.
	}
	\label{fig:histogram}
\end{figure}

\bigskip\noindent
{\bf Sparse Vector Technique (SVT).} 
We consider original sparse vector technique mechanism {\em SVT} \cite{LSL17}, and its variations {\em iSVT1} \cite{SCM14}, {\em iSVT2}  \cite{CXZ15}, and {\em iSVT3} \cite{LC14}. They are discussed in Section~\ref{sec:intro} and also studied and tested in \cite{DWW18}.
With the setting of privacy budget $\varepsilon_0 = 0.5$, Figure~\ref{fig:histogram} shows that {\em SVT} is likely to be $(\varepsilon_0, 0)$-DP. However, {\em iSVT1} and {\em iSVT2} are not likely to be pure differentially private for $\varepsilon \in [0,1]$ with budget $\varepsilon_0=0.5$. As discussed in \cite{LSL17}, {\em iSVT3} is in fact $(\frac{(1+6N)}{4}\varepsilon_0,0)$-DP, where $N$ is the bound of number of trues in the output Boolean vector and set as $N=1$ in this experiment. Figure~\ref{fig:histogram} middle panel shows that with $\hat{\delta} =0$, $\varepsilon$ is likely to be in the range $[0.8, 0.9]$, which verifies the theoretic guarantee.

\bigskip\noindent
{\bf Mixture of Truncated Geometric Mechanism (MTGM)}
With privacy budget $\varepsilon_0$, {\em Truncated Geometric Mechanism (TGM)} proposed by \cite{GRS12} is provably to be $(\varepsilon_0, 0)$-DP. With probability privacy budget $\varepsilon_0$ and $\delta_0\in [0,1]$, {\em Mixture of Truncated Geometric Mechanism (MTGM)} outputs the original query answer with probability $\delta_0$, and outputs the randomized query answer with probability $1-\delta_0$. {\em MTGM} can be proved to be $(\varepsilon_0,\delta_0)$-DP by composition theorem. Note that {\em TGM} and {\em MTGM} both take single counting query as query function $\cQ$. In the experiment, we consider the single counting query with range $\{0,1,2,3\}$.
Figure~\ref{fig:histogram} right panel confirms that {\em MTGM} satisfied the claimed $(\varepsilon_0,\delta_0)$ differential privacy.

%
%
\newpage
\section{Proof}
\label{sec:proof}

\subsection{Auxiliary lemmas}
\label{sec:proof_aux}

\subsubsection{Lemmas on Poisson distribution}

\begin{lemma}[{\cite[Exercise~4.7]{mitzenmacher2005probability}}]
\label{lem:poisson_tail}
If $X\sim \Poi(\lambda)$, then for any $\delta>0$, we have
\begin{eqnarray}
	\prob\Big(X\geq (1+\delta)\lambda \Big)&\;\;\leq\;\;& \Big(\frac{e^{\delta}}{(1+\delta)^{1+\delta}}\Big)^\lambda \;\;\leq\;\; e^{-\delta^2\lambda/3}\;\vee\; e^{-\delta\lambda/3}\;,\\
	\prob\Big(X\leq (1-\delta)\lambda \Big)&\;\;\leq\;\;& \Big(\frac{e^{-\delta}}{(1-\delta)^{1-\delta}}\Big)^\lambda \;\;\leq\;\; e^{-\delta^2\lambda/2}\;.
\end{eqnarray}
\end{lemma}

\begin{lemma}[{\cite[Exercise~4.14]{mitzenmacher2005probability}}]
\label{lem:weighted_poisson_tail}
	Suppose $X\sim \Poi(\lambda_1)$, $Y\sim\Poi(\lambda_2)$, and $Z = \alpha X+ Y$, where $\alpha>1$ is a constant. Then $\E[Z] = \alpha\lambda_1+\lambda_2$,
	and for any $\delta>0$, we have
	\begin{eqnarray}
		\prob\Big(Z\geq (1+\delta)(\alpha\lambda_1+\lambda_2)\Big) &\;\;\leq\;\;& e^{-\delta^2(\alpha\lambda_1+\lambda_2)/3}\vee e^{-\delta(\alpha\lambda_1+\lambda_2)/3}\;,\\
		\prob\Big(Z\leq (1-\delta)(\alpha\lambda_1+\lambda_2)\Big) &\;\;\leq\;\;& e^{-\delta^2(\alpha\lambda_1+\lambda_2)/2}\;.
		\end{eqnarray}
\end{lemma}

\begin{lemma}
	\label{lem:poiss_relu}
	Suppose $n\hq\sim \Poi(nq)$, then 
	\begin{eqnarray}
		\E\big[ \,[\hq - q]^+\,\big]\; \in \;\left\{ 
		\begin{array}{rl}
		q \, e^{-nq}							& \;, 0 \leq q \leq \frac{1}{n}\\
		\Big[\, \sqrt{\frac{q}{4n}} , \sqrt{\frac{q}{2n}} \, \Big] 	& \;, q\geq \frac1n
		\end{array}
		\right.
	\end{eqnarray}
	Hence, 
	\begin{eqnarray}
		\frac{1}{2} \Big( q \wedge \sqrt{\frac{q}{n}}\Big) \;\leq \; \E[[\hq-q]^+]\;\leq \; 
		\Big( q \wedge \sqrt{\frac{q}{n}}\Big)\;.
	\end{eqnarray}
\end{lemma}
\begin{proof}
	Let $\lambda=nq$, then 
	\begin{eqnarray*}
		\E\big[ \,[\hq - q]^+\,\big] 
		&=& \frac{1}{n} \sum_{k= \lfloor \lambda \rfloor + 1}^\infty \frac{ \lambda^k e^{-\lambda} }{k!} (k-\lambda) \\
		&=& \frac{1}{n} \sum_{k= \lfloor \lambda \rfloor + 1}^\infty \frac{ \lambda^k e^{-\lambda} }{k!} k \;-\; \frac{1}{n} \sum_{k= \lfloor \lambda \rfloor + 1}^\infty \frac{ \lambda^k e^{-\lambda} }{k!}\lambda \\
		&=& \frac{1}{n} \lambda \sum_{k= \lfloor \lambda \rfloor}^\infty \frac{ \lambda^k e^{-\lambda} }{k!}  \;-\; \frac{1}{n} \lambda \sum_{k= \lfloor \lambda \rfloor + 1}^\infty \frac{ \lambda^k e^{-\lambda} }{k!} \\
		&=& \frac{ \lambda^{\lfloor\lambda\rfloor+1} e^{-\lambda} }{n\, \lfloor\lambda\rfloor!}
	\end{eqnarray*}
	For $\lambda=nq \leq1$, this is $qe^{-nq} $.
	For $\lambda \geq 1$, we use Stirling's approximation to get 
	\begin{eqnarray*}
		\frac{ \lambda^{\lfloor\lambda\rfloor+1} e^{-\lambda} }{n\, \lfloor\lambda\rfloor!} 
		\;\; \in \;\; \Big[\frac1e,\frac{1}{\sqrt{2\pi}}\Big] \times   \frac{ \lambda^{\lfloor\lambda\rfloor+1} \,e^{-\lambda} }
		{n\, \lfloor\lambda\rfloor^{\lfloor\lambda\rfloor+\frac12}\,e^{-\lfloor\lambda\rfloor}}\;.
	\end{eqnarray*} 
	As $ \frac{ \lambda^{\lfloor\lambda\rfloor+\frac12} \,e^{-\lambda} }
		{ \lfloor\lambda\rfloor^{\lfloor\lambda\rfloor+\frac12}\,e^{-\lfloor\lambda\rfloor}}$ is 
		in $[1\,,\,1.12]$ for $\lambda\geq 1$, this gives the desired bound. 
\end{proof}

\begin{lemma}
	\label{lem:poiss_shiftedrelu}
	Suppose $n\hq\sim \Poi(nq)$, then 
	\begin{eqnarray}
		\E\big[ \,[p-e^\varepsilon \hq ]^+\,-\, [p-e^\varepsilon q]^+ \big]\; \; \leq \;\; e^\varepsilon \, \min\, \Big\{\, q , e^{-\varepsilon}p ,\sqrt{\frac{q}{n}}  , \sqrt{\frac{e^{-\varepsilon}p}{n}} \, \Big\} \;.  
	\end{eqnarray}
\end{lemma}

\begin{proof}
	
	\begin{eqnarray*}
		\E\big[ \,[ p- e^{\varepsilon} \hq]^+ - [p- e^{\varepsilon} q]^+ \,\big] 
		&=& \E\big[ \,\frac{(p - e^{\varepsilon}\hq)+|p - e^{\varepsilon}\hq| }{2}- \frac{(p-e^{\varepsilon}q)+|p-e^{\varepsilon}q|}{2} \,\big]  \\
		&=& \frac12 \, \Big( \E \big[\, |p-e^{\varepsilon}\hq|\,\big] - | p - e^{\varepsilon}q | \Big)\; .
	\end{eqnarray*}

	If $ p\geq e^{\varepsilon} q$, then 
	\begin{eqnarray*}
		\frac12 \, \Big( \E \big[\, |p-e^{\varepsilon}\hq|\,\big] - | p - e^{\varepsilon}q | \Big)
		&=& \frac12\,\Big( \E \big[\, |p-e^{\varepsilon}\hq| - ( p-e^\varepsilon q ) \,\big]\Big)\\
		&=& \frac12\, \Big(\E \big[\, (p-e^{\varepsilon}\hq) - ( p-e^{\varepsilon}q)+ 2 [e^\varepsilon\hq-p]^+ \,\big] \Big)\\
		&=& \E \big[\, [e^\varepsilon\hq-p]^+ \,\big]\\
		&\leq & e^\varepsilon\E \big[\, [\hq - q]^+ \,\big]\\
		&\leq &e^\varepsilon \Big( q\wedge \sqrt{\frac{q}{n}}\Big)\;,
	\end{eqnarray*}
	where the second equality is because of the fact that $x = [x]^+-[-x]^+$ and $|x| = [x]^++[-x]^+$, the first inequality follows from the fact that $[x]^+$ is monotone, and the last inequality follows from Lemma \ref{lem:poiss_relu}.

	If $ p< e^\varepsilon q$, then 
	\begin{eqnarray*}
		\frac12 \, \Big( \E \big[\, |p-e^{\varepsilon}\hq|\,\big] - | p - e^{\varepsilon}q | \Big)
		&=& \frac12\,\Big( \E \big[\, |p-e^{\varepsilon}\hq| - ( e^\varepsilon q - p) \,\big]\Big)\\
		&=& \frac12\, \Big(\E \big[\, (e^{\varepsilon}\hq-p) - ( e^{\varepsilon}q-p)+ 2 [p-e^\varepsilon\hq]^+ \,\big] \Big)\\
		&=& \E \big[\, [p-e^{\varepsilon}\hq]^+ \,\big] \;,
	\end{eqnarray*}

	Now we construct new random variable $\hp$ by $n\hq = ne^{-\varepsilon} \hp+Z$, where $Z$ is independent of $\hp$ and $Z\sim \Poi(n(q-e^{-\varepsilon} p))$. Hence, $e^{-\varepsilon}\hp\leq \hq$ with probability one. And the marginal distribution satisfies $ne^{-\varepsilon}\hp\sim \Poi(ne^{-\varepsilon }p)$. We have

	\begin{eqnarray*}
		\E \big[\, [e^{-\varepsilon}p-\hq]^+ \,\big] 
		&\leq & \E \big[\, [e^{-\varepsilon}p-e^{-\varepsilon}\hp]^+ \,\big]\\
		&\leq &  \Big(e^{-\varepsilon} p \wedge \sqrt{\frac{e^{-\varepsilon}p}{n}}\Big)\;,
	\end{eqnarray*}
	where the first inequality follows from that fact that  $[x]^+$ is monotone, and the last inequality follows from Lemma \ref{lem:poiss_relu}.
\end{proof}


\begin{lemma}
	\label{lem:poiss_varub}
	Suppose $n\hq\sim\Poi(nq)$, then for any $p$ and $q$, we have
		\begin{eqnarray}
		\var\big(\, [p-e^\varepsilon\, \hq ]^+ \,\big) & \;\;\lesssim \;\;& \frac{e^{\varepsilon}p}{n}\;.
	\end{eqnarray}
\end{lemma}

\begin{proof}

\bigskip\noindent{\bf Case 1: If $nq< 1$ and $e^{-\varepsilon}p\geq q$ or if $nq\geq 1$ and $ne^{-\varepsilon}p > \lfloor n q \rfloor-1 $: }

\begin{eqnarray*}
	\var\big(\, [p - e^{\varepsilon} \hq]^+ \,\big)
	&\;\;=\;\;& \inf_a \E  \big(\, [ p - e^{\varepsilon} \hq]^+ -a \,\big)^2\\
	&\;\;\leq \;\;& \E  \big(\, [ p - e^{\varepsilon} \hq]^+ -[p - e^{\varepsilon} q]^+ \,\big)^2\\
	&\;\;=\;\;& \E  \big(\, \frac{ (p - e^{\varepsilon} \hq)-(p - e^{\varepsilon} q)}{2} +\frac{|p - e^{\varepsilon} \hq|-|p - e^{\varepsilon} q|}{2} \,\big)^2\\
	&\;\;\leq\;\; & \E  \big(\, \frac{ e^\varepsilon|\hq -q|}{2} +\frac{\big||p - e^{\varepsilon} \hq|-|p - e^{\varepsilon} q|\big|}{2} \,\big)^2\\
	&\;\;\leq\;\; & e^{2\varepsilon}\E  \big(\, \frac{ |\hq -q|}{2} +\frac{ |\hq -q|}{2} \,\big)^2\\
	&\;\;=\;\;& e^{2\varepsilon}\E  \big(\,  \hq -q \,\big)^2\\
	&\;\;=\;\;& \frac{e^{2\varepsilon}q}{n}  \label{eq:e2pn}\\
	&\;\;\leq\;\;& \frac{e^{\varepsilon}p}{n} \;, \label{eq:epn}
\end{eqnarray*}
where the last step follows from the assumption that $e^{-\varepsilon}p\geq q$.

\bigskip\noindent{\bf Case 2: If $nq< 1$ and $e^{-\varepsilon}p< q$ or if  $nq\geq 1$ and $ne^{-\varepsilon}p<1$ : }

In both cases, $ne^{-\varepsilon}p<1$, and we have
\begin{eqnarray*}
	[ne^{-\varepsilon}p - n\hq]^+  = \left\{ \begin{array}{ll}
		ne^{-\varepsilon}p \quad &\text{w.p.}\quad e^{-nq}\\
		0 \quad &\text{w.p.} \quad 1-e^{-nq}\\
		\end{array} \right.\;,
\end{eqnarray*}
which is a Bernoulli random variable. The variance of it is
\begin{eqnarray*}
	\Var\big([p - e^{\varepsilon}\hq]^+\big)  &\;\;=\;\;&\frac{e^{2\varepsilon}}{n^2}\Var\big([ne^{-\varepsilon}p - n\hq]^+\big)  \\
	&\;\;=\;\;& p^2(1-e^{-nq})e^{-nq}\\
	&\;\;\leq\;\;& p^2 \;\;<\;\; \frac{e^{\varepsilon}p }{n}\;,
\end{eqnarray*}
where we used the assumption that $ne^{-\varepsilon}p<1$.

\bigskip\noindent{\bf Case 3: If $nq\geq 1$ and $ 1\leq ne^{-\varepsilon}p \leq \lfloor n q \rfloor-1$: }

Let $\lambda = nq$, and denote the random variable $X\sim \Poi(\lambda)$.

If $1\leq  ne^{-\varepsilon}p  \leq \lfloor nq \rfloor-1$, we know from Lemma~\ref{lem:poisson_tail} that
\begin{eqnarray}
	\prob \big(X\leq ne^{-\varepsilon}p \big)\;\;\leq\;\;  e^{-D_{\rm KL}( ne^{-\varepsilon}p\|nq)} \;,
\end{eqnarray}
where $D_{\rm KL}(k\|m) \triangleq m-k+k\ln(k/m)$.

We have
\begin{eqnarray*}
	\E\big(\, [p - e^{\varepsilon} \hq]^+ \,\big)^2 &\;\;=\;\;& \sum_{k=0}^{\lfloor ne^{-\varepsilon}p\rfloor}\Big(p-e^{\varepsilon}\frac{k}{n}\Big)^2\Big(\frac{\lambda^ke^{-\lambda}}{k!}\Big)\\
	&\;\;=\;\;& p^2 \prob \big(X\leq ne^{-\varepsilon}p \big)-2e^{\varepsilon}pq\prob\big(X\leq ne^{-\varepsilon}p-1\big)\\
	&&+\frac{e^{2\varepsilon}q}{n}\prob\big(X\leq ne^{-\varepsilon}p-1\big)+e^{2\varepsilon}q^2\prob\big(X\leq ne^{-\varepsilon}p-2\big)\\
	&\;\;=\;\;& p^2 \prob \big(X\leq ne^{-\varepsilon}p \big)-2e^{\varepsilon}pq\prob\big(X\leq ne^{-\varepsilon}p\big)+2e^{\varepsilon}pq\prob\big(X= \lfloor ne^{-\varepsilon}p\rfloor\big)\\
	&&+\frac{e^{2\varepsilon}q}{n}\prob\big(X\leq ne^{-\varepsilon}p\big)-\frac{e^{2\varepsilon}q}{n}\prob\big(X= \lfloor ne^{-\varepsilon}p\rfloor\big)\\
	&&+e^{2\varepsilon}q^2\prob\big(X\leq ne^{-\varepsilon}p\big)-e^{2\varepsilon}q^2\prob\big(X=\lfloor ne^{-\varepsilon}p\rfloor\big)-e^{2\varepsilon}q^2\prob\big(X= \lfloor ne^{-\varepsilon}p\rfloor-1\big)\\
	&\;\;=\;\;& (p^2-2e^\varepsilon pq+\frac{e^{2\varepsilon}q}{n}+e^{2\varepsilon}q^2) \prob \big(X\leq ne^{-\varepsilon}p \big)\\
	&&+(2e^{\varepsilon}pq-\frac{e^{2\varepsilon}q}{n}-e^{2\varepsilon}q^2-\frac{e^{2\varepsilon}q\lfloor n e^{-\varepsilon}p\rfloor}{n})\prob\big(X= \lfloor ne^{-\varepsilon}p\rfloor\big)\\
	&\;\;\lesssim \;\;&\Big( (e^{\varepsilon}q-p)^2\vee \frac{e^{2\varepsilon}q}{n}\Big)  e^{-D_{\rm KL}(ne^{-\varepsilon}p\|nq)}\;.
\end{eqnarray*}

Let $y = q/(e^{-\varepsilon}p)$, as $1\leq  ne^{-\varepsilon}p  \leq \lfloor nq \rfloor-1$, we know $y>1$. 

We have
\begin{eqnarray*}
	\frac{e^{2\varepsilon}q}{n}  e^{-D_{\rm KL}(ne^{-\varepsilon}p\|nq)} &\;\; =  \;\;& \frac{e^\varepsilon p y}{n} e^{-ne^{-\varepsilon}p(y-1-\ln y)}\\
	&\;\; \lesssim \;\;& \frac{e^\varepsilon p}{n}\;,
\end{eqnarray*}
where we used the assumption that $ne^{-\varepsilon}p\geq 1$ and the inequality that $ye^{-(y-1-\ln y)}$ is bounded by some constant for $y>1$. 

We have
\begin{eqnarray*}
	(e^{\varepsilon}q-p)^2  e^{-D_{\rm KL}(ne^{-\varepsilon}p\|nq)} &\;\; =  \;\;& (y-1)^2p^2e^{-ne^{-\varepsilon}p(y-1-\ln y)}\\
	&\;\; \lesssim \;\;& \frac{e^\varepsilon p}{n}\;,
\end{eqnarray*}
where we used the assumption that $ne^{-\varepsilon}p\geq 1$ and the inequality that $(y-1)^2e^{-x(y-1-\ln y)}\lesssim 1/x$ for $x,y>1$. 

It suffices to show that $f_y(x)=\ln x-x(y-1-\ln y)+2\ln(y-1)$ is bounded by some constant for $x,y>1$. Indeed, when $y-1-\ln y>1$, $f_y'(x) = 1/x-(y-1-\ln y)<0$. $f_y(x)$ is monotonically decreasing and $f_y(x)< f_y(1) =2\ln(y-1) -y+1+\ln y$, which is bounded. When $y-1-\ln y<1$, $f_y(x)$ attains maximum at $x = 1/(y-1-\ln y)$. In this case, we have $f_y(x)\leq 2\ln(y-1)-\ln(y-1-\ln y)-1$, which is also bounded.

\end{proof}

\begin{lemma}
	\label{lem:pqtail}
	Suppose $n\hq\sim \Poi(nq)$. Then,
	\begin{eqnarray}
	\prob\left(e^{\varepsilon}\hq \notin U(e^{2\varepsilon} q;c_1,c_1)\right) \;\;\leq\;\; \frac{2}{n^{-c_1e^{-\varepsilon} /3}}\;\;\leq\;\;\frac{2}{n^{-c_1 /3}} \;,
	\end{eqnarray}
	and
	\begin{eqnarray}
	\prob\left(\hq \notin U(e^{\varepsilon} q;c_1,c_1)\right) \;\;\leq\;\; \frac{2}{n^{-c_1 /3}}\;.
	\end{eqnarray}
	
\end{lemma}
\begin{proof}
	The second inequality is exactly \cite[Lemma~1]{JYT18}. We now prove the first inequality.
	
	If $q\leq \frac{c_1  e^{-\varepsilon} \ln n}{n}$, we have
	\begin{eqnarray*}
		\prob\left(e^{\varepsilon}\hq \notin U(e^{2\varepsilon}q:c_1,c_1)\right) &\;\; = \;\;& \prob\left(n\hq \geq 2c_1 e^{-\varepsilon} \ln n\right) \\
		&\;\; \leq \;\;& \prob \left(\Poi(c_1 e^{-\varepsilon} \ln n)\geq 2c_1 e^{-\varepsilon} \ln n\right)\\
		& \;\;\leq\;\; & e^{-\frac{c_1 e^{-\varepsilon} \ln n}{3}}\;,
	\end{eqnarray*}
	where we applied Lemma~\ref{lem:poisson_tail} in the last inequality.
	
	If $q> \frac{c_1  e^{-\varepsilon} \ln n}{n}$, we have
	\begin{eqnarray*}
		\prob\left(e^{\varepsilon}\hq \notin U(e^{2\varepsilon}q;c_1,c_1)\right) &\;\; = \;\;& \prob\left(e^{\varepsilon}\hq>e^{\varepsilon}q+\sqrt{\frac{c_1e^{\varepsilon}q \ln n}{n}} \right)+\prob\left(e^{\varepsilon}\hq<e^{\varepsilon}q-\sqrt{\frac{c_1e^{\varepsilon}q \ln n}{n}} \right)\\
		&\;\; \leq \;\;& \prob\left(\Poi(nq)>nq+\sqrt{c_1qe^{-\varepsilon} n \ln n}\right)+\prob\left(\Poi(nq)<nq-\sqrt{c_1qe^{-\varepsilon} n\ln n}\right)\\
		&\;\; \leq \;\;& e^{-\frac{c_1e^{-\varepsilon} \ln n}{nq}\frac{nq}{3}}+e^{-\frac{c_1e^{-\varepsilon} \ln n}{nq}\frac{nq}{2}}\\
		&\;\; \leq \;\;& \frac{2}{n^{c_1e^{-\varepsilon}/3}}\;.
	\end{eqnarray*}
	
\end{proof}

\subsubsection{Lemmas on the best polynomial approximation}
The first-order and second-order symmetric difference with function $\varphi(x) = \sqrt{x(1-x)}$ are defined as
\begin{eqnarray}
	\Delta_{h\varphi} f(x) \;\;\triangleq \;\; f(x+\frac{h\varphi(x)}{2})-f(x-\frac{h\varphi(x)}{2})
	\;, \label{eq:first order symmetric difference}
\end{eqnarray}
and
\begin{eqnarray}
	\Delta_{h\varphi}^2 f(x) \;\;=\;\; f(x+h\varphi(x))-2f(x)+f(x-h\varphi(x))
	\;, \label{eq:second order symmetric difference}
\end{eqnarray}
respectively.

For function $f(x)$ with domain $[0,1]$, the first-order Ditzian-Totik modulus of smoothness is defined as
\begin{eqnarray}
	\omega_\varphi^1(f,t) \;\;\triangleq \;\; \sup_{0<h\leq t}\|\Delta_{h\varphi}^1 f(x)\|_\infty 
	\;, \label{eq:smooth1}
\end{eqnarray}
and the second-order Ditzian-Totik modulus of smoothness is defined as
\begin{eqnarray}
	\omega_\varphi^2(f,t) \;\;\triangleq \;\; \sup_{0<h\leq t}\|\Delta_{h\varphi}^2 f(x)\|_\infty 
	\;. \label{eq:smooth2}
\end{eqnarray}
The following lemma upper bounds the best polynomial approximation error by the Ditzian-Totik moduli.


\begin{lemma}[{\cite[Theorem~7.2.1 and 12.1.1]{ditzian1987totik}}]
\label{lem:poly_error}
	There exists a constant $M(r)>0$ such that for any function $f\in C[0,1]$, 
	\begin{eqnarray}
		\Delta_L[f;[0,1]]\;\;\leq\;\; M(r)\omega_\varphi^r(f,\frac{1}{L})\;, \quad L> r\;,
	\end{eqnarray} 
	where $\Delta_L[f;I]$ denotes the distance of the function $f$ to the space ${\rm poly_L}$ in the uniform norm $\|\cdot\|_\infty,I$ on $I\subset \reals$. Moreover, if $f(x):[0,1]^2\mapsto \reals$, we have
	\begin{eqnarray}
		\Delta_L[f;[0, 1]^2]\;\;\leq\;\;M \omega_{[0,1]^2}^r(f,\frac{1}{L})\;, \quad L>r\;,
	\end{eqnarray}
	where $M$ is independent of $f$ and $L$, and $\Delta_L[f;[0,1]^2]$ denotes the distance of the function $f$ to the space ${\rm poly_L^2}$ in the uniform norm on $[0,1]^2$.
\end{lemma}


\begin{lemma}
	\label{lem:delf2}
	For $f(x)= [p-e^\varepsilon 2x\Delta  ]^+$, for some $\Delta>0$, $p\in[0,2e^\varepsilon\Delta]$, $x\in[0,1]$, and any integer $K\geq 1$, 
	\begin{eqnarray*}
	\omega_\varphi^2(f,K^{-1}) \;\;=\;\; \left\{ 
		\begin{array}{ll}
		p  &\frac{p}{2e^\varepsilon\Delta}\leq \frac{1}{1+K^2}\\ 
		 \frac{\sqrt{p(2e^\varepsilon\Delta-p)}}{K}  &\frac{1}{1+K^2}\leq\frac{p}{2e^\varepsilon\Delta}\leq \frac{K^2}{1+K^2}\\
		e^\varepsilon\Delta-p  &\frac{K^2}{1+K^2}\leq \frac{p}{2e^\varepsilon\Delta}\leq 1
		\end{array}
		\right.
		\;\; \lesssim \;\;  \min \Big\{p \;,\frac{\sqrt{p(2e^\varepsilon\Delta-p)}}{K} \;, e^\varepsilon\Delta-p    \Big\} \;,
	\end{eqnarray*}
	where $\omega_\varphi^2(f,t)$ is defined in Eq.~\eqref{eq:smooth2}.
	
\end{lemma}
\begin{proof}

Let $g(x):=| p-2e^\varepsilon\Delta x|$. 
\begin{eqnarray*}
	\Delta_{h\varphi}^2 f(x) &=& [p-2e^\varepsilon\Delta (x+h\varphi(x))]^+
 -2[ p-2e^\varepsilon\Delta x]^++[ p-2e^\varepsilon\Delta (x-h\varphi(x))]^+\\
 &=&\frac{\big(p-2e^\varepsilon\Delta (x+h\varphi(x))\big)
 -2\big( p-2e^\varepsilon\Delta x\big)+\big(p-2e^\varepsilon\Delta (x-h\varphi(x))\big)}{2}+\\
 &&\frac{\big|p-2e^\varepsilon\Delta (x+h\varphi(x))\big|
 -2\big| p-2e^\varepsilon\Delta x\big|+\big|p-2e^\varepsilon\Delta (x-h\varphi(x))\big|}{2}\\
 &=& \frac{1}{2}\Delta_{h\varphi}^2g(x)
\end{eqnarray*}

Hence,
\begin{eqnarray}
	\omega_\varphi^2(f,t) = \frac{1}{2}\omega_\varphi^2(g,t) 
\end{eqnarray}

It follows from \cite[Lemma~12]{JYT18} that, for some $\Delta>0$, $p\in [0,2e^\varepsilon \Delta]$, $x\in [0,1]$ and any integer $K\geq 1$, we have:
\begin{eqnarray*}
	\omega_\varphi^2(g,K^{-1}) \;\;=\;\; \left\{ 
		\begin{array}{ll}
		2p  &\frac{p}{2e^\varepsilon\Delta}\leq \frac{1}{1+K^2}\\ 
		 \frac{2\sqrt{p(2e^\varepsilon\Delta-p)}}{K}  &\frac{1}{1+K^2}\leq\frac{p}{2e^\varepsilon\Delta}\leq \frac{K^2}{1+K^2}\\
		2e^\varepsilon\Delta-p &\frac{K^2}{1+K^2}\leq \frac{p}{2e^\varepsilon\Delta}\leq 1
		\end{array}
		\right.
		\;\;,
\end{eqnarray*}
which implies the desired bound.

\end{proof}
\begin{lemma}
\label{lem:smoothness_sqrt_xy}
	Suppose $f(x) = [\sqrt{x}-\sqrt{a}]^+$, $x\in [0,1]$ and $a\in [0,1]$. Then
	\begin{eqnarray}
	\omega_\varphi^1(f,t)\;\;\leq\;\;\frac{t}{\sqrt{2}}\;.
	\end{eqnarray}	
	Similarly, suppose $f(x) = [\sqrt{a}-\sqrt{x}]^+$, $x\in [0,1]$ and $a\in [0,1]$. Then
	\begin{eqnarray}
	\omega_\varphi^1(f,t)\;\;\leq\;\;\frac{t}{\sqrt{2}}\;.
	\end{eqnarray}	
\end{lemma}
\begin{proof}
Let $g(x) = |\sqrt{x}-\sqrt{a}|$.
	\begin{eqnarray*}
		\Delta^1_{h\varphi}f(x) &\;\; = \;\;& \left|f(x+\frac{h\varphi(x)}{2})-f(x-\frac{h\varphi(x)}{2})\right|\\
		&\;\; = \;\;& \left|\left[\sqrt{x+\frac{h\varphi(x)}{2}}-\sqrt{a}\right]^+-\left[\sqrt{x-\frac{h\varphi(x)}{2}}-\sqrt{a}\right]^+\right|\\
		&\;\; = \;\;&  \left|\frac{1}{2}\left(\left|\sqrt{x+\frac{h\varphi(x)}{2}}-\sqrt{a}\right|-\left|\sqrt{x-\frac{h\varphi(x)}{2}}-\sqrt{a}\right|\right)+\right. \nonumber\\
		&&\left.\frac{\sqrt{x+\frac{h\varphi(x)}{2}}-\sqrt{x-\frac{h\varphi(x)}{2}}}{2}\right|\\
		&\;\; \leq \;\;& \frac{1}{2}\nabla^1_{h\varphi}g(x)+\frac{1}{2} \left|\sqrt{x+\frac{h\varphi(x)}{2}}-\sqrt{x-\frac{h\varphi(x)}{2}}\right|\\
		&\;\; = \;\;& \frac{1}{2}\nabla^1_{h\varphi}g(x)+\frac{1}{2} \left|
		\frac{ h\varphi(x) }{\sqrt{x+h\varphi(x)/2}+\sqrt{x-h\varphi(x)/2}}\right|\\
		&\;\; \leq \;\;&\frac{1}{2}\nabla^1_{h\varphi}g(x) +\frac{1}{2}\left|
		\frac{ h\varphi(x) }{\sqrt{x+h\varphi(x)/2+x-h\varphi(x)/2}}\right|\\
		&\;\; \leq \;\;&\frac{1}{2}\nabla^1_{h\varphi}g(x) +\frac{h\sqrt{1-x}}{2\sqrt{2}}\\
		&\;\; \leq \;\;& \frac{t}{\sqrt{2}}\;,
	\end{eqnarray*}
	where we used the fact that $\sqrt x+\sqrt y\leq x+y$ and \cite[Lemma~11]{JYT18} for smoothness of $g$.
\end{proof}

\begin{lemma}[{\cite[Lemma~27]{HJW16}}]
	\label{lem:Qopt_magnitude}
	Let $p_n(x) = \sum_{v=0}^na_vx^v$ be a polynomial of degree at most $n$ such that $|p_n(x)|\leq A$ for $x\in [a,b]$. Then
	\begin{enumerate}
		\item If $a+b\neq 0$, then
		\begin{eqnarray}
			|a_v|\;\;\leq\;\; 2^{7n/2}A\left|\frac{a+b}{2}\right|^{-v}\left(\left|\frac{b+a}{b-a}\right|^n+1\right)\;,\qquad v= 0,1,\cdots, n\;.
		\end{eqnarray}
		\item If $a+b =0$, then
		\begin{eqnarray}
			|a_v|\;\;\leq\;\; Ab^{-v}(\sqrt(2)+1)^n\;,\qquad v= 0,1,\cdots, n\;.
		\end{eqnarray}
	\end{enumerate}
\end{lemma}


\subsubsection{Lemmas on the uniformly unbiased minimum variance unbiased estimator}
\label{sec:MVUE}


\begin{lemma}[{\cite[Lemma~18]{JYT18}}]
\label{lem:Qopt_moment}
Suppose $nX \sim \Poi(np)$, $p\geq 0$, $q\geq 0$. Then, the estimator
\begin{eqnarray}
	g_{j,q}(X)\triangleq \sum_{k=0}^j{j \choose k} (-q)^{j-k}\prod_{h=0}^{k-1}\left(X-\frac{h}{n}\right)
\end{eqnarray}
is the unique uniformly minimum variance unbiased estimator for $(p-q)^j$, $j\geq 0$, $j\in \mathbb{N}$, and its second moment is given by
\begin{eqnarray}
	\E\left[\big(g_{j,q}(X)\big)^2\right] \;\; =\;\; \sum_{k=0}^j{j\choose k}^2(p-q)^{2(j-k)}\frac{p^kk!}{n^k}\;\;=\;\; j!\left(\frac{p}{n}\right)^jL_j\left(-\frac{n(p-q)^2}{p}\right)\text{assuming $p>0$}\;,
\end{eqnarray}
where $L_m(x)$ stands for the Laguerre polynomial with order $m$, which is defined as:
\begin{eqnarray}
	L_m(x)\;\;=\;\;\sum_{k=0}^m{m\choose k}\frac{(-x)^k}{k!}\;.
\end{eqnarray}
If $M\geq \max\left\{\frac{n(p-q)^2}{p},j\right\}$, we have
\begin{eqnarray}
	\E\left[\big(g_{j,q}(X)\big)^2\right]\;\;\leq\;\;\left(\frac{2Mp}{n}\right)^j\;.
\end{eqnarray}
When $k=0$, $\prod_{h=0}^{k-1}\left(X-\frac{h}{n}\right)\triangleq 1$. When $p=0$, $g_{j,q}(X)\equiv (-q)^j$, $\E\left[g_{j,q}(X)\right]^2\equiv q^{2j}$.
\end{lemma}
\begin{lemma}
\label{lem:uumvue}
	Suppose $(n\hp,n\hq)\sim \Poi(np)\times\Poi(nq)$. Then the following estimator using $(\hp,\hq)$ is the unique uniformly minimum unbiased estimator for $(e^\varepsilon q-p)^j$, $j\geq 0$, $j\in \Z$:
	
	\begin{eqnarray}
		\hat{A}_j(\hp,\hq) \;\;=\;\; \sum_{k=0}^{j} {j\choose k }\prod_{i=0}^{k-1}\left(\hq-\frac{i}{n}\right)e^{\varepsilon k}(-1)^{j-k} \prod_{m=0}^{j-k}\left(\hp-\frac{m}{n}\right)\;.
	\end{eqnarray}
	Furthermore,
	\begin{eqnarray}
		\E \hat{A}_{j}^{2} \;\;\leq\;\;\left(2(e^\varepsilon q-p)^{2} \vee \frac{8 j(e^{2\varepsilon} q \vee p)}{n}\right)^{j}\;.
	\end{eqnarray}
\end{lemma}

\begin{proof}
	It follows from \cite[Lemma~19]{JYT18} and binomial theorem that $\hat{A}_j(\hp,\hq)$ is the unique uniformly minimum variance unbiased estimator for $(e^\varepsilon q-p)^j$. Now we show $\E \hat{A}_{j}^{2}$ is bounded.
	
	It follows from binomial theorem again that for any fixed $r>0$,
	\begin{eqnarray}
		(e^\varepsilon q- p )^j &\;\;=\;\;& (e^\varepsilon q -r + r- p)^j\\
		&\;\;=\;\;& \sum_{k=0}^j{j \choose k} (e^\varepsilon q-r)^k(-1)^{j-k}(p-r)^{j-k}\;.
	\end{eqnarray}
	The following estimator is also unbiased for estimating $(e^\varepsilon q-p)^j$,
	\begin{eqnarray}
		\sum_{k=0}^je^{\varepsilon k}g_{k, \frac{r}{e^\varepsilon}}(\hq)(-1)^{j-k}g_{j-k,r}(\hp)\;,
	\end{eqnarray}
	where $g_{i, q}(\hp)$ is defined in Lemma~\ref{lem:Qopt_moment}.
	
	Define $M_1 = \frac{n(q-\frac{r}{e^{\varepsilon}})^2}{p}\vee j$, $M_2 = \frac{n(p-r)^2}{p}\vee j$, $M = 2(e^\varepsilon q-p)^{2} \vee \frac{8 j(e^{2\varepsilon} q \vee p)}{n}$ and set $r = \frac{e^\varepsilon q+p}{2}$.
	
	Denote $\|X\|_2 = \sqrt{\E(X-\E X)^2}$ for random variable $X$. It follows from Lemma~\ref{lem:Qopt_moment} that
	\begin{eqnarray*}
		\|\hat{A}_j\|_2 &\;\;\leq\;\;& \sum_{k=0}^j {j \choose k} e^{\varepsilon k} \|g_{k, \frac{r}{e^\varepsilon}}(\hq)\|_2\cdot\|g_{j-k,r}(\hp)\|_2\\
		&\;\;\leq\;\;& \sum_{k=0}^j {j \choose k} e^{\varepsilon k} \left(\frac{2 M_{1} q}{n}\right)^{k / 2}\left(\frac{2 M_{2} p}{n}\right)^{(j-k) / 2}\\
		&\;\;=\;\;& \left(e^\varepsilon \sqrt{\frac{2 M_{1} q}{n}}+\sqrt{\frac{2 M_{2} p}{n}}\right)^{j}\\
		&\;\;=\;\;& \left( \sqrt{\frac{(e^\varepsilon q-p)^2}{2}\vee \frac{2j e^{2\varepsilon}q}{n}}+\sqrt{\frac{(e^\varepsilon q-p)^2}{2}\vee \frac{2j p}{n}}\right)^{j}\\
		&\;\;=\;\;& M^{j/2}\;.
	\end{eqnarray*}
\end{proof}

\subsection{Proof of Theorem~\ref{thm:Qmle}}
\label{sec:proof_Qmle}
	
Note that 
\begin{eqnarray} 
	\E_Q \big[\,|d_\varepsilon(P\| Q_n)-d_\varepsilon(P \| Q)|^2\,\big]  \;\;=\;\; 
	\Big( \sum_{i=1}^S \E_Q\big[ \, [p_i-e^\varepsilon \hq_i]^+\big]  - [p_i-e^\varepsilon q_i]^+  \Big)^2 + \var\big(d_\varepsilon(P\| Q_n)\big)\;. 
	\label{eq:Qmle_biasvariance}
\end{eqnarray}

We first claim the following upper bound for all $P$: 
\begin{eqnarray}
 	 \sum_{i=1}^S \E_Q\big[ \, [p_i-e^\varepsilon \hq_i]^+ \big] - [p_i-e^\varepsilon q_i]^+  
	\;\; \leq \;\;   \sum_{i=1}^S \, p_i \wedge \sqrt{\frac{e^\varepsilon p_i}{n }}   \;,
	\label{eq:Qmle_ub}
\end{eqnarray}
where the inequality follows from Lemma~\ref{lem:poiss_shiftedrelu}.

For the upper bound of the variance term in Eq.~\eqref{eq:Qmle_biasvariance}, we have
\begin{eqnarray}
	\var\big(d_\varepsilon(P_n\| Q)\big) \;\;=\;\;\sum_{i=1}^S\Var\big(\, [p_i - e^{\varepsilon} \hq_i]^+ \,\big)\;\;\lesssim  \;\; \sum_{i=1}^S\frac{e^\varepsilon q_i}{n}\;\;= \;\; \frac{e^\varepsilon}{n}\;,
\end{eqnarray}
where the inequality follows from Lemma~\ref{lem:poiss_varub}.

We next construct $Q$ to get the lower bound. 
Let
\begin{eqnarray}
q_i \;\;  = \;\; \left\{ 
	\begin{array}{rl}
		 e^{-\varepsilon} p_i& \;,\;i \in S_+\\
		\frac{1- Q(S_+)}{S-|S_ + | } & \;,\; i \in S_- 
	\end{array}\right.
\end{eqnarray}
where $S_+$ is a set of indices satisfying $Q(S_+) = \sum_{i\in S_+} q_i \leq e^{-\varepsilon}$. 

Note that each term in the bias of Eq.~\eqref{eq:Qmle_biasvariance} is non-negative via Jensen's inequality, which gives 
\begin{eqnarray}
	 \sum_{i=1}^{S} \E_Q\big[ \, [p_i-e^\varepsilon \hq_i]^+ \big] - [p_i-e^\varepsilon q_i]^+   
		&\;\;\geq\;\; & e^\varepsilon \sum_{i\in S_+} \E_Q\big[ \, [q_i- \hq_i ]^+ \big] \\
		&\;\;\gtrsim\;\; & \sum_{i\in S_+} \Big\{  p_i \wedge \sqrt{\frac{e^\varepsilon p_i}{n}} \Big\}  \;,
\end{eqnarray}
where we used Lemma~\ref{lem:poiss_relu}. Note that we can choose $Q$ such that $|S_+| = S/3$. This implies the desired lower bound when plugged into Eq.~\eqref{eq:Qmle_biasvariance}. 

\subsection{Proof of Theorem~\ref{thm:Qopt}} 
\label{sec:proof_Qopt}
	

Define good events, where our choice of the regimes are correct as 
\begin{eqnarray}
	E &\triangleq& \Big\{ \, \{ i : \hq_{i,1} > U(p_i;c_1,c_2) \} \subseteq S^+ \,\Big\}  \cap 
	\Big\{ \{ i: \hq_{i,1} < U(p_i;c_1,c_2)\} \subseteq S^- \Big\} \nonumber\\ 
	&&\;\;\;\;\; \cap 
	\Big\{ \, \{i: \hq_{i,1}\in U(p_i,c_1,c_2) \subseteq \{i:q_i \in U(p_i,c_1,c_1) \}  \}\,\Big\}   \;,
	\label{eq:Qopt_defE}
\end{eqnarray}
where $S^+=\{i:e^\varepsilon q_i \geq p_i \}$ and 
$S^- = \{i:e^\varepsilon q_i \leq p_i \}$.
Decompose the error under the good events as 
\begin{eqnarray}
	\cE_1 &\;\;\triangleq \;\;& \sum_{i\in I_1} \big\{ [p_i-e^\varepsilon \hq_{i,2} ]^+ -[p_i-e^\varepsilon q_i]^+ \big\} \;, \\
	\cE_2 &\;\;\triangleq \;\;& \sum_{i\in I_2} \big\{ \tD_K(\hq_{i,2};p_i) -[p_i-e^\varepsilon q_i]^+ \big\} \;.
\end{eqnarray}
where the indices of those regimes under the good event are 
\begin{eqnarray}
	I_1 &\;\;\triangleq \;\;& \{ i : \hq_{i,1} < U(p_i;c_1,c_2), e^\varepsilon q_i \leq p_i   \}\\
	I_2 &\;\;\triangleq \;\;& \{ i : \hq_{i,1} \in U(p_i;c_1,c_2), q_i \in U(p_i;c_1,c_1) \}\;.
\end{eqnarray}
We can  bound the squared error as 
\begin{eqnarray}
	\E \big[ \, \big( \hd_{\varepsilon,K,c_1,c_2}(P\| Q_n) - d_\varepsilon(P\|Q) \big)^2 \, \big] & \;\;\leq \;\;& 
		\E[ \, \big( \hd_{\varepsilon,K,c_1,c_2}(P\| Q_n) - d_\varepsilon(P\|Q) \big)^2\, \ind(E)\, ] +  \prob( E^c)    \nonumber \\
		&\;\;\leq\;\; & \E[(\cE_1+\cE_2)^2] + \prob( E^c)    \nonumber\\
		&\;\;\leq\;\;& 2  \E[(\cE_1)^2] + 2  \E[(\cE_2)^2]  +  \prob( E^c)  \;,
	\label{eq:Qopt_error}
\end{eqnarray}

The last term on the bad event is bounded by $ 3S/n^\beta $ as shown in the following lemma, 
and a proof is provided in Section~\ref{sec:proof_Qopt_error}.  This is a direct consequence of standard concentration inequality for Poisson variables. 
\begin{lemma}
	\label{lem:Qopt_error}
	Let $\beta=\min\{\frac{c_2^2}{3c_1}, \frac{(c_1 - c_2)^2}{4c_1}, \frac{(\sqrt{c_1} - \sqrt{c_2})^2}{3} \} $, then  for the good event $E$ defined in \eqref{eq:Qopt_defE}, 
	\begin{eqnarray}
		\prob(E^c)  \;\; \leq \;\; \frac{3\,S}{n^\beta}\;. \label{eq:Qopt_errorub}
	\end{eqnarray}
\end{lemma}

The first term is bounded by $e^{\varepsilon}/n$, as 
\begin{eqnarray}
	\E[(\cE_1)^2] & = &  \E [\var(\cE_1 | I_1) + (E[ \cE_1| I_1])^2] \; = \; \E[\var(\cE_1|I_1)] \; \leq \; \sum_{i=1}^S \frac{e^{\varepsilon} \, p_i}{n} \;\;\leq\;\; \frac{e^{\varepsilon}}{n}
\label{eq:Qopt_term2}
\end{eqnarray}
where we used Lemma \ref{lem:poiss_varub} and the fact that $\E[\cE_1|I_1] = 0$ with probability one.

The second term is bounded by the following lemma, with a proof in Section~\ref{sec:proof_Qopt_key}. 
\begin{lemma}
	\label{lem:Qopt_key}
	For $n\hq\sim\Poi(nq)$ and $q\in U(p;c_1,c_1)$, there exists a universal constant $B>0$ 
	such that 
	\begin{eqnarray}		
		\big|\, \E [ \tD_K(\hq;p) ]  - [p-e^\varepsilon\,q]^+ \,\big|  &\;\lesssim \;& 
			p\wedge \frac{1}{K}\sqrt{\frac{e^\varepsilon \,p\,c_1\, \ln n}{n}} \;, \text{ and }\label{eq:Qopt_keybias} \\
		  \var\big(\tD_K(\hq;p) \big) &\;\lesssim \;& \frac{B^K\,e^\varepsilon \, c_1\, \ln n}{n}(p+e^\varepsilon q) \;, \label{eq:Qopt_keyvar} 
	\end{eqnarray}
	where 
	$\tD_K(\hp;q)$ is the uniformly minimum variance unbiased estimate (MVUE) 
	defined in Eq.~\eqref{eq:Qopt_mvue}, 
	$U(q,c_1)$ is defined in Eq.~\eqref{eq:Qopt_region1}, and 
	$K=c_3 \ln n$ for some $c_3<c_1$.
	
\end{lemma}

We have
\begin{eqnarray}
	\E[(\cE_2)^2] & \lesssim & \sum_{i=1}^S\frac{B^K e^\varepsilon c_1\ln n}{n}(p_i+e^\varepsilon q_i)+ \Big(\sum_{i=1}^{S} p_i\wedge \frac{1}{K}\sqrt{\frac{e^\varepsilon p_i c_1 \ln n}{n}} \Big)^2\\
	&\lesssim & \frac{c_1 \ln n}{n^{1-c_3 \ln B}}e^\varepsilon(e^\varepsilon + 1) +\Big(\sum_{i=1}^{S} p_i\wedge \sqrt{\frac{e^\varepsilon p_i c_1 }{c_3^2n \ln n}} \Big)^2 \;.\label{eq:Qopt_term3} 
\end{eqnarray}

Substituting bounds \eqref{eq:Qopt_term3}, \eqref{eq:Qopt_term2} and \eqref{eq:Qopt_errorub}, we get that 
\begin{eqnarray}
	\E\big[\, \big(   \hd_{\varepsilon,K,c_1,c_2}(P \| Q_n) - d_\varepsilon(P\|Q)  \big)^2 \,\big] &\;\;\lesssim\;\; &  
	\frac{c_1 \ln n}{n^{1-c_3 \ln B}}e^{2\varepsilon} +\Big(\sum_{i=1}^{S} p_i\wedge \sqrt{\frac{e^\varepsilon p_i c_1 }{c_3^2n \ln n}} \Big)^2 + \frac{S}{n^\beta} 
\end{eqnarray}
where we use the fact that $\frac{e^{\varepsilon}}{n}\lesssim \frac{c_1 \ln n}{n^{1-c_3 \ln B}}e^\varepsilon(e^\varepsilon +1)$.

As $\ln n \gtrsim \ln S $,  one may choose  $c_1$ large enough to and let $c_2 = c_1/2$ to ensure that $\frac{S}{n^\beta}\lesssim  \frac{c_1 \ln n}{n^{1-c_3 \ln B}}e^{2\varepsilon}$. As $\ln n \lesssim \ln \left(e^{-\varepsilon}\sum_{i=1}^{S} \sqrt{e^\varepsilon p_{i}} \wedge p_{i} \sqrt{n \ln n}\right)$, one may choose $c_3$ small enough to ensure $\frac{c_1 \ln n}{n^{1-c_3 \ln B}}e^{2\varepsilon} \lesssim \Big(\sum_{i=1}^{S} p_i\wedge \sqrt{\frac{e^\varepsilon p_i c_1 }{c_3^2n \ln n}} \Big)^2$. The worst case of $P$ result is proved upon noting
\begin{eqnarray}
	\sum_{i=1}^{S} p_i\wedge \sqrt{\frac{e^\varepsilon p_i}{n \ln n}}
	\;\;\leq \;\; \sum_{i=1}^{S} \sqrt{\frac{e^\varepsilon p_i  }{n \ln n}}
x\end{eqnarray}

Note that in the worst case of $P$, we do not require $\ln n\gtrsim \ln S$, as we can take $c_1$ large enough and $c_2 = c_1/2$ to ensure $\frac{S}{n^\beta}\lesssim \frac{e^\varepsilon S}{n\ln n}$.

\subsubsection{Proof of Lemma~\ref{lem:Qopt_error}}
\label{sec:proof_Qopt_error}
	
Let $E_1 = \Big\{ \, \{ i : \hq_{i,1} > U(p_i;c_1,c_2) \} \subseteq S^+ \,\Big\}$, $E_2 = \Big\{ \{ i: \hq_{i,1} < U(p_i;c_1,c_2)\} \subseteq S^- \Big\}$ and $E_3 = \Big\{ \, \{i: \hq_{i,1}\in U(p_i,c_1,c_2) \subseteq \{i:q_i \in U(p_i,c_1,c_1) \}  \}\,\Big\}$. 
	We first show $\prob(E_1^c) \leq S n^{-\beta}$ for $\beta\leq (c_2)^2/(3c_1)$.
	\begin{eqnarray*}
		\prob(E_1^c) &=& \prob\Big(\, \bigcup_{i=1}^S \{\hq_{i,1} > U(p_i;c_1,c_2), e^\varepsilon q_i<p_i\} \,\Big)\\
		&\leq & S\, \max_{i\in S} \prob(\{\hq_{i,1} > U(p_i;c_1,c_2), e^\varepsilon q_i<p_i\} )\\
		&= & S\, \max_{i\in S} \prob(\{\Poi(nq_i) > n U(p_i;c_1,c_2), q_i<e^{-\varepsilon}p_i\} )\\
		&\leq & S\, \max_{i\in S} \prob(\{\Poi(ne^{-\varepsilon}p_i) > n U(p_i;c_1,c_2)\} )
	\end{eqnarray*}
If $p_i\leq {(c_1e^\varepsilon\ln n)}/{n}$, it follows from Lemma~\ref{lem:poisson_tail} that
	\begin{eqnarray*}
		\prob(\{\Poi(ne^{-\varepsilon}p_i) > n U(p_i;c_1,c_2)\} ) 
		&=& \prob(\{\Poi(ne^{-\varepsilon}p_i) >  (c_1+c_2)\ln n\} ) \\
		&\;\;\leq\;\; & \prob(\{\Poi(c_1\ln n) >  (c_1+c_2)\ln n\} ) \\
		&\leq & e^{-\frac{c_2^2}{3c_1}\ln n}\;.
	\end{eqnarray*}
If $p_i> {(c_1e^\varepsilon\ln n)}/{n}$, it follows from  Lemma~\ref{lem:poisson_tail}  that
	\begin{eqnarray*}
		\prob(\{\Poi(ne^{-\varepsilon}p_i) > n U(p_i;c_1,c_2)\} )
		&\;\;=\;\;& \prob(\{\Poi(ne^{-\varepsilon}p_i) >  ne^{-\varepsilon}p_i+\sqrt{c_2e^{-\varepsilon}p_i n\ln n}\} ) \\
		&\;\;\leq\;\; & e^{-\frac{c_2\ln n}{3}}\;.
	\end{eqnarray*}
Together, these bounds imply that $\prob(E_1^c)\leq S n^{-\beta}$. 
	
Next, we show $\prob(E_2^c)\leq S n^{-\beta}$, for positive constant $\beta\leq c_2^2/(3 c_1)$. Recall that 
	\begin{eqnarray*}
		\prob(E_2^c) &\;\;=\;\;& \prob\Big(\,\bigcup_{i=1}^S \{\hq_{i,1} < U(p_i;c_1,c_1), e^\varepsilon q_i\geq p_i\} \,\Big)\;.
	\end{eqnarray*}
	If $p_i\leq (c_1e^\varepsilon\ln n)/n$, $\prob(\{\hq_{i,1} < U(p_i;c_1,c_2), e^\varepsilon p_i\geq q_i\} )= 0$. 
	If  $p_i> (c_1e^\varepsilon\ln n)/n$, it follows from Lemma~\ref{lem:poisson_tail} that
	\begin{eqnarray*}
		\prob(\{\hq_{i,1} < U(p_i;c_1,c_2), e^\varepsilon q_i\geq p_i\})
		&\;\;=\;\;& \prob(\{\Poi(nq_i) < ne^{-\varepsilon}p_i-\sqrt{c_2e^{-\varepsilon}p_i n\ln n},  q_i\geq e^{-\varepsilon} p_i\})\\
		&\leq & \prob(\{\Poi(ne^{-\varepsilon} p_i) < ne^{-\varepsilon}p_i-\sqrt{c_2e^{-\varepsilon}p_i n\ln n},  q_i\geq e^{-\varepsilon} p_i\})\\
		&\leq & e^{-\frac{c_2\ln n }{2}}\;.
	\end{eqnarray*}
	As $c_2 /2 \geq c_2^2/(3c_1)$, we have $\prob(E_2^c)\leq S n^{-\beta}$.

	Finally, we show that $\prob(E_3^c) \leq S n^{-\beta}$ for $\beta\leq \min\{ (c_1-c_2)^2/(4 c_1), (\sqrt{c_1}-\sqrt{c_2})^2/3\}$.
	Recall  
	\begin{eqnarray*}
		\prob(E_3^c) &\;\;=\;\;& \prob(\bigcup_{i=1}^S \{\hq_{i,1} \in  U(p_i;c_1,c_2), q_i\notin U(p_i;c_1,c_1)\} )\;.
	\end{eqnarray*}
	If $p_i\leq (c_1e^\varepsilon\ln n)/n$,
	\begin{eqnarray*}
		\prob(\{\hq_{i,1} \in  U(p_i;c_1,c_2), q_i\notin U(p_i;c_1,c_1)\})
		&\;\; = \;\;&\prob(\{\Poi(nq_i) \leq (c_1+c_2)\ln n, nq_i\geq 2c_1\ln n\})\\
		&\leq &\prob(\{\Poi(2c_1\ln n) \leq (c_1+c_2)\ln n\})\\
		&\leq & e^{-\frac{(c_1-c_2)^2}{4c_1}\ln n}\;. 
	\end{eqnarray*}
	If $p_i> (c_1e^\varepsilon\ln n)/n$, 
	\begin{eqnarray*}
		&& \prob(\{\hq_{i,1} \in  U(p_i;c_1,c_2), q_i> U(p_i;c_1,c_1)\})\\
		&\;\;=\;\;& \prob(\{\Poi(nq_i) \leq ne^{-\varepsilon}p_i+\sqrt{c_2e^{-\varepsilon}p_i n\ln n}, nq_i> ne^{-\varepsilon}p_i+\sqrt{c_1e^{-\varepsilon}p_i n\ln n}\})\\
		&\;\;\leq\;\;&  \prob(\{\Poi(ne^{-\varepsilon}p_i+\sqrt{c_1e^{-\varepsilon}p_i n\ln n}) \leq ne^{-\varepsilon}p_i+\sqrt{c_2e^{-\varepsilon}p_i n\ln n} \})\\
		&\;\;\leq\;\;&  e ^ { - \left( \frac { ( \sqrt{c_1} - \sqrt {c_2} ) \sqrt {e^{-\varepsilon}p_i n \ln n } } { n e^{-\varepsilon}p_i + \sqrt { c_1 e^{-\varepsilon}p_i n \ln n } } \right)^2\frac { 1 } { 2 } ( ne^{-\varepsilon} p_i + \sqrt { c_1 e^{-\varepsilon}p_i n \ln n})}\\
		&\;\;\leq\;\;&  e^{-\frac{(\sqrt{c_1}-\sqrt{c_2})^2 \ln n}{4}}\;,
	\end{eqnarray*}
Similarly, we can show that $\prob(\{\hq_{i,1} \in  U(p_i;c_1,c_2), q_i < U(p_i;c_1,c_1)\}) \leq e^{ - ( \sqrt{c_1} - \sqrt {c_2})^2 \ln n/3 }$.

\subsubsection{Proof of Lemma~\ref{lem:Qopt_key}}
\label{sec:proof_Qopt_key}
	
Let $\Delta = (c_1 \ln n )/n$. We divide the analysis into two regimes. 

\bigskip\noindent{\bf Case 1: $p \leq e^\varepsilon\Delta$ and $q\in U(p,c_1,c_1) = [0,2\Delta]$.}

First we analyze the bias. As we apply the universally minimum variance unbiased estimator (MVUE) to $D_K (q;p)$, the bias is entirely due to the functional approximation. 
Recall that we consider the best polynomial approximation $H_K(y)$ of function
$g(y)= [p- e^\varepsilon 2\Delta y  ]^+$ on $[0,1]$ with order $K$, i.e. $H_K(y) = \argmin_{P \in {\rm poly_K}} \max_{y'\in [0,1]} \big|\, g(y')-P(y') \,\big| $.  Denote it as $H_K(y)=\sum_{j=0}^Ka_jy^j$. Then $D_K(x;p) = H_K({x}/({2\Delta}))$.
It follows from Lemma \ref{lem:delf2}  that there exists a 
universal constant $M\geq 0$ such that for all $K\geq 1$, 
\begin{eqnarray}
	\sup_{x\in[0,2\Delta]} \big|\, D_K(x;p) - [p-e^\varepsilon x ]^+ \,\big| &\;\;=\;\;& \sup_{y\in[0,1]} \big|\, D_K(2\Delta y;p) - [p-e^\varepsilon 2\Delta y ]^+ \,\big| \nonumber \\
	&=& \sup_{y\in [0,1]}|H_K(y) -g(y)| \nonumber\\
		&\leq& M\left(p \wedge \frac{1}{K}\sqrt{\frac{e^\varepsilon c_1  p \ln n}{n}}\right) \;. 
\end{eqnarray}

Next to analyze the variance, we upper bound the
 magnitude of the coefficients in $\tD_K$ using Lemma~\ref{lem:Qopt_magnitude}, 
 and upper bound the second moment of the unique MVUE using the tail bound of Poisson distribution 
 in Lemma~\ref{lem:Qopt_moment}.
As the universal MVUE is of the form 
$\sum_{j=0}^K a_j (2\Delta)^{-j} \prod_{k=0}^{j-1} (\hp- k/n )$ as shown in Section \ref{sec:MVUE}, 
The variance is upper bounded by 
\begin{eqnarray}
	\var(\, \tD_K(\hq;p)\,) &\;\;=\;\;& \var \Big( \sum_{j=0}^K a_j (2\Delta)^{-j} \prod_{k=0}^{j-1} (\hq-\frac{k}{n}) \Big)  \nonumber\\
		&\leq &\Big(\sum_{j=0}^K |a_j| (2\Delta)^{-j}\Big(\var\big(\prod_{k=0}^{j-1} (\hq-\frac{k}{n})\big)\Big)^\frac{1}{2} \Big)^2 \nonumber\\
		&\leq & \max_{0\leq j' \leq K} |a_{j'}|^2 \; \Big(\sum_{j=1}^K (2\Delta)^{-j}\Big(4\Delta q\Big)^\frac{j}{2} \Big)^2 \label{eq:Qopt_var_1} \\
		& = &\max_{0\leq j' \leq K} |a_{j'}|^2 \; \Big(\sum_{j=1}^K \big(\frac{q}{\Delta}\big)^\frac{j}{2} \Big)^2 \nonumber\\
		& = &\max_{0\leq j' \leq K} |a_{j'}|^2 \, \frac{q}{\Delta} \, \Big(\sum_{j=0}^{K-1} \big(\frac{q}{\Delta}\big)^\frac{j}{2} \Big)^2 \nonumber\\
		&\leq &\max_{0\leq j' \leq K} |a_{j'}|^2\, \frac{q}{\Delta}\, \Big(\sum_{j=0}^{K-1} 2^\frac{j}{2} \Big)^2 \nonumber\\
		&\lesssim & e^{2\varepsilon}B^K\Delta^2\frac{q}{\Delta} \label{eq:Qopt_var_2}\\
		&\lesssim &   B^K\frac{e^{\varepsilon} c_1 \ln n}{n}(e^\varepsilon q) \;, \nonumber
\end{eqnarray}
where $B$ is some universal constant, 
we use $q\leq 2\Delta$ and $\Delta=(c_1 \ln n) / n$ in the last inequality, 
and \eqref{eq:Qopt_var_1} follows from Lemma~\ref{lem:Qopt_moment}, 
\eqref{eq:Qopt_var_2}  from Lemma~\ref{lem:Qopt_magnitude}. 
Concretely, it follows from Lemma \ref{lem:Qopt_moment}  
that for $M_2=\max\big\{2n\Delta, K\big\}=2n\Delta$, 
\begin{eqnarray*}
	\var\big(\prod_{k=0}^{j-1} (\hq-\frac{k}{n})\big)
	\;\leq\; \E \big(\prod_{k=0}^{j-1} (\hq-\frac{k}{n})\big)^2
	\;\leq \; \big(\frac{2M_2q}{n}\big)^j
	\;=\; \big(4\Delta q\big)^j\;.
\end{eqnarray*}

To apply Lemma \ref{lem:Qopt_magnitude}, 
we first transform the domain of the polynomial approximation to be symmetric around the origin by change of variables. 
We consider $H_K(z^2) = \sum_{j=0}^Ka_jz^{2j}$, 
which is a polynomial with degree no more than $2K$ and satisfies 
\begin{eqnarray*}
	\sup_{z\in [-1,1]}|H_K(z^2)|\leq M_1 e^\varepsilon\Delta\;.
\end{eqnarray*}
This bound follows from
a triangular inequality applied to 
$\max_{y\in [0,1]}|g(y)|\lesssim e^\varepsilon\Delta$ and $\sup_{y\in [0,1]}|H_K(y) -g(y)|\lesssim e^\varepsilon\Delta$. 
It follows that  there exists a
 universal constant $M_1>0$ such that $\sup_{y\in [0,1]}|H_K(y)|\leq M_1 e^\varepsilon\Delta$. 
It follows from Lemma \ref{lem:Qopt_magnitude} that for all $0\leq j\leq K$,
\begin{eqnarray}
	|a_j|\leq M_1 e^\varepsilon\Delta(\sqrt{2}+1)^{2K}\;.
\end{eqnarray}

\bigskip\noindent{\bf Case 2: $p > e^\varepsilon\Delta$ and $q\in [e^{-\varepsilon}p-\sqrt{e^{-\varepsilon}p\Delta}, e^{-\varepsilon}p+\sqrt{e^{-\varepsilon}p\Delta}]$.} 

First, we analyze the bias. 
In this regime, 
we claim that the best polynomial approximation $D_K(x;p)$ of 
 $[p-e^\varepsilon x]^+$ is given by
\begin{eqnarray}
	D_K(x;p) \;\; =\;\; \frac{e^\varepsilon}{2}\sum_{j=0}^Kr_j \big(\sqrt{e^{-\varepsilon}p\Delta}\big)^{-j+1}(x-e^{-\varepsilon}p)^j+\frac{p-e^\varepsilon x}{2} \;, 
\end{eqnarray} 
where $r_j$'s are defined from  the best  polynomial approximation $R_K(y)$ of $g(y) = |y|$ on $[-1,1]$ with order $K$: 
$R_K(y)  =  \sum_{j=0}^K r_j y^j$\;.
And it is well known (e.g.~\cite[Chapter~9, Theorem~3.3]{devore1993constructive}) 
that there exists a universal constant $M_3$ such that $|R_K(y)-|y||  \leq  {M_3}/{K}$, for all  $y\in [-1,1]$.
As $[a]^+ = (1/2)a + (1/2)|a| $, 
the optimality of $D_K(x;q)$ follows from 
\begin{eqnarray*}
	\big|\, D_K(x;p)-[p-e^\varepsilon x]^+\,\big| 
	&\;\;= \;\;& \big|\frac{e^\varepsilon\sum_{j=0}^Kr_j \big(\sqrt{e^{-\varepsilon}p\Delta}\big)^{-j+1}(x-e^{-\varepsilon}p)^j-|p-e^\varepsilon x|}{2}\\
	&&+\frac{(p-e^\varepsilon x)-(p-e^\varepsilon x)}{2}\big|\\
	&=& \frac{e^\varepsilon}{2}\Big|\, \sum_{j=0}^Kr_j \big(\sqrt{e^{-\varepsilon}p\Delta}\big)^{-j+1}(x-e^{-\varepsilon}p)^j-| e^{-\varepsilon}p-x| \,\Big|\\
	&=&\frac{e^\varepsilon\sqrt{e^{-\varepsilon}p\Delta}}{2}\Big|R_k(\frac{x-e^{-\varepsilon}p}{\sqrt{e^{-\varepsilon}p\Delta}})-|\frac{e^{-\varepsilon}p-x}{\sqrt{e^{-\varepsilon}p\Delta}}|\Big|\;\\
	&=&\frac{\sqrt{e^{\varepsilon}p\Delta}}{2}\Big|R_k(y)-|y|\Big|\;,
\end{eqnarray*}
where we let $x = e^{-\varepsilon}p+y\sqrt{e^{-\varepsilon}p\Delta}$, and we want small approximation error in $y\in[-1,1]$. 
This gives the desired bound on the bias: 
\begin{eqnarray*}
	|D_K(x;p)-[p-e^\varepsilon x]^+| \;=\; \frac{\sqrt{e^{\varepsilon}p\Delta}}{2}\Big|R_k(y)-|y|\Big| 
\;\leq \; \frac{M_3\sqrt{e^{\varepsilon}p\Delta}}{K} \; \lesssim \; \frac{1}{K}\sqrt{\frac{e^\varepsilon \,p\,c_1\, \ln n}{n}}
\end{eqnarray*}

Next, we analyze the variance. 
Recall from Lemma~\ref{lem:Qopt_moment}
 that  $g_{j,c}(\hq)$ defined as
\begin{eqnarray}
	g_{j, c}(\hq) \;\; \triangleq \;\; \sum_{k=0}^{j}{j \choose k}(-c)^{j-k} \prod_{h=0}^{k-1}\left(\hq-\frac{h}{n}\right)
\end{eqnarray}
is the unique uniformly minimum variance unbiased estimator (MVUE) for $(q-c)^j$, $j\geq 0$, $j\in \mathbb{N}$.
Hence, 
\begin{eqnarray}
	\tD_K(x;p) = \frac{e^\varepsilon}{2}\Big (\sum_{j=0}^Kr_j \big(\sqrt{e^{-\varepsilon}p\Delta}\big)^{-j+1}g_{j,e^{-\varepsilon}p}(\hq)+ g_{1,e^{-\varepsilon}p}(\hq)\Big)
\end{eqnarray}
Let $a_j = r_j$ for $j=0,2,3,\ldots, K$ and $a_1 = r_1-1$ and we can write $\tD_K(x;p)$ as
\begin{eqnarray*}
	\tD_K(x;p) \; = \; \frac{e^\varepsilon}{2}\sum_{j=0}^Ka_j \big(\sqrt{e^{-\varepsilon}p\Delta}\big)^{-j+1}g_{j,e^{-\varepsilon}p}(\hq)\;.
\end{eqnarray*}
It is shown in \cite[Lemma~2]{cai2011testing} that $|r_j|\leq 2^{3K}$, for $0\leq j\leq K$. So we can safely say $\max_{0\leq j\leq K} |a_j|^2\leq 4\cdot 2^{6K} $. 
It follows from Lemma~\ref{lem:Qopt_moment}
that for $M_4 = \max\{\frac{n(q-e^{-\varepsilon}p)^2}{q}, K\}$, 
\begin{eqnarray*}
	\var(g_{j,e^{-\varepsilon}p}(\hq))\; \leq \; \E g_{j,e^{-\varepsilon}p}^2(\hq) \; \leq \; (\frac{2M_4q}{n})^j\;.
\end{eqnarray*}
Note that if $q=0$, the variance is $0$. We  consider the case $q\neq 0$. 
The variance is 
\begin{eqnarray}
	\var(\tD_K(x;p)) &\;\;=\;\;& \frac{e^{2\varepsilon}}{4}\var(\sum_{j=0}^Ka_j \big(\sqrt{e^{-\varepsilon}p\Delta}\big)^{-j+1}g_{j,e^{-\varepsilon}p}(\hq)) \nonumber \\
	&\leq & \frac{e^{2\varepsilon}}{4} \Big(\sum_{j=0}^K|a_j|\big(\sqrt{e^{-\varepsilon}p\Delta}\big)^{-j+1}\var^{\frac12}(g_{j,e^{-\varepsilon}p}(\hq)) \Big)^2 \nonumber \\
	&\leq & e^{2\varepsilon}2^{6K}e^{-\varepsilon}p\Delta\Big(\sum_{j=0}^K\big(\sqrt{e^{-\varepsilon}p\Delta}\big)^{-j}(\frac{2M_4q}{n})^{\frac{j}{2}} \Big)^2 \nonumber \\
	&=& e^{\varepsilon}2^{6K}p\Delta\Big(\sum_{j=0}^K(\frac{2M_4q}{ne^{-\varepsilon}p\Delta })^{\frac{j}{2}} \Big)^2 \nonumber \\
	&\leq & e^{\varepsilon}2^{6K}p\Delta \Big(\frac{c^{K+1}-1}{c-1}\Big)^2 \label{eq:Qopt_var_3} \\
	&\leq & \frac{c^{2}}{(c-1)^{2}}(8 c)^{2 K}e^{\varepsilon} p \Delta \nonumber \\
	&\lesssim & B^K \frac{c_1 \ln n}{n} e^\varepsilon p \nonumber \;,
\end{eqnarray}
where $c = \max\{\sqrt{2}, 2\sqrt{c_3/c_1}\} $, and $B>0$ is some universal constant as $c_3<c_1$. 
The inequality in \eqref{eq:Qopt_var_3} follows from $\sqrt{2M_4q / ({ne^{-\varepsilon}p\Delta })}\leq c$, which follows from  
\begin{eqnarray*}
	&&\sqrt{\frac{2Kq}{ne^{-\varepsilon}p\Delta }}
	\;\leq\;  \sqrt{\frac{2K(e^{-\varepsilon}p+\sqrt{e^{-\varepsilon}p\Delta})}{ne^{-\varepsilon}p\Delta }}
	\;\leq\; \sqrt{\frac{2K\cdot 2e^{-\varepsilon}p}{ne^{-\varepsilon}p\Delta }}
	\;=\; \sqrt{\frac{4c_3}{c_1}}\;, \text{ and }\\
	&&\sqrt{\frac{2q}{ne^{-\varepsilon}p\Delta }\cdot \frac{n(q-e^{-\varepsilon}p)^2}{q}}
	\; = \; \sqrt{\frac{2(q-e^{-\varepsilon}p)^2}{e^{-\varepsilon}p\Delta}} 
	\; \leq \;\sqrt{\frac{2e^{-\varepsilon}p\Delta}{e^{-\varepsilon}p\Delta}}
	\; = \; \sqrt{2}\;.
\end{eqnarray*}

\subsection{Proof of Theorem~\ref{thm:Qlb}} 
\label{sec:proof_Qlb}
	
Note that $d_\varepsilon(P\|Q)  = \sum_{i=1}^S [p_i-e^\varepsilon q_i]^+$ is well defined even if $Q$ does not sum to exactly one. 
Define a set of such approximate probability vectors as 
\begin{eqnarray}
	\cM_S(\zeta) \;=\; \left\{Q: \Big|\sum_{i=1}^Sq_i-1\Big| \leq \zeta \right\}\;. 
	\label{eq:def_appxprob}
\end{eqnarray}
Later in this section, 
we use the method of two fuzzy hypotheses from \cite{Tsy08} to show that 
for some 
$\chi \gtrsim  \sum_{j=1}^{S}p_j \wedge \sqrt{{e^{\varepsilon} p_j}/{(n \ln n) }}$ and $\chi\leq e^\varepsilon$, 
the estimation error exceeds $\chi/4$  with a strictly positive probability, 
under a minimax setting over the approximate probability class 
$\cM_S(\zeta)$ with  $\zeta=\chi/(10 e^\varepsilon)$:
\begin{eqnarray}
	\inf_{\hd_\varepsilon(P\|Q_n)}\sup_{Q\in \cM_S(\chi/(10 e^\varepsilon))}P\left(\big|\hd_\varepsilon(P\|Q_n)- d_\varepsilon(P\|Q)\big|\geq \frac{\chi}{4} \right) \;\; \geq \;\; \frac13 \;,
	\label{eq:Qlb_key} 
\end{eqnarray}
for a sufficiently large $n$, 
where we extend the definition of $Q_n$ to be 
 Poisson sampling each alphabet with the appropriate rate. 
This gives a lower bound on the minimax risk for $\zeta=\chi/(10e^\varepsilon)$:
\begin{eqnarray}
	R(S,n,P,\zeta )& \triangleq & \inf_{\hd_\varepsilon(P\|Q_n)}\sup_{Q\in \cM_S(\zeta) }\E_Q\left[\Big(\hd_\varepsilon(P\|Q_n)-d_\varepsilon(P\|Q)\Big)^2 \right]
	\label{eq:def_R}\\
		&\;\;\geq\;\;&
		\frac{\chi^2}{16} \left\{ \inf_{\hd_\varepsilon(P\|Q_n)}\sup_{Q\in \cM_S (\zeta)}Q\left(\big|\hd_\varepsilon(P\|Q_n)-d_\varepsilon(P\|Q)\big|\geq \frac{\chi}{4} \right) \right\} \nonumber\\
		&\;\;\geq\;\;&  \frac{\chi^2}{48}\;,	\nonumber
\end{eqnarray}
As our goal is to prove a lower bound on the minimax error, which is $R(S,n,P,0)$, 
we use the following lemma. We provide a proof in Section~\ref{sec:proof_zeta}. 
\begin{lemma}
	\label{lem:zeta}
	For any $S,n \in \mathbb{N}_+$, $0<\zeta<1$, any distribution $P\in \cM_S$, and 
	any $\varepsilon>0$ that defines the quantity $d_\varepsilon(\cdot\|\cdot)$ used in the definition of $R(\cdot)$ in \eqref{eq:def_R},  
	we have
	\begin{eqnarray}
		R(S,n(1-\zeta)/4,P,0) \;\;\geq\;\; \frac 1 4R(S,n,P,\zeta)-\frac{1}{2}e^{-n(1-\zeta)/8}-\frac{1}{2}e^{2\varepsilon}\zeta^2\;.
	\end{eqnarray}
\end{lemma}
This implies that for our choice of $\zeta = \chi/(10e^\varepsilon)$, 
	\begin{eqnarray*}
		R(S,n(1-\zeta)/4, P,0) &\;\;\geq\;\;&
		\frac 1 4R(S,n,P,\zeta)-\frac{1}{2}e^{-n(1-\zeta)/8}-\frac{1}{2}e^{2\varepsilon}\zeta^2\;\\
		&\;\;\geq\;\;& \frac{\chi^2}{192} - \frac{1}{2}e^{-n(1-\chi/(10e^\varepsilon))/8}-\frac{\chi^2}{200}\\
		&\;\;\gtrsim \;\;& \chi^2\\
		&\;\;\gtrsim \;\;&\left( \sum_{j=1}^{S}p_j \wedge \sqrt{\frac{e^{\varepsilon} p_j}{n \ln n }}\right)^2\;,
	\end{eqnarray*}
	where 
$\zeta\leq 1/10$, which follows from $\chi\leq e^\varepsilon$, this proves the desired theorem.

Now, we are left to prove Eq.~\eqref{eq:Qlb_key}, by applying the following Lemma from \cite{Tsy08}. 
The idea is to construct two fuzzy hypotheses, such that 
they are sufficiently close to each other (as measured by total variation) 
to be challenging, while sufficiently separated in  $d_\varepsilon$.  
Translating the theorem into our context, we get the following corollary. 
\begin{lemma}[Corollary of {\cite[Theorem 2.15]{Tsy08}}]
	\label{lem:fuzzy}
	For any $s>0$, $\zeta>0$, $0\le \beta_0,\beta_1<1$, $\lambda \in \reals$, 
	if there exists two distributions $\sigma_0$ and $\sigma_1$ on $Q=[q_1,\ldots,q_S] \in \cM_S(\zeta)$ such that 
	\begin{eqnarray}
		\sigma_0(Q\,:\, d_\varepsilon(P\|Q) \leq \lambda - s ) &\ge& 1-\beta_0 \;, \label{eq:sep0} \\
		\sigma_1(Q\,:\, d_\varepsilon(P\|Q) \geq \lambda + s ) &\ge& 1-\beta_1 \;, \label{eq:sep1} 
	\end{eqnarray}
	and $D_{\rm TV}(F_1,F_0) \leq \eta <1 $, then 
	\begin{eqnarray}
		\inf_{\hd_\varepsilon(P\|Q_n) } \sup_{Q \in \cM_S(\zeta) } \prob_Q \Big( \,| \hd_\varepsilon(P\|Q_n) - d_\varepsilon(P\|Q) | \ge s \, \Big) \;\; \ge \;\; \frac{1-\eta-\beta_0-\beta_1}{2}\;\; \;,
	\end{eqnarray}
	where $F_i$ is the marginal distribution of $Q_n$ 
	given the prior $\sigma_i$ for $i\in\{0,1\}$. 
\end{lemma}
We construct two hypotheses, satisfying the assumptions with 
choices of $s = \chi/4 \gtrsim \sum_{j=1}^{S}p_j \wedge \sqrt{{e^{\varepsilon} p_j}/{(n \ln n) }}$ and  $\eta,\beta_0,\beta_1=o(1)$ such that 
Eq.~\eqref{eq:Qlb_key} follows. 
We will first introduce the construction,  
check the separation conditions in Eqs.~\eqref{eq:sep0} and \eqref{eq:sep1}, 
and check the total variation condition. 

\bigskip\noindent
{\bf Constructing two prior distributions.}
Fix the distribution $P\in \cM_S$, and assume $p_S = \min_{1\leq i \leq S}p_i$.
Let $\Bmu_0$, $\Bmu_1$ be two prior distributions on the parameter $Q$ where $Q_n$ will be drawn from, and set 
\begin{eqnarray}
	\Bmu_{0}&\;\;=\;\;&\mu_{0}^{(p_{1})} \otimes \mu_{0}^{(p_{2})} \otimes \ldots \otimes \mu_{0}^{(p_{S-1})} \otimes \delta_{1-\gamma}\;,\\
	\Bmu_{1}&\;\;=\;\;&\mu_{1}^{(p_{1})} \otimes \mu_{1}^{(p_{2})} \otimes \ldots \otimes \mu_{1}^{(p_{S-1})} \otimes \delta_{1-\gamma}\;,
\end{eqnarray}
where 
\begin{eqnarray}
	\gamma \;\;=\;\; \sum_{j : p_{j} \leq \frac{c e^\varepsilon\ln n}{n}} \frac{p_{j}}{e^\varepsilon D}+\sum_{j : p_{j}>\frac{c e^\varepsilon \ln n}{n}} e^{-\varepsilon} p_{j}\;,\label{eq:gammadef}
\end{eqnarray} and $c\in (0,1)$ is a constant, $D$ is the universal constant in Lemma~\ref{lem:tvbound}. 
Note that this does not produce a valid probability distribution, as it will not sum to one almost surely. 
However, this is sufficient as we can bound the difference in the minimax rate between exact and approximate probability distributions 
using Lemma~\ref{lem:zeta}.
This choice of $\gamma$ ensures that the sum $\sum_{j=1}^S q_j $ concentrates 
around one.
For a $p\in (0,1)$, we construct $\mu_i^{(p)}$, $i\in \{0,1\}$ depending on $p$ in two separate cases.
Our goal is to construct two prior distributions, which match in the first $L$ degree moments 
(such that the marginal total variation distance is sufficiently small), 
but at the same time sufficiently different in estimation of $d_\varepsilon$,  
such that they differ approximately as much as the resolution of 
the best polynomial function approximation. 


\bigskip\noindent{\bf Case 1: $p> ( c e^\varepsilon\ln n)/{n}$, for some constant $c\in (0,1)$.} 
	Define function $g(x) = e^{-\varepsilon}p+(\sqrt{({ce^{-\varepsilon} p\ln n})/{n}}) x$, where $x\in [-1,1]$. Let $\nu_i,i=0,1$ be two measures constructed in Lemma~\ref{lem:moment matching relu}. 
\begin{lemma}[ {\cite[Lemma~1]{cai2011testing}}]
\label{lem:moment matching relu}
	For any positive integer $L>0$, there exists two probability measure $\nu_0$ and $\nu_1$ on $[-1,1]$ such that
	\begin{enumerate}
		\item $\int t^{l} \nu_{1}(d t)\;\;=\;\;\int t^{l} \nu_{0}(d t)$, for all $l=0,1,2, \ldots, L$;
		\item $\int [-t]^+ \nu_{1}(d t) - \int [-t]^+ \nu_{0}(d t) \;\;=\;\;\Delta_L[[-t]^+;[-1, 1]]$,
	\end{enumerate}
	where $\Delta_L[[-t]^+;[-1, 1]]$ is the distance in the uniform norm on $[-1,1]$ from the function $[-t]^+$ to the space 
	of polynomial functions of degree $L$: ${\rm poly}_L$.
\end{lemma}
	We define two new measures $\mu_i^{(p)},i=0,1$ on $[e^{-\varepsilon}p-\sqrt\frac{ce^{-\varepsilon} p\ln n}{n},e^{-\varepsilon}p+\sqrt\frac{ce^{-\varepsilon} p\ln n}{n}]$ by $\mu_i^{(p)}(A) = \nu_i(g^{-1}(A))$. 
	Note that we need the lower bound on $p$ to ensure that this is non-negative. 
	Let $L = d_2\ln n, d_2>1$. It follows that
	\begin{enumerate}
	\item
		\begin{eqnarray}
			\int t \mu_{1}^{(p)}(d t) \;=\; \int t \mu_{0}^{(p)}(d t) = e^{-\varepsilon}p\; ;\label{eq:case2meanbound}
		\end{eqnarray}
	\item
		\begin{eqnarray}
			\int t^l \mu_{1}^{(p)}(d t) \;=\; \int t^l \mu_{0}^{(p)}(d t), \quad \forall l = 2\ldots,L+1;
		\end{eqnarray}
	\item
		\begin{eqnarray}
			\int [p-e^\varepsilon t]^+ \mu_{1}^{(p)}(d t)-\int [p-e^\varepsilon t]^+ \mu_{0}^{(p)}(d t) &\;=\;& \sqrt{\frac{ce^{\varepsilon} p\ln n}{n}} \, \Delta_L[[-t]^+;[-1, 1]]\label{eq:case2sep}\nonumber\\
			 &\;\gtrsim \;& p \wedge \sqrt{\frac{ce^{\varepsilon} p}{d_2^2n \ln n }}\;.\label{eq:case2estbound}
		\end{eqnarray}
	\end{enumerate}
	The last inequality follows from the following lemma, with a choice of $L=d_2 \ln n$ for some constant $d_2$.  
\begin{lemma}[{\cite{bernstein1914meilleure} }]
	For $L>1$,
	\begin{eqnarray}
		\Delta_L[[-t]^+;[-1, 1]] \;\;=\;\; \Delta_L[|t|;[-1, 1]]\;\;=\;\;\beta_*L^{-1}(1+o(1)) \;\;\asymp \;\; \frac{1}{L}\;,
	\end{eqnarray}
\end{lemma}
where $\beta_*  \approx 0.2802$ is the Bernstein constant.

\bigskip\noindent{\bf Case 2: $0<p\leq (c e^\varepsilon\ln n)/{n}$, for some constant $c\in (0,1)$.} 
 When $p$ is too close to zero, 
directly applying  the above strategy only gives a lower bound on the difference in $d_\varepsilon$ under the two hypotheses that scales only as $p/\ln n$, 
and not as $\sqrt{p/(n\ln n)}$ as desired. 
Instead, we construct an approximation of $f(x;a)=([a-e^\varepsilon x  ]^+-a)/(e^\varepsilon x) $.

Our strategy is to first construct two prior distributions $\tnu_i^{\eta,a}$'s on 
$\{0\} \cup [\eta,1]$ which are difference in estimating 
$f(x;a) ={([a-e^\varepsilon x  ]^+-a)}/(e^\varepsilon x) $ (instead of $[a-e^\varepsilon x ]^+$). 
The non-smoothness of $f(x;a)$ near zero allows one to 
control the hardness of this estimation by choosing $\eta$, 
while ensuring the non-negativity of the resulting random variable $p$ drawn from $\mu_i$'s 
and also the expectation is close to $q$. 
Concretely, we let 
$\mu_i^{(p)}$ to be a measure on $\{0\}\cup [p/(De^\varepsilon),M]$, 
where $g(x) = Mx$ and $\mu_i^{(q)} = \tnu_{0}^{\eta, a}(g^{-1}(A))$. 
We first construct two new probability measures $\tnu_i^{\eta, a}$, $i=0,1$ from the two probability measures $\nu_i^{\eta, a}$, $i=0,1$  constructed in Lemma~\ref{lem:moment matching measures}. 

	\begin{lemma}
\label{lem:moment matching measures}
	Let $f(x;a) \triangleq {([a-e^\varepsilon x]^+-a)}/({e^\varepsilon x})$, $a\in [0,1]$. For any $0<\eta<1$ and positive integer $L>0$, there exists two probability measure $\nu_0^{\eta, a}$ and $\nu_1^{\eta, a}$ on $[\eta,1]$ such that
	\begin{enumerate}
		\item $\int t^{l} \nu_{1}^{\eta, a}(d t)\;\;=\;\;\int t^{l} \nu_{0}^{\eta, a}(d t)$, for all $l=0,1,2, \ldots, L$;
		\item $\int f(t;a) \nu_{1}^{\eta, a}(d t) - \int f(t;a) \nu_{0}^{\eta, a}(d t) \;\;=\;\; \Delta_L[f(x;a);[\eta, 1]]$,
	\end{enumerate}
	where $\Delta_L[f(x;a);[\eta, 1]]$ is the distance in the uniform norm on $[\eta,1]$ from the function $f(x;a)$ to the space 
	of polynomial functions of degree $L$: ${\rm poly}_L$.
\end{lemma}
	We construct $\tnu_i^{\eta,a}$ by scaling down $\nu_i^{\eta,a}$ and putting the remaining probability mass on zero. 
	This ensures that the restriction on $[\eta,1]$ of $\tnu_i^{\eta, a}$ is absolutely continuous with $\nu_i^{\eta, a}$, 
	and we construct the Radon-Nikodym derivative to be 
	\begin{eqnarray}
		\frac{d \tnu_i^{\eta, a}}{d nu_i^{\eta, a}}\;\;=\;\;\frac{\eta}{t} \leq 1, \quad t \in[\eta, 1]\;,
	\end{eqnarray}
	and $\tnu_i^{\eta, a}(\{0\}) = 1- \tnu_i^{\eta, a}([\eta,1])\geq 0$. 
	This choice of scaling ensures that 
	
	It follows that $\tnu_i^{\eta, a}, i=0,1$ are probability measures on $[0,1]$ that satisfy the following properties
	\begin{enumerate}
	\item $\int t \tnu_{1}^{\eta, a}(d t)\;\;=\;\;\int t \tnu_{0}^{\eta, a}(d t) \;\;=\;\; \eta$;
	
	\item $\int t^{l} \tnu_{1}^{\eta, a}(d t)\;\;=\;\;\int t^{l} \tnu_{0}^{\eta, a}(d t)$, for all $l=2, \ldots, L+1$;
	
	\item $\int [a-e^\varepsilon t ]^+ \tnu_{1}^{\eta, a}(d t) - \int [a-e^\varepsilon t ]^+ \tnu_{0}^{\eta, a}(d t) \;\;=\;\; \eta e^\varepsilon \Delta_L[f(x;a);[\eta, 1]]$.	
	\end{enumerate}

	Define 
	\begin{eqnarray}
		L=d_2 \ln n, \quad \eta=\frac{a}{D}, \quad a=\frac{p}{e^\varepsilon M}, \quad M=\frac{2 c \ln n}{n}\;,
	\end{eqnarray}
	where $D$ is a universal constant in Lemma~\ref{lem:tvbound} and $d_2>1$ is a constant.

\begin{lemma}
\label{lem:tvbound}
	Let $f(x;a) := \frac{[a-e^\varepsilon x]^+-a}{e^\varepsilon x}$, $a\in (0,\frac{1}{2}]$, $x\in [0,1]$, there exists universal constant $D$ such that
	\begin{eqnarray}
		\Delta_{L}\left[f(x ; a) ;[\frac{a}{D}, 1]\right]\;\;\gtrsim\;\;  \left(1\wedge \frac{1}{L\sqrt{a}}\right)\;.
	\end{eqnarray}
\end{lemma}
	Let $g(x) = Mx$ and let $\mu_i^{(p)}$ be the measure on $[0,M]$ defined by $\mu_i^{(p)} = \tnu_{0}^{\eta, a}(g^{-1}(A))$. Then we have $\mu_i^{(p)} (A) = M \tnu_{0}^{\eta, a}(A)$. 
	It then follows that
	\begin{enumerate}
	\item
		\begin{eqnarray}
			\int t \mu_{1}^{(p)}(d t) \;=\; \int t \mu_{0}^{(p)}(d t) = \frac{p}{e^\varepsilon D}\; ;\label{eq:case1meanbound}
		\end{eqnarray}
	\item
		\begin{eqnarray}
			\int t^l \mu_{1}^{(p)}(d t) \;=\; \int t^l \mu_{0}^{(p)}(d t), \quad \forall l = 2\ldots,L+1;
		\end{eqnarray}
	\item
		\begin{eqnarray}
			\int [p-e^\varepsilon t]^+ \mu_{1}^{(p)}(d t)-\int [p-e^\varepsilon t]^+ \mu_{0}^{(p)}(d t) &\;=\;& 
			\eta \, e^\varepsilon \, M \, \Delta _L[f(x;a);[\frac a D, 1]]\label{eq:case1sep}\\
			 &\;\gtrsim \;& p \wedge \sqrt{\frac{ce^\varepsilon p}{d_2^2n \ln n }}\;.\label{eq:case1estbound}
		\end{eqnarray}
	\end{enumerate}

\bigskip\noindent{\bf Separation conditions.}
	In both cases, since we set $q_S = 1-\gamma$, which is defined in Eq.~\eqref{eq:gammadef}, it follows from Eq.~\eqref{eq:case1meanbound} and \eqref{eq:case2meanbound} that
	\begin{eqnarray}
		\E_{\Bmu_0}\left[\sum_{j=1}^nq_j\right] \;\;=\;\; \E_{\Bmu_1}\left[\sum_{j=1}^nq_j\right] \;\;=\;\; 1\;.
	\end{eqnarray}
	Let 
	\begin{eqnarray}
		\chi &\;\;=\;\;& \E_{\Bmu_1}\left[d_\varepsilon(P\|Q)\right]-\E_{\Bmu_0}\left[d_\varepsilon(P\|Q)\right]\;,\\
		\zeta &\;\;=\;\;& \frac{\chi}{10 e^\varepsilon }\;.
	\end{eqnarray}
	
	We know from Eq.~\eqref{eq:case1estbound} and \eqref{eq:case2estbound} that, by construction, the estimates are separated in expectation:
	\begin{eqnarray*}
		\chi &\;\;\gtrsim\;\; & \sum_{j=1}^{S-1}p_j \wedge \sqrt{\frac{ce^{\varepsilon} p_j}{d_2^2n \ln n }}\\
		&\;\; \geq \;\;&(1-\frac 1 S) \sum_{j=1}^{S}p_j \wedge \sqrt{\frac{ce^{\varepsilon} p_j}{d_2^2n \ln n }}\\
		&\;\;\gtrsim\;\; & \sum_{j=1}^{S}p_j \wedge \sqrt{\frac{ce^{\varepsilon} p_j}{d_2^2n \ln n }}\;,
	\end{eqnarray*}
	where the second inequality follows from the assumption that $p_S = \min_{1\leq j \leq S}p_j$.
	To show concentration of $\E_{\Bmu_i}[d_\varepsilon(P\|Q)]$ around its mean, for $i=0,1$, we introduce the events
	\begin{eqnarray}
		E_i &\;\;=\;\;& \cM_S(\zeta)\cap \left\{Q:\Big|d_\varepsilon(P\|Q)-\E_{\Bmu_i} \big[d_\varepsilon(P\|Q)\big]\Big|\leq \frac\chi 4\right\}\;.
	\end{eqnarray}
	Introduce
	\begin{eqnarray*}
		F(P) &\;\;=\;\;& \sum_{j:p_j\leq \frac{ce^\varepsilon \ln n}{n}}\left(\frac{2 c\ln n}{n}\right)^2+\sum_{j:p_j>\frac{c e^\varepsilon \ln n}{n}}\left(\frac{4c e^{-\varepsilon}p_j\ln n}{n}\right)\\
		&\;\;\leq\;\;& \frac{4c^2S\ln^2n}{n^2}+\frac{4ce^{-\varepsilon}\ln n}{n}\;.
	\end{eqnarray*}
	It follows from the union bound and Hoeffding bound that 
	\begin{eqnarray*}
		\Bmu_i(E_i^c) &\;\;\leq\;\;& \Bmu_i\left(\Big|\sum_{j=1}^S q_j-1\Big|>\zeta\right)+\Bmu_i\left( \Big|d_\varepsilon(P\|Q)-\E_{\Bmu_i} \big[d_\varepsilon(P\|Q)\big]\Big|> \frac\chi 4 \right)\;\\
		&\;\;\leq\;\;& 2\exp\left(-\frac{2\zeta^2}{F(P)}\right)+2\exp\left(-\frac{\chi^2}{8F(P)}\right)\;.
	\end{eqnarray*}
	
	Then we choose parameter $c$ such that $\Bmu_i(E_i^c)$ can be made arbitrarily small. 
	Since we assumed $c\in (0,1)$ and $d_2>1$, and 
	from the assumption that $\sum_{j=1}^{S}p_j \wedge \sqrt{\frac{e^{\varepsilon} p_j}{n \ln n }}\geq C'\left(\sqrt{\frac{e^{\varepsilon}\ln n}{n}}+\frac{e^\varepsilon\sqrt{S}\ln n}{n}\right)$, 
	we have
	\begin{eqnarray*}
		\zeta &\;\;\asymp\;\;& \frac{\chi}{e^\varepsilon}\\
		&\;\;\gtrsim \;\;&  \frac{1}{e^\varepsilon}\sum_{j=1}^{S}p_j \wedge \sqrt{\frac{ce^{\varepsilon} p_j}{d_2^2n \ln n }}\\
		&\;\;\geq\;\;& \frac{\sqrt c}{e^\varepsilon d_2}\sum_{j=1}^{S}p_j \wedge \sqrt{\frac{e^{\varepsilon} p_j}{n \ln n }}\\
		&\;\;\geq\;\;& \frac{\sqrt{c}}{d_2} C'\left(\sqrt{\frac{e^{-\varepsilon}\ln n}{n}}+\frac{\sqrt{S}\ln n}{n}\right)\\
		&\;\;\gtrsim\;\;& \frac{1}{d_2} C'\left(\sqrt{\frac{ce^{-\varepsilon}\ln n}{n}}+\frac{c\sqrt{S}\ln n}{n}\right)\\
		&\;\;\gtrsim\;\;& \frac{1}{d_2} C'\sqrt{F(P)}\;.
	\end{eqnarray*}
	
	Hence, it suffices to take $C'$ large enough to ensure $\Bmu_i(E_i^c)$, is as small as we desire for $i=0,1$. 
	So with this constant $C'$, we have
	\begin{eqnarray}
	\mu_0 \left( d_\varepsilon(P\|Q)\leq 
	\E_{\Bmu_0}\big[d_\varepsilon(P\|Q)\big]+\frac{\chi}{4} \right) & \;\; \geq \;\; & 1 - \beta_0 \;,\\
	\mu_1 \left( d_\varepsilon(P\|Q)\geq 
	\E_{\Bmu_1}\big[d_\varepsilon(P\|Q)\big]-\frac{\chi}{4} \right) & \;\; \geq \;\; & 1 - \beta_1 \;,
	\end{eqnarray}
	fo any constant $\beta_0$ and $\beta_1$, 
	which satisfy the conditions  of  Lemma~\ref{lem:fuzzy} and $s = {\chi}/{4}$.

	\bigskip\noindent
{\bf Total variation condition.}	
	Let $G_i$ be marginal distribution of $(X_1,X_2,\ldots, X_S)$ under priors $\Bmu_i$ for $i=0,1$.
	Denote by $\pi_i$ the probability measures defined as
	\begin{eqnarray}
		\pi_i(A)=\frac{\Bmu_i\left(E_{i} \cap A\right)}{\Bmu_i\left(E_{i}\right)}, \quad i=0,1\;.
	\end{eqnarray}
	Let $F_i$ be marginal distribution of $(X_1,X_2,\ldots, X_S)$ under priors $\pi_i$ for $i=0,1$.

	Triangle inequality of  total variation yields
	\begin{eqnarray*}
		\TV(F_0, F_1) &\;\;\leq\;\;& \TV(F_0,G_0)+\TV(G_0,G_1)+\TV(G_1,F_1)\\
		&\;\;=\;\;& \sup_A\Big|\Bmu_0(A)-\frac{\Bmu_0(A\cap E_0)}{\Bmu_0(E_0)}\Big|+\TV(G_0,G_1)+\sup_A\Big|\Bmu_1(A)-\frac{\Bmu_1(A\cap E_1)}{\Bmu_1(E_1)}\Big|\\
		&\;\;=\;\;& \TV(G_0,G_1)+\Bmu_0(E_0^c)+\Bmu_1(E_1^c)\;.
	\end{eqnarray*}
	
	In view of fact that $\TV\left(\otimes_{i=1}^{S} P_i, \otimes_{i=1}^{S} Q_i\right) \leq \sum_{i=1}^{S} \TV\left(P_i, Q_i\right)$, we have
	\begin{eqnarray*}
		\TV(G_0,G_1) &\;\;\leq\;\;& \sum_{i=1}^{S-1}\TV(\mu_0^{(p_i)},\mu_1^{(p_i)})\\
		&\;\;\leq\;\;&   \sum_{i=1}^{S-1}2\left(\frac 1 2 \right)^{d_2\ln n}\\
		&\;\;\leq\;\;& \frac{2S}{n^{d_2 \ln 2}}\;,\\
	\end{eqnarray*}
	where in the second inequality we applied the following lemmas, 
	assumming $d_2$ to be large enough. 
	\begin{lemma}[{\cite[Lemma~3]{WY16} when $q_i \leq c e^\varepsilon\ln n/n$}] 
	Suppose $U_0$, $U_1$ are two random variables supported on $[0,M]$, where $M\geq 0$ is constant. Suppose $\E[U_0^j]=\E[U_1^j]$, $0\leq j\leq L$. Denote the marginal distribution of $X$ where $X|\lambda\sim \Poi(\lambda)$, $\lambda\sim U_i$ as $F_i$ for $i = 0,1$. If $L>2eM$, then
	\begin{eqnarray}
		\TV(F_0,F_1)\;\;\leq\;\; \left(\frac{2eM}{L}\right)^L\;.
	\end{eqnarray}
	\end{lemma}
	
	\begin{lemma}[{\cite[Lemma~32]{JYT18} when $q_i > c e^\varepsilon\ln n/n$}]
	Suppose $U_0$, $U_1$ are two random variables supported on $[a-M, a+M]$, where $a\geq M\geq 0$ are constants. Suppose $\E[U_0^j]=\E[U_1^j]$, $0\leq j\leq L$. Denote the marginal distribution of $X$ where $X|\lambda\sim \Poi(\lambda)$, $\lambda\sim U_i$ as $F_i$ for $i = 0,1$. If $L+1\geq (2eM)^2/a$, then
	\begin{eqnarray}
		\TV(F_0,F_1)\;\;\leq \;\; 2\left(\frac{eM}{\sqrt{a(L+1)}}\right)^{L+1}\;.
	\end{eqnarray} 
	\end{lemma}
	 Since there exists a constant $C>0$ such that $\ln n \geq C\ln S$, $ S\geq 2$, we can conclude that with chosen parameters, $T(F_0,F_1) = o(1)$.


\subsubsection{Proof of Lemma~\ref{lem:zeta}}
\label{sec:proof_zeta}
	
We define minimax risk under the multinomial sampling model for a fixed $P$ as 
	\begin{eqnarray}
		R_B(S,n,P) \;\; \triangleq \;\; \inf_{\hd_\varepsilon(P\|Q_n)}\sup_{Q\in \cM_S}\E_Q\left[\Big(\hd_\varepsilon(P\|Q_n)-d_\varepsilon(P\|Q)\Big)^2 \right]\;.
	\end{eqnarray}
	Let $\hat{T} = \hat{T}(X_1,X_2,\ldots,X_S)$ be a near-minimax estimator under multinomial model such that for every sample size $n$,
	\begin{eqnarray}
		\sup_{Q\in \cM_S}\E_Q\left[\Big(\hat{T}-d_\varepsilon(P\|Q)\Big)^2 \right]< R_B(S,n,P)+\xi\;,
	\end{eqnarray}
	where $\xi>0$.

	For any $Q\in \cM_S(\zeta)$, let $\sum_{i=1}^S q_i = A$, we have
	\begin{eqnarray*}
		\left|d_\varepsilon\big(P\|\frac{Q}{\sum_{i=1}^S q_i}\big)-d_\varepsilon\big(P\|Q\big)\right|
		&\;\;\leq \;\;&\sum_{i=1}^S \left|[p_i-e^\varepsilon q_i/A]^+-[p_i-e^\varepsilon q_i]^+\right|\\
		&\;\;\leq \;\;&\sum_{i=1}^S e^\varepsilon q_i |1/A-1|\\
		&\;\;= \;\;&e^\varepsilon |1-A|\\
		&\;\;\leq \;\;& e^\varepsilon \zeta\;.
	\end{eqnarray*}
	Now we consider risk of  $\hat{T}$ for $Q\in \cM_S(\zeta)$ under Poisson sampling model, where $X_i$ are mutually independent with marginal distributions $X_i\sim \Poi(n q_i)$. Let $n'=\sum_{i=1}^SX_i$, we know $n'\sim \Poi( n \sum_{i=1}^S q_i)$. In view of fact that conditioned on $n'=m$, $(X_1,X_2,\ldots,X_S)$ follows multinomial distribution parameterized by $\left(m, \frac{Q}{\sum_{i=1}^Sq_i}\right)$, we have
	\begin{eqnarray}
		&\;\;&\E_Q\left[\Big(\hat{T}-d_\varepsilon(P\|Q)\Big)^2 \right]\\
		&\;\;= \;\;& \E_Q\left[\Big(\hat{T}-d_\varepsilon\big(P\|\frac{Q}{\sum_{i=1}^S q_i}\big)+d_\varepsilon\big(P\|\frac{Q}{\sum_{i=1}^S q_i}\big)-d_\varepsilon(P\|Q)\Big)^2 \right] \\
		&\;\;\leq \;\;& 2\E_Q\left[\Big(\hat{T}-d_\varepsilon\big(P\|\frac{Q}{\sum_{i=1}^S q_i}\big)\Big)^2\right]+2\E_Q\left[\Big(d_\varepsilon\big(P\|\frac{Q}{\sum_{i=1}^S q_i}\big)-d_\varepsilon(P\|Q)\Big)^2 \right]\\
		&\;\;\leq \;\;& 2\E_Q\left[\Big(\hat{T}-d_\varepsilon\big(P\|\frac{Q}{\sum_{i=1}^S q_i}\big)\Big)^2\right]+2e^{2\varepsilon}\zeta^2\\
		&\;\; = \;\;& 2\sum_{m=0}^\infty\E_Q\left[\Big(\hat{T}-d_\varepsilon\big(P\|\frac{Q}{\sum_{i=1}^S q_i}\big)\Big)^2|n'=m\right]\prob(n'=m)+2e^{2\varepsilon}\zeta^2\\
		&\;\; \leq \;\;& 2\sum_{m=0}^\infty R_B(S,m,P)\prob(n'=m)+2(e^{2\varepsilon}\zeta^2+\xi)\\
		&\;\; \leq \;\;& 2\left(1\cdot \prob\big(n'\leq n(1-\zeta)/2\big)+R_B\big(S,n(1-\zeta)/2,P\big)\prob\big(n'\geq n(1-\zeta)/2\big)\right)\nonumber\\
		&&+2(e^{2\varepsilon}\zeta^2+\xi)\label{eq:lessthan1}\\
		&\;\; \leq \;\;& 2 R_B\big(S,n(1-\zeta)/2,P\big)+2\prob\big(n'\leq n(1-\zeta)/2\big)+2(e^{2\varepsilon}\zeta^2+\xi)\\
		&\;\; \leq \;\;& 2 R_B\big(S,n(1-\zeta)/2,P\big)+2\prob\big(\Poi(n(1-\zeta)/2)\leq n(1-\zeta)/2\big)+2(e^{2\varepsilon}\zeta^2+\xi)\\
		&\;\; \leq \;\;& 2 R_B\big(S,n(1-\zeta)/2,P\big)+2e^{-n(1-\zeta)/8}+2(e^{2\varepsilon}\zeta^2+\xi)\;,
	\end{eqnarray}
	where \eqref{eq:lessthan1} follows from $R_B(S,m,P)\leq 1$, and the last inequality follows from Lemma~\ref{lem:poisson_tail}.	
	Taking the supremum of $\E_Q\left[\Big(\hat{T}-d_\varepsilon(P\|Q)\Big)^2 \right]$ over $\cM_S(\zeta)$ and using the arbitrariness of $\zeta$, we have
	\begin{eqnarray}
		R(S,n,P,\zeta)\;\;\leq\;\; 2 R_B\big(S,n(1-\zeta)/2,P\big)+2e^{-n(1-\zeta)/8}+2e^{2\varepsilon}\zeta^2\;,
	\end{eqnarray}
	which is equivalent to
	\begin{eqnarray}
		R_B\big(S,n(1-\zeta)/2,P\big) \;\;\geq\;\; \frac{1}{2}R(S,n,Q,\zeta)-e^{-n(1-\zeta)/8}-e^{2\varepsilon}\zeta^2\;.
	\end{eqnarray}
	It follows from \cite[Lemma~16]{JKYT15} that $R_B\big(S,n,P\big)\leq 2R\big(S,n/2,P, 0\big)$. Hence,
	\begin{eqnarray}
		R(S,n(1-\zeta)/4,P,0) &\;\;\geq \;\;& \frac{1}{2} R_B(S,n(1-\zeta)/2,Q)\\
		&\;\;\geq \;\;& \frac{1}{4}R(S,n,P,\zeta)-\frac{1}{2}e^{-n(1-\zeta)/8}-\frac{1}{2}e^{2\varepsilon}\zeta^2\;.
	\end{eqnarray}

\subsubsection{Proof of Lemma~\ref{lem:moment matching measures}}
\label{sec:proof moment matching measures}
	By \cite[Lemma~31]{JYT18}, there are two probability measures $\nu_{1}^{\eta, a}$ and $\nu_{0}^{\eta, a}$ on $[\eta,1]$ such that
	\begin{enumerate}
		\item 
		\begin{eqnarray*}
		\int t^{l} \nu_{1}^{\eta, a}(d t)\;\;=\;\;\int t^{l} \nu_{0}^{\eta, a}(d t)\text{, for all }l=0,1,2, \ldots, L\;.
		\end{eqnarray*}
		
		\item 
		\begin{eqnarray*}
			\int \frac{|e^\varepsilon x-a|-a}{e^\varepsilon x} \nu_{1}^{\eta, a}(d t) - \int\frac{|e^\varepsilon x-a|-a}{e^\varepsilon x} \nu_{0}^{\eta, a}(d t) \;\;=\;\;2 \Delta_L\left[\frac{|e^\varepsilon x-a|-a}{e^\varepsilon x};[\eta, 1]\right]\;,
		\end{eqnarray*}
		which is equivalent to 
		\begin{eqnarray*}
			\int \frac{2[a-e^\varepsilon x]^+-2a}{e^\varepsilon x} \nu_{1}^{\eta, a}(d t) - \int\frac{2[a-e^\varepsilon x]^+-2a}{e^\varepsilon x} \nu_{0}^{\eta, a}(d t) \;\;=\;\;2 \Delta_L\left[\frac{[a-e^\varepsilon x]^+-a}{e^\varepsilon x};[\eta, 1]\right]\;.
		\end{eqnarray*}
		
	\end{enumerate}
		The two desired measures are constructed.
\subsubsection{Proof of Lemma~\ref{lem:tvbound}}
	
Let $g(x;a) = \frac{[ a-x]^+-a}{x}$. We have
\begin{eqnarray}
	g(x;a) 
	&\;\;=\;\;& \frac{|x-a|-a}{2x}-\frac{1}{2}\;.
\end{eqnarray}
By the definition of best polynomial approximation error, for $L>1$, we have
\begin{eqnarray*}
	\Delta_{L}\left[g(x ; a) ;[\frac{a}{D}, 1]\right]
		&\;\;=\;\;& \inf_{h(x)\in {\rm Poly}_L}\sup_{x\in [\frac{a}{D},1]}\left|\frac{|x-a|-a}{2x}-\frac{1}{2}-h(x)\right|\\
		&\;\;=\;\;& \inf_{h(x)\in {\rm Poly}_L}\sup_{x\in [\frac{a}{D},1]}\left|\frac{|x-a|-a}{x}-h(x)\right|\\
		&\;\;=\;\;& \Delta_{L}\left[\frac{|x-a|-a}{x};[\frac{a}{D}, 1]\right]\\
		&\;\;\gtrsim\;\;&
		\begin{cases}
			\frac{1}{L \sqrt{a}} & \frac{1}{L^{2}} \leq a \leq \frac{1}{2} \\ 
			1 & 0<a<\frac{1}{L^{2}}
		\end{cases}\;,
	\end{eqnarray*}
	where $D$ is from \cite[Lemma~30]{JYT18}.
	
	Now we consider $f(x;a) = \frac{[ e^{-\varepsilon}a-x]^+-e^{-\varepsilon}a}{x}$, where $a\in (0,1/2]$. As $e^{-\varepsilon}a \in (0,\frac{1}{2}]$. there exists $D>0$ such that
	\begin{eqnarray}
		\Delta_{L}\left[f(x; a); [\frac{a}{D}, 1]\right]\;\;\geq\;\;\Delta_{L}\left[f(x; a); [\frac{e^{-\varepsilon} a}{D}, 1]\right]\;\;\gtrsim\;\;\left(1\wedge \frac{1}{L\sqrt{e^{-\varepsilon}a}}\right)\;\;\geq \;\;\left(1\wedge \frac{1}{L\sqrt{a}}\right)\;.
	\end{eqnarray}

\subsection{Proof of Theorem~\ref{thm:opt}} 
\label{sec:proof_opt}

Define good events, where our choice of regimes are correct as
\begin{eqnarray}
	E_1 \;=\; \left\{\Big\{i: \hp_{i,1}-e^\varepsilon \hq_{i,1}>\sqrt{\frac{(c_1+c_2)\ln n }{n}}(\sqrt{\hp_{i,1}}+\sqrt{e^\varepsilon \hq_{i,1}})\Big\}\subseteq \Big\{i: e^\varepsilon q_i \leq p_i\Big\}\right\}\;,
\end{eqnarray}
\begin{eqnarray}
	E_2 \;=\; \left\{\Big\{i:\hp_{i,1}- e^\varepsilon \hq_{i,1}<-\sqrt{\frac{(c_1+c_2)\ln n }{n}}(\sqrt{\hp_{i,1}}+\sqrt{e^\varepsilon \hq_{i,1}})\Big\}\subseteq \Big\{i: e^\varepsilon q_i \geq p_i\Big\}\right\}\;,
\end{eqnarray}
\begin{eqnarray}
	E_3 = \left\{\Big\{i: e^\varepsilon \hq_{i,1}+\hp_{i,1}<\frac{c_1 \ln n}{n} \Big\}\subseteq \Big\{i: \big(p_i, e^\varepsilon q_i \big)\in [0,\frac{2c_1\ln n}{n}]^2\Big\}\right\}\;,
\end{eqnarray}
  and 
\begin{eqnarray}
	E_4 = \left\{\Big \{i:(\hp_{i,1}, e^{\varepsilon }\hq_{i,1}) \in U(c_1,c_1), \hp_{i,1}+e^\varepsilon \hq_{i,1}\geq \frac{c_1 \ln n}{n}\Big\} \nonumber
	\right.\\ \left. \;\subseteq\; \Big\{i: (p_i, e^{\varepsilon }q_i) \in U(c_1,c_1), p_i+e^\varepsilon q_i\geq \frac{c_1\ln n}{2n}, \hp_{i,1}+e^{\varepsilon }\hq_{i,1}\geq \frac{p_i+e^{\varepsilon }q_i}{2} \Big\} \right\}\;.
\end{eqnarray}

Denote the overall good event as
\begin{eqnarray}
	E \;\; = \;\; E_1\cap E_2\cap E_3\cap E_4\;.\label{eq:good_event}
\end{eqnarray}

Decompose the error under good events as 
\begin{eqnarray}
	\cE_1 &\triangleq& \sum_{i\in I_1} \big\{ \hp_{i,2}-e^\varepsilon \hq_{i,2} -  [p_i-e^\varepsilon q_i]^+ \big\} \;, \\
	\cE_2 &\triangleq& \sum_{i\in I_2} \big\{ \tD_K^{(1)}(\hp_{i,2},\hq_{i,2})-[p_i-e^\varepsilon q_i]^+ \big\} \;,\\
	\cE_3 &\triangleq& \sum_{i\in I_3} \big\{ \tD_K^{(2)}(\hp_{i,2},\hq_{i,2}; \hp_{i,1},\hq_{i,1})-[p_i-e^\varepsilon q_i]^+ \big\} \;,
\end{eqnarray}

where the indices of those regimes under the good events are
\begin{eqnarray}
	I_1 \triangleq  \left\{ i : \hp_{i, 1}-e^{\varepsilon}\hq_{i, 1}>\sqrt{\frac{(c_{1}+c_{2}) \ln n}{n}}(\sqrt{\hp_{i, 1}}+\sqrt{e^{\varepsilon}\hq_{i, 1}}), e^\varepsilon q_i\leq p_i \right\}\;,\\
	I_2 \triangleq  \left\{ i :  \hp_{i,1}+e^\varepsilon \hq_{i,1}\leq \frac{c_1e^\varepsilon \ln n}{n}, (p_i,e^{\varepsilon }q_i)\in \left[0,\frac{2c_1\ln n}{n}\right]^2 \right\}\;,\\
	I_3 \triangleq  \left\{ i :  \hp_{i,1}+e^\varepsilon \hq_{i,1}\geq \frac{c_1 \ln n}{n}, (\hp_{i,1}, e^{\varepsilon }\hq_{i,1})\in U(c_1,c_2),  \right.\nonumber\\
	  \left. (p_i, e^{\varepsilon }q_i)\in U(c_1,c_1),  p_i+e^\varepsilon q_i\geq \frac{c_1 \ln n}{2n}, \hp_{i,1}+e^{\varepsilon }\hq_{i,1} \geq \frac{p_i+e^{\varepsilon }q_i}{2} \right \}\;.
\end{eqnarray}

We can bound the squared error as 
\begin{eqnarray}
	\E \big[ \, \big( \hd_{\varepsilon,K,c_1,c_2}(P_n \| Q_n) - d_\varepsilon(P\|Q) \big)^2 \, \big] & \leq & 
		\E[ \, \big( \hd_{\varepsilon,K,c_1,c_2}(P_n \| Q_n) - d_\varepsilon(P\|Q) \big)^2\, \ind(E)\, ] + \prob( E^c)    \nonumber \\
		&\leq & \E[(\cE_1+\cE_2+\cE_3)^2] +\prob( E^c)    \nonumber\\
		&\leq& 3 \E[\cE_1^2] + 3  \E[\cE_2^2] + 3  \E[\cE_3^2] +  \prob( E^c)  \;.
	\label{eq:opt_error}
\end{eqnarray}

The last term on the bad event is bounded by $15S/n^\beta$ as shown in following lemma, with a proof in \ref{sec:proof_bad_event}.

\begin{lemma}
\label{lem:bad_event}
	  Assuming $\frac{c_2}{c_1}<\frac{8}{(\sqrt{2}+1)^{2}}-1 \approx 0.373$, and let $\beta = \min \left\{\frac{c_1}{6}, \frac{(c_1-c_2)^2}{96c_1}, \frac{1}{3}\left(\sqrt{2c_1}-\frac{\sqrt{2}+1}{2}\sqrt{c_1+c_2}\right)^2\right\}$, we have
	  \begin{eqnarray}
	  	\prob(E^c)\;\;\leq\;\; \frac{15S}{n^\beta}\;,
	  \end{eqnarray}
	  where good event $E$ is defined in Eq.~\eqref{eq:good_event}.
\end{lemma}

For the first term in Eq.~\eqref{eq:opt_error}, we have
\begin{eqnarray}
	\E[\cE_1^2] \;\;=\;\; \E [\var(\cE_1 | I_1) + (E[ \cE_1| I_1])^2] \; = \; \E[\var(\cE_1|I_1)] \; \leq \; \sum_{i\in \{i:e^\varepsilon q_i\leq p_i\}} \frac{p_i+e^{2\varepsilon}  q_i}{n}
	\;\;\lesssim\;\;\frac{e^\varepsilon}{n}
\end{eqnarray}
where we use the fact that $\E[\cE_1|I_1] = 0$ with probability one and the fact that $\hp_{i,2}$ and $\hq_{i,2}$ are independent for indices in $I_1$.

The second term in Eq.~\eqref{eq:opt_error} is bounded by following lemma, with a proof in Section~\ref{sec:proof_biasvariance1}.

\begin{lemma}
\label{lem:biasvariance1}
	Suppose $(p, e^{\varepsilon}q)\in \left[0, \frac{2c_1\ln n}{n}\right]^2$, $(n\hp,n\hq)\sim \Poi(np)\times\Poi(nq)$. Then,
	\begin{eqnarray}
			\left|\E\tD_K^{(1)}(\hp,\hq)-[p-e^\varepsilon q]^+\right|\;\;\lesssim  \frac{1}{K}\sqrt{\frac{c_1  \ln n}{n}}(\sqrt{ p}+\sqrt{e^\varepsilon q})+\frac{1}{K^2}\frac{c_1 \ln n}{n}\;,
	\end{eqnarray}
	and
	\begin{eqnarray}
		\Var\left(\tD_K^{(1)}(\hp,\hq)\right)\;\;\lesssim \;\; \frac{B^K c_1 c_3^4\ln^5 n}{n}(p+e^{\varepsilon} q )\;,
	\end{eqnarray}
	for some constant $B>0$. The estimator $\tD_K^{(1)}$ is introduced in Eq.~\eqref{eq:Opt1_dk} and $K = c_3\ln n$, $c_3 e^{\varepsilon}<c_1$.
\end{lemma} 

We have
\begin{eqnarray}
	&&\E[\cE_3^2] \\
	\;\;\lesssim \;\; && \sum_{i=1}^S\frac{B^K c_1 c_3^4\ln^5 n}{n}( p_i + e^{\varepsilon}q_i)+\left(\sum_{i=1}^S \frac{1}{K}\sqrt{\frac{c_1  \ln n}{n}}(\sqrt{ p_i}+\sqrt{e^\varepsilon q_i})+\frac{1}{K^2}\frac{c_1 \ln n}{n}\right)^2\\
	\;\;\lesssim \;\; && \frac{c_1c_3^4\ln^5 n}{n^{1-c_3\ln B}}(e^{\varepsilon}+1)+\frac{c_1(e^\varepsilon+1)S}{c_3^2n \ln n} \vee \left(\frac{c_1 S }{c_3^2n\ln n}\right)^2\;.
\end{eqnarray}

The third term in Eq.~\eqref{eq:opt_error} is bounded by following lemma, with a proof in Section~\ref{sec:proof_biasvariance2}.

\begin{lemma}
\label{lem:biasvariance2}
	Suppose $(p,e^{\varepsilon}q)\in U(c_1,c_1)$, $ p+e^{\varepsilon}q\geq \frac{c_1 \ln n}{2n}$, $x+e^{\varepsilon}y\geq \frac{p+e^{\varepsilon}q}{2}$, $x\in [0,1]$, $y\in [0,1]$. Suppose $(n\hp,n\hq)\sim \Poi(np)\times\Poi(nq)$. Then,
	\begin{eqnarray}
		\left|\E \tD_K^{(2)}(\hp, \hq ; x, y)-[p-e^\varepsilon q]^+ \right| \;\;\lesssim\;\; \frac{1}{K} \sqrt{\frac{c_1 \ln n}{n}}(\sqrt{x}+\sqrt{e^\varepsilon y})\;,
	\end{eqnarray}
	and 
	\begin{eqnarray}
		\Var\left(\tD_K^{(2)}(\hp, \hq ; x, y)\right) \;\;\lesssim\;\; \frac{B^K c_1  \ln n}{n}( x+e^\varepsilon y)\;,
	\end{eqnarray}
	for some constant $B>0$, and $K=c_3\ln n$, $c_3 e^{\varepsilon} < c_1$.
\end{lemma}

We have
\begin{eqnarray}
	&&\E\left[\cE_4^2|\hp_{i,1},\hq_{i,1}:1\leq i\leq S\right] \\
	\;\;\lesssim \;\; && \sum_{i=1}^S\frac{B^K c_1  \ln n}{n}( \hp_{i,1}+e^\varepsilon\hq_{i,1})+\left(\sum_{i=1}^S\frac{1}{K} \sqrt{\frac{c_1 \ln n}{n}}(\sqrt{ \hp_{i,1}}+\sqrt{e^\varepsilon\hq_{i,1}})\right)^2\;,
\end{eqnarray}
where $B$ is the larger constant defined in both Lemma~\ref{lem:biasvariance1} and Lemma~\ref{lem:biasvariance2}.

Taking expectation with respect to $\{\hp_{i,1},\hq_{i,1}:1\leq i\leq S\}$, we have
\begin{eqnarray}
	&&\E\left[\cE_4^2\right]\\
	\;\;\lesssim \;\; && \sum_{i=1}^{S} \frac{c_{1} \ln n}{n^{1-c_3\ln B}} \E\left(\hp_{i, 1}+e^\varepsilon\hq_{i, 1}\right)+\E\left(\sum_{i=1}^{S} \frac{\sqrt{c_{1}\left(\hp_{i, 1}+e^\varepsilon\hq_{i, 1}\right)}}{\sqrt{c_{3}^{2} n \ln n}}\right)^{2}\\
	\;\;\lesssim \;\; &&  \frac{c_{1} \ln n}{n^{1-c_3\ln B}}(e^\varepsilon+1) + \sum_{i=1}^{S} \E\frac{c_{1}\left(\hp_{i, 1}+e^\varepsilon\hq_{i, 1}\right)}{c_{3}^{2} n \ln n}+\\
	&&\sum_{1 \leq i, j \leq S, i \neq j} \sqrt{\frac{\E\left[c_{1}\left(\hp_{i, 1}+e^\varepsilon\hq_{i, 1}\right)\right]}{c_{3}^{2} n \ln n}} \sqrt{\frac{\E\left[c_{1}\left(\hp_{j, 1}+e^\varepsilon\hq_{j, 1}\right)\right]}{c_{3}^{2} n \ln n}}\\
	\;\;\lesssim \;\; &&\frac{c_{1} \ln n}{n^{1-c_3\ln B}}(e^\varepsilon+1) + \frac{c_1}{c_3^2n\ln n}(e^\varepsilon+1)+ \sum_{1 \leq i, j \leq S} \frac{c_{1} \left( p_{i}+e^\varepsilon q_{i}+ p_{j}+e^\varepsilon q_{j}\right)}{c_{3}^{2} n \ln n}\\
	\;\;\lesssim \;\; && \frac{c_{1} \ln n}{n^{1-c_3\ln B}}(e^\varepsilon+1) + \frac{c_1}{c_3^2n\ln n}(e^\varepsilon+1)\;.
	\end{eqnarray}
	
Combing everything together, we have
\begin{eqnarray}
	&&\E \big[ \, \big( \hd_{\varepsilon,K,c_1,c_2}(P_n \| Q_n) - d_\varepsilon(P\|Q) \big)^2 \, \big] \\
	\;\;\lesssim\;\; && \frac{e^{\varepsilon}}{n}+\frac{ (c_1c_3^4+c_1)\ln^5 n}{n^{1-c_3\ln B}}(e^\varepsilon+1)+\frac{c_1(e^\varepsilon+1)S}{c_3^2n \ln n} \vee \left(\frac{c_1 S }{c_3^2n\ln n}\right)^2 +\frac{S}{n^\beta}\;.
\end{eqnarray}
If $\ln n\lesssim \ln S$, as we assume  $\frac{c_2}{c_1}<0.373$, we can take $c_2$ small enough and $c_1,c_3$ large enough to guarantee that $\frac{S}{n^\beta}\lesssim \frac{S}{n\ln n}$, $\frac{\ln ^{5} n}{n^{1-c_3\ln B}} \lesssim \frac{S}{n \ln n}$. 
We have,
\begin{eqnarray}
	\E \big[ \, \big( \hd_{\varepsilon,K,c_1,c_2}(P_n \| Q_n) - d_\varepsilon(P\|Q) \big)^2 \, \big] \;\; \lesssim\;\; \frac{ e^\varepsilon S}{n\ln n}\;.
\end{eqnarray}


\subsubsection{Proof of Lemma~\ref{lem:bad_event}}
\label{sec:proof_bad_event}
The following lemma shows that non-smooth region $U(c_1,c_2)$ contains the region $U(p;c_1,c_2)$ defined previously, which will be later used to bound the probability of bad events.
\begin{lemma}
\label{lem:decompose region}
The two-dimensional set $U(c_1,c_1)$ defined in Eq.~\eqref{eq:Qopt_region1} satisfies
	\begin{eqnarray}
		\cup_{x = e^\varepsilon y, x,y\in [0,1]}U(  e^{\varepsilon}x;c_1,c_1)\times U(  e^{2\varepsilon} y;c_1,c_1) \subset U(c_1,c_1) \;.
	\end{eqnarray}
\end{lemma}

{\bf 1) Analysis of $\prob(E_1^c)$:}

It follows from Lemma~\ref{lem:decompose region} that
\begin{eqnarray*}
	\prob(E_1^c) &\;\; = \;\;& \prob\left(\bigcup_{i=1}^S  \Big\{e^\varepsilon q_i< p_i, e^\varepsilon \hq_{i,1}-\hp_{i,1}>\sqrt{\frac{(c_1+c_2)\ln n }{n}}(\sqrt{e^\varepsilon \hq_{i,1}}+\sqrt{\hp_{i,1}})  \Big\} \right)\\
	&\;\;\leq \;\;& S\max_{i\in [S]}\prob\left(e^\varepsilon q_i< p_i, e^\varepsilon \hq_{i,1}-\hp_{i,1}>\sqrt{\frac{(c_1+c_2)\ln n }{n}}(\sqrt{e^\varepsilon \hq_{i,1}}+\sqrt{\hp_{i,1}})  \right)\\
	&\;\;\leq \;\;& S\max_{i\in [S]}\prob\Big(e^\varepsilon q_i = p_i, e^\varepsilon \hq_{i,1}-\hp_{i,1}>\sqrt{\frac{(c_1+c_2)\ln n }{n}}(\sqrt{e^\varepsilon \hq_{i,1}}+\sqrt{\hp_{i,1}})  \Big)\\
	&\;\;\leq \;\;& S\max_{i\in [S]}\prob\Big(e^\varepsilon q_i = p_i, ( \hp_{i,1},e^{\varepsilon}\hq_{i,1})\notin U(c_1,c_2) \Big)\\
	&\;\;\leq \;\;& S\max_{i\in [S]}\prob\Big(e^{\varepsilon} q_i = p_i, (  \hp_{i,1}, e^{\varepsilon}\hq_{i,1})\notin U(e^{\varepsilon}p_i;\frac{c_1+c_2}{2},\frac{c_1+c_2}{2})\times\\
	& &	 U(e^{2\varepsilon}  q_i;\frac{c_1+c_2}{2},\frac{c_1+c_2}{2}) \Big)\\
	&\;\;\leq \;\;& S\max_{i\in [S]}\Big(1- \prob\big( \hp_{i,1}\in U(  e^{\varepsilon} p_i;\frac{c_1+c_2}{2},\frac{c_1+c_2}{2}) \big)\times \\
	 & &\prob\big( e^\varepsilon q_i = p_i, e^{\varepsilon} \hq_{i,1}\in U( e^{2\varepsilon} q_i;\frac{c_1+c_2}{2},\frac{c_1+c_2}{2})\big)\Big)\\
	 &\;\;\leq \;\;& S\left(1-\left(1-\frac{2}{n^{\frac{c_1+c_2}{6}}}\right)^2\right)
	 \;\;\leq \;\; \frac{4S}{n^{\frac{c_1+c_2}{6}}}\;,
\end{eqnarray*}
where we have applied Lemma~\ref{lem:pqtail} in the last inequality.

{\bf 2) Analysis of $\prob(E_2^c)$:}
Similarly, we have
\begin{eqnarray}
	\prob(E_2^c) \;\;\leq\;\; \frac{4S}{n^{\frac{c_1+c_2}{6}}}\;.
\end{eqnarray}

{\bf 3) Analysis of $\prob(E_3^c)$:}

\begin{eqnarray*}
	\prob(E_3^c) &\;\;= \;\;& \prob\left(\bigcup_{i=1}^S\left\{(p_i, e^{\varepsilon} q_i )\notin \left[0,\frac{2c_1\ln n}{n}\right]^2, \hp_{i,1}+e^\varepsilon \hq_{i,1}<\frac{c_1 \ln n}{n}\right\}\right)\\
	&\;\;\leq \;\;& \prob\left(\bigcup_{i=1}^S\left\{p_i+e^{\varepsilon}q_i > \frac{2c_1\ln n}{n}, \hp_{i,1}+e^\varepsilon \hq_{i,1}<\frac{c_1 \ln n}{n}\right\}\right)\\
	 &\;\;\leq \;\;&S\max_{i\in[S]} \prob\left(p_i+e^{\varepsilon}q_i > \frac{2c_1\ln n}{n}, \hp_{i,1}+e^{\varepsilon}\hq_{i,1}<\frac{c_1 \ln n}{n} \right)\\
	&\;\;\leq \;\;& \frac{S}{n^{c_1/4}}\;,
\end{eqnarray*}
where we have applied Lemma~\ref{lem:weighted_poisson_tail} in the last inequality.

{\bf 4) Analysis of $\prob(E_4^c)$:}

\begin{eqnarray*}
	\prob(E_4^c) &\;\;\leq \;\;& S\max_{i\in[S]}\prob\left((p_i,e^{\varepsilon}q_i)\notin U(c_1,c_1), (\hp_{i,1},e^{\varepsilon}\hq_{i,1} )\in U(c_1,c_2)\right)+\\
	&&S\max_{i\in[S]} \prob\left(\hp_{i,1}+ e^\varepsilon \hq_{i,1}> \frac{c_1 \ln n}{n},  p_i+e^\varepsilon q_i< \frac{c_1 \ln n}{2n}\right)+\\
	&&S\max_{i\in[S]} \prob\left(\hp_{i,1}+ e^\varepsilon \hq_{i,1}\geq \frac{c_1 \ln n}{n},  p_i+e^\varepsilon q_i\geq  2(\hp_{i,1}+e^\varepsilon \hq_{i,1})\right)\;.
\end{eqnarray*}

Using Lemma~\ref{lem:weighted_poisson_tail} again, we have
\begin{eqnarray*}
	S\max_{i\in[S]} \prob\left(  \hp_{i,1}+e^{\varepsilon}\hq_{i,1}> \frac{c_1 \ln n}{n},   p_i+e^{\varepsilon}q_i< \frac{c_1\ln n}{2n}\right)
	\;\;\leq\;\; \frac{S}{n^{c_1/6}}\;,
\end{eqnarray*}
and
\begin{eqnarray*}
	&&S\max_{i\in[S]} \prob\left( \hp_{i,1}+e^\varepsilon \hq_{i,1}\geq\frac{c_1 \ln n}{n},  p_i+e^\varepsilon q_i\geq  2(\hp_{i,1}+e^\varepsilon \hq_{i,1})\right)\\
	\;\;\leq\;\;&& S\max_{i\in[S]}\prob\left( p_i+e^\varepsilon q_i\geq \frac{2c_1 \ln n}{n}, p_i+ e^\varepsilon q_i\geq  2(\hp_{i,1}+e^\varepsilon \hq_{i,1})\right)\\
	\;\;\leq\;\;&& \frac{S}{n^{c_1/4}}\;.
\end{eqnarray*}

It suffices to show that for $p,q \in [0,1]$, there exists some constant $c>0$ such that
\begin{eqnarray}
	\left(\bigcup_{(p, e^{\varepsilon}q)\notin U(c_1,c_1) } U(e^{\varepsilon} p;c,c)\times U( e^{2\varepsilon}q;c,c) \right) \bigcap U(c_1,c_2) = \emptyset.\label{eq:decomppq}
\end{eqnarray}
Indeed, we have
\begin{eqnarray*}
	&&S\max_{i\in[S]}\prob\left((p_i,e^{\varepsilon}q_i )\notin U(c_1,c_1), (\hp_{i,1},e^{\varepsilon}\hq_{i,1})\in U(c_1,c_2)\right)\\
	\;\;\leq\;\;&& S\max_{i\in[S]}\prob\left(   (\hp_{i,1}, e^{\varepsilon}\hq_{i,1})\notin U(e^{\varepsilon} p_i;c,c)\times U(e^{2\varepsilon} q_i;c,c)     \right)\\
	\;\;\leq\;\;&& S\max_{i\in[S]}\Big(1-\prob\big( \hp_{i,1}\in U(e^{\varepsilon} p_i;c,c)\big) \prob\big( e^{\varepsilon} \hq_{i,1}\in U(e^{2\varepsilon} q_i;c,c)    \big)\Big)\\
	\;\;\leq \;\;&& S\left(1-\left(1-\frac{2}{n^{\frac{c}{3}}}\right)^2\right)
	 \;\;\leq \;\; \frac{4S}{n^{c/3}}\;,
\end{eqnarray*}
where the last inequality follows from Lemma~\ref{lem:pqtail}.

Now we work out a $c$ that satisfies \eqref{eq:decomppq}. We prove the case when $\sqrt{p}-\sqrt{e^{\varepsilon}q}\geq \sqrt{\frac{2c_1\ln n}{n}}$. The other case can be proved in a similar way. Assume $c<c_1$. In this case $p\geq \frac{2c_1\ln n}{n}$. We will show that for any point $(x,e^{\varepsilon}y)\in U(e^{\varepsilon} p;c,c)\times U( e^{2\varepsilon}q;c,c)$, we have $\sqrt{x}-\sqrt{e^{\varepsilon}y}\geq \sqrt{\frac{(c_1+c_2)\ln n}{n}}$.

If $q\leq \frac{ce^{-\varepsilon}\ln n}{n}$, for any $(x, e^{\varepsilon}y)\in U(e^{\varepsilon} p;c,c)\times U( e^{2\varepsilon} q;c,c)$, we have 
\begin{eqnarray*}
	\sqrt{x}-\sqrt{e^{\varepsilon}y}&\;\;\geq\;\;& \sqrt{p-\sqrt{\frac{cp\ln n}{n}}}-\sqrt{\frac{2c\ln n}{n}}\\
	&\;\;\geq\;\;& \sqrt{\frac{2c_1\ln n}{n}-\sqrt{2cc_1}\frac{\ln n}{n}}-\sqrt{\frac{2c\ln n}{n}}\\
	&\;\; = \;\;& \sqrt{\frac{\ln n}{n}} \left(\sqrt{2c_1-\sqrt{2cc_1}}-\sqrt{2c}\right)\;,
\end{eqnarray*}
where in the second step, we use the fact that $x-\sqrt{ax}$, $a>0$ is a monotonically increasing function when $x\geq a/4$ and the fact that $p\geq \frac{2c_1\ln n}{n}$. Let $c = \frac{(c_1-c_2)^2}{32c_1}$, we can verify that 
\begin{eqnarray}
	\sqrt{x}-\sqrt{e^{\varepsilon}y}&\;\;\geq\;\;& \sqrt{\frac{\ln n}{n}} \left(\sqrt{2c_1-\sqrt{2cc_1}}-\sqrt{2c}\right)
	\;\;\geq\;\; \sqrt{\frac{\ln n}{n}}\sqrt{c_1+c_2}\;.
\end{eqnarray}

If $q> \frac{ce^{-\varepsilon}\ln n}{n}$, for any $(x,e^{\varepsilon}y)\in U(e^{\varepsilon} p;c,c)\times U( e^{2\varepsilon}q;c,c)$, we have 

\begin{eqnarray*}
	\sqrt{x}-\sqrt{e^{\varepsilon}y}&\;\;\geq\;\;& \sqrt{p-\sqrt{\frac{cp\ln n}{n}}}-\sqrt{e^{\varepsilon}q+\sqrt{\frac{ce^{\varepsilon}q\ln n}{n}}}\\
	&\;\; = \;\;& \frac{p-\sqrt{\frac{c p \ln n}{n}}-e^{\varepsilon}q-\sqrt{\frac{c e^{\varepsilon}q \ln n}{n}}}{\sqrt{p-\sqrt{\frac{c p \ln n}{n}}}+\sqrt{e^{\varepsilon}q+\sqrt{\frac{c e^{\varepsilon}q \ln n}{n}}}}\\
	&\;\; = \;\;& \frac{(\sqrt{p}-\sqrt{e^{\varepsilon}q})(\sqrt{p}+\sqrt{e^{\varepsilon}q})-\sqrt{\frac{c \ln n}{n}}(\sqrt{e^{\varepsilon}q}+\sqrt{p})}{\sqrt{p-\sqrt{\frac{c p \ln n}{n}}}+\sqrt{e^{\varepsilon}q+\sqrt{\frac{c e^{\varepsilon}q \ln n}{n}}}}\\
	&\;\; \geq \;\;& (\sqrt{2c_1}-\sqrt{c})\sqrt{\frac{\ln n}{n}}\frac{\sqrt{p}+\sqrt{e^{\varepsilon}q}}{\sqrt{p-\sqrt{\frac{c p \ln n}{n}}}+\sqrt{e^{\varepsilon}q+\sqrt{\frac{c e^{\varepsilon}q \ln n}{n}}}}\;.
\end{eqnarray*}
Further, since $e^{\varepsilon}q>\frac{c\ln n}{n}$, 
\begin{eqnarray*}
	\frac{\sqrt{p}+\sqrt{e^{\varepsilon}q}}{\sqrt{p-\sqrt{\frac{c p \ln n}{n}}}+\sqrt{e^{\varepsilon}q+\sqrt{\frac{c e^{\varepsilon}q \ln n}{n}}}} &\;\;\geq\;\;& \frac{\sqrt{p}+\sqrt{e^{\varepsilon}q}}{\sqrt{p}+\sqrt{2 e^{\varepsilon}q}}\\
	&\;\;\geq\;\;& \frac{\sqrt{e^{\varepsilon}q}+\sqrt{\frac{2 c_{1} \ln n}{n}}+\sqrt{e^{\varepsilon}q}}{\sqrt{2 e^{\varepsilon}q}+\sqrt{e^{\varepsilon}q}+\sqrt{\frac{2 c_{1} \ln n}{n}}}\\
	&\;\;\geq\;\;& \frac{2}{\sqrt{2}+1}\;,
\end{eqnarray*}
where in the second inequality, we used the fact that $\frac{x+\sqrt{e^{\varepsilon}q}}{x+\sqrt{2e^{\varepsilon}q}}$ is a monotonically increasing function of $x$ when $x\geq 0$, and in the third inequality, we used the fact that $\frac{2 x+a}{(\sqrt{2}+1) x+a}$ is a monotonically decreasing function of $x$ when $a>0$, $x>0$. To guarantee that $\sqrt{x}-\sqrt{e^{\varepsilon}y}\geq \sqrt{\frac{(c_1+c_2)\ln n}{n}}$, we need 
\begin{eqnarray}
	\frac{2}{\sqrt{2}+1}\left(\sqrt{2 c_{1}}-\sqrt{c}\right) \;\;\geq\;\; \sqrt{c_{1}+c_{2}}\;,
\end{eqnarray}
which is equivalent to 
\begin{eqnarray}
	c \leq\left(\sqrt{2 c_{1}}-\frac{\sqrt{2}+1}{2} \sqrt{c_{1}+c_{2}}\right)^{2}\;,
\end{eqnarray}
with the constraint that $\frac{c_{2}}{c_{1}}<\frac{8}{(\sqrt{2}+1)^{2}}-1 \approx 0.373$.
\subsubsection{Proof of Lemma~\ref{lem:decompose region}}
	If $x\leq \frac{c_1 e^\varepsilon \ln n}{n}$ and thus $y \leq \frac{c_1  \ln n}{n}$, it suffices to show $[0, \frac{2c_1\ln n}{n}]^2 \subset U(c_1,c_1)$. For $(u, e^{\varepsilon} v)\in [0, \frac{2c_1\ln n}{n}]^2 $, we have
	\begin{eqnarray}
		| \sqrt{  u}-\sqrt{ e^{\varepsilon}v}| \;\;\leq\;\; \sqrt{\frac{2c_1 \ln n}{n}}\;.
	\end{eqnarray}
	
	For $x> \frac{c_1 e^\varepsilon \ln n}{n} $ and thus $y > \frac{c_1  \ln n}{n}$, it suffices to show
	$\left[x-\sqrt{\frac{c_{1} x\ln n}{n}}, x +\sqrt{\frac{c_{1} x\ln n}{n}}\right]^2\subset U(c_1,c_1)$. It is shown in \cite[Lemma~3]{JYT18} that for any $( u, e^{\varepsilon}v)\in \left[y-\sqrt{\frac{c_{1} y\ln n}{n}}, y +\sqrt{\frac{c_{1} y \ln n}{n}}\right]^2 $, we have
	\begin{eqnarray}
		| \sqrt{ u}-\sqrt{ e^{\varepsilon}v}| \;\;\leq\;\; \sqrt{\frac{2c_1 \ln n}{n}}\;.
	\end{eqnarray}

\subsubsection{Proof of Lemma~\ref{lem:biasvariance1}}
\label{sec:proof_biasvariance1}
	We first analyze the bias. Let $\Delta = \frac{c_1 \ln n}{n}$. As we have applied unbiased estimator $\tD_K^{(1)}(x,y)$ of $D_K^{(1)}(x,y)$, the bias is entirely due to the functional approximation. We show that for $(x,y)\in [0,1]^2$, $\left|u_K(x,y)v_K(x,y)-[x - y]^+\right|\lesssim \left(\frac{\sqrt{x}+\sqrt{y}}{K}+\frac{1}{K^2}\right)$. Indeed, we have
	\begin{eqnarray*}
		&&\left|u_K(x,y)v_K(x,y)-[x - y]^+\right| \\
		\;\;=\;\;&& \left|u_K(x,y)v_K(x,y)-u_K(x,y)[\sqrt{x}-\sqrt{y}]^++u_K(x,y)[\sqrt{x}-\sqrt{y}]^+-[x - y]^+\right|\\
		\;\;\leq \;\;&& |u_K(x,y)||v_K(x,y)-[\sqrt{x}-\sqrt{y}]^+|+[\sqrt{x}-\sqrt{y}]^+|u_K(x,y)-\sqrt{x}-\sqrt{y}|\\
		\;\;\leq \;\;&&  |u_K(x,y)-\sqrt{x}-\sqrt{y}||v_K(x,y)-[\sqrt{x}-\sqrt{y}]^+|+\nonumber\\
		&&|\sqrt{x}+\sqrt{y}||v_K(x,y)-[\sqrt{x}-\sqrt{y}]^+|+[\sqrt{x}-\sqrt{y}]^+|u_K(x,y)-\sqrt{x}-\sqrt{y}|\;.
	\end{eqnarray*}
	
	It follows from  Lemma~\ref{lem:poly_error} and Lemma~\ref{lem:smoothness_sqrt_xy} that 
	\begin{eqnarray}
		\Delta_K\left[\sqrt{x};[0,1]\right] \;\;\lesssim \;\;\frac{1}{K}\;,
	\end{eqnarray}
	which implies 
	\begin{eqnarray}
		|u_K(x,y)-\sqrt{x}-\sqrt{y}| \;\;\lesssim\;\; \frac{1}{K}\;.\label{eq:uapprox}
	\end{eqnarray}
	It follows from Lemma~\ref{lem:poly_error}, Lemma~\ref{lem:smoothness_sqrt_xy} and the fact that $[\sqrt{b}-\sqrt{a}]^+\leq [\sqrt{b}-\sqrt{c}]^++[\sqrt{c}-\sqrt{a}]^+$, we have
	\begin{eqnarray}
		|v_K(x,y)-[\sqrt{x}-\sqrt{y}]^+| \;\;\lesssim\;\; \frac{1}{K}\;.\label{eq:vapprox}
	\end{eqnarray}
	Together with Eq.~\eqref{eq:uapprox}, we have
	\begin{eqnarray}
		\left|u_K(x,y)v_K(x,y)-[x - y]^+\right| &\;\;\lesssim\;\;&\frac{1}{K^2}+\frac{\sqrt{x}+\sqrt{y}+[\sqrt{x}-\sqrt{y}]^+}{K}\\
		&\;\;\lesssim\;\;& \frac{1}{K^2}+\frac{\sqrt{x}+\sqrt{y}}{K}\;,
	\end{eqnarray}
	which implies there exists a constant $M>0$ such that
	\begin{eqnarray}
		\left|u_K(x,y)v_K(x,y)-u_K(0,0)v_K(0,0)-[x - y]^+\right| \;\;\leq\;\;  M\left(\frac{1}{K^2}+\frac{\sqrt{x}+\sqrt{y}}{K}\right)\;.
	\end{eqnarray}
	
	Let $x = p/(2\Delta)$ and $y =e^{\varepsilon}q/(2\Delta)$. We have
	\begin{eqnarray*}
	&&\sup_{\left( p,  e^{\varepsilon}q\right)\in  \left[0, 2\Delta\right]^2} \left|\E\tD_K^{(1)}(\hp,\hq)-[p-e^\varepsilon q ]^+\right|\\
		\;\;=\;\;&&\sup_{\left(p,  e^{\varepsilon}q\right)\in  \left[0, 2\Delta\right]^2} \left|D_K^{(1)}(p,q)-[p-e^\varepsilon q ]^+\right|\\
		\;\;=\;\;&& \sup_{\left(p,  e^{\varepsilon}q\right)\in  \left[0, 2\Delta\right]^2} 2\Delta \left|h_{2K}(\frac{p}{2\Delta}, \frac{e^{\varepsilon}q}{2\Delta})-[\frac{p}{2\Delta}-\frac{e^{\varepsilon}q}{2\Delta} ]^+\right|\\
		\;\;=\;\;&& \sup_{\left(x, y\right)\in  \left[0, 1\right]^2} 2\Delta  \left|h_{2K}(x, y)-[x - y]^+\right|\\
		\;\;=\;\;&& \sup_{\left(x, y\right)\in  \left[0, 1\right]^2} 2\Delta  \left|u_K(x,y)v_K(x,y)-u_K(0,0)v_K(0,0)-[x- y]^+\right|\\
		\;\;\leq \;\; && 2\Delta  M\left(\frac{1}{K^2}+\frac{\sqrt{x}+\sqrt{y}}{K}\right)\\
		\;\;\lesssim\;\;&& \frac{1}{K}\sqrt{\frac{c_1  \ln n}{n}}(\sqrt{p}+\sqrt{e^\varepsilon q})+\frac{1}{K^2}\frac{c_1 \ln n}{n}\;.
	\end{eqnarray*}
	We now analyze the variance. Express the polynomial $h_{2K}(x,y)$ explicitly as
	\begin{eqnarray}
		h_{2K}(x,y) &\;\; = \;\;& \sum_{0 \leq i \leq 2 K, 0 \leq j \leq 2 K, i+j \geq 1} h_{i j} x^{i} y^{j}\\
		&\;\; = \;\;& \sum_{0 \leq i \leq 2 K}\left(\sum_{0 \leq j \leq 2 K, i+j \geq 1} h_{i j} y^{j}\right) x^{i}\;.
	\end{eqnarray}
	For any fixed value of $y$, $h_{2K}(x^2,y^2)$ is a polynomial of $x$ with degree no more than $4K$ that is uniformly bounded by a universal constant on $[0,1]$. It follows from Lemma~\ref{lem:Qopt_magnitude} that for any fixed $y\in [-1,1]$,
	\begin{eqnarray}
		\left|\sum_{0 \leq j \leq 2 K} h_{i j} y^{2 j}\right| \;\;\leq\;\; M(\sqrt{2}+1)^{4 K}\;,
	\end{eqnarray}
	which together with Lemma~\ref{lem:Qopt_magnitude}, implies that
	\begin{eqnarray}
		\left|h_{i j}\right| \;\;\leq\;\; M(\sqrt{2}+1)^{8 K}\;.
	\end{eqnarray}
	Since $\tD_K^{(1)}$ is the unbiased estimator of $2\Delta  h_{2K}(\frac{p}{2\Delta}, \frac{e^{\varepsilon}q}{2\Delta})$, we know
	\begin{eqnarray}
		\tD_{K}^{(1)}(\hp, \hq)\;\;=\;\;  \sum_{0 \leq i, j \leq 2 K, i+j \geq 1} h_{i j}(2 \Delta)^{1-i-j} e^{j\varepsilon} g_{i, 0}(\hp) g_{j, 0}(\hq)\;,\label{eq:opt1_estimator}
	\end{eqnarray}
	where $g_{i,0}(\hp) = \prod_{k=0}^{i-1}(\hp-\frac{k}{m})$ introduced by Lemma~\ref{lem:Qopt_moment}.

	Denote $\|X\|_2 = \sqrt{\E(X-\E X)^2}$ for random variable $X$, and $M_1 = 2K \vee 2n\Delta$, $M_2 = 2K \vee 2ne^{-\varepsilon}\Delta$. Using triangle inequality of the norm $\|\cdot\|_2$ and Lemma~\ref{lem:Qopt_moment}, we know
	\begin{eqnarray*}
		\|\tD_{K}^{(1)}(\hp, \hq)\|_2 &\;\; \leq \;\;& \sum_{0 \leq i, j \leq 2 K, i+j \geq 1}\left|h_{i j}\right|(2 \Delta)^{1-i-j}e^{j\varepsilon}\left\|g_{i, 0}(\hp)\right\|_{2}\left\|g_{j, 0}(\hq)\right\|_{2}\\
		&\;\; \leq \;\;& \sum_{0 \leq i, j \leq 2 K, i+j \geq 1} M(\sqrt{2}+1)^{8 K}(2 \Delta )\left(\frac{1}{2 \Delta } \sqrt{\frac{2 M_{1} p}{n}}\right)^{i}\left(\frac{e^{\varepsilon}}{2 \Delta} \sqrt{\frac{2 M_{2} q}{n}}\right)^{j} \\
		&\;\; \leq \;\;& (\sqrt{2}+1)^{8 K} \frac{c_{1} \ln n}{n} \sum_{0 \leq i, j \leq 2 K, i+j \geq 1}\left(\sqrt{\frac{p}{2 \Delta}}\right)^{i}\left(\sqrt{\frac{e^{\varepsilon}q}{2 \Delta}}\right)^{j}\;.
	\end{eqnarray*}
	Since for any $x\in [0,1]$. $y\in [0,1]$, 
	\begin{eqnarray*}
		\left|\sum_{0 \leq i, j \leq 2 K, i+j \geq 1} x^{i} y^{j}\right| &\;\;\leq\;\;&\left|\sum_{j=1}^{2 K} y^{j}\right|+\left|\sum_{i=1}^{2 K} x^{i}\right|+x y\left|\sum_{0 \leq i, j \leq 2 K-1} x^{i} y^{j}\right|\\
		&\;\;\leq\;\;& y(2 K)+x(2 K)+x y(2 K)^{2}\\
		&\;\;\leq\;\;& 2(2 K)^{2}(x+y)\;,
	\end{eqnarray*}
	we know
	\begin{eqnarray}
		\|\tD_{K}^{(1)}(\hp, \hq)\|_2 &\;\; \lesssim \;\; & (\sqrt{2}+1)^{8 K} \frac{c_{1} K^2 \ln n}{n} \left(\sqrt{\frac{p}{2 \Delta}}+\sqrt{\frac{e^{\varepsilon}q}{2 \Delta}}\right)\\
		&\;\; \lesssim \;\; & \sqrt{B^K\frac{c_1 c_3^5\ln^5 n}{n}\left(p+e^{\varepsilon}q\right)}\;,
	\end{eqnarray}
	for some constant $B>0$.

\subsubsection{Proof of Lemma~\ref{lem:biasvariance2}}
\label{sec:proof_biasvariance2}
	We first analyze the bias. As we apply the unbiased estimator $\tD_K^{(2)}(\hp, \hq ; x, y)$ of $D_K^{(2)}(p,q;x,y)$. The bias is entirely due to the functional approximation error. Namely,	\begin{eqnarray}
	\E\left[\tD_K^{(2)}(\hp,\hq; x,y)|x,y\right] &\;\;= \;\;& D^{(2)}_K\left(p,q; x, y\right) \\ &\;\;=\;\;&\frac{1}{2}\sum_{j=0}^Kr_j W^{-j+1}(e^{\varepsilon}q-p)^j+\frac{p-e^\varepsilon q }{2}\;,
\end{eqnarray}
where $W = \sqrt{\frac{8c_1\ln n}{n}}\left(\sqrt{e^{\varepsilon}x+y}\right)$, and $r_j$ is defined as the coefficient of best polynomial approximation $R_K(t)$ of $|t|$ over $[-1,1]$ with order $K$: $R_K(t) = \sum_{j=0}^Kr_jt^j$.

	Since $(p, e^{\varepsilon}q) \in U(c_1,c_1)$, we know
	\begin{eqnarray*}
		|p-e^{\varepsilon}q| &\;\;\leq\;\;& \sqrt{\frac{2c_1\ln n}{n}}\left(\sqrt{p}+\sqrt{e^{\varepsilon}q}\right)\\
		&\;\;\leq\;\;& \sqrt{\frac{2c_1\ln n}{n}}\sqrt{2}\left(\sqrt{p+e^{\varepsilon}q}\right)\\
		&\;\;\leq\;\;& \sqrt{\frac{2c_1\ln n}{n}}\sqrt{2}\left(\sqrt{2(x+e^{\varepsilon}y)}\right)\\
		&\;\;\leq\;\;& W\,
	\end{eqnarray*}
	where we have used the fact that $\sqrt{p}+\sqrt{e^{\varepsilon}q}\leq \sqrt{2(p+e^{\varepsilon}q)}$ and the assumption that $p+e^{\varepsilon}q\leq 2(x+e^{\varepsilon}y)$.
	
	It is known in \cite[Chapter~9, Theorem~3.3]{devore1993constructive} that 
	\begin{eqnarray}
		\left|R_k(t)-|t| \right|\;\;\lesssim\;\; \frac{1}{K}\;,\label{eq:Rapprox}
	\end{eqnarray}
	 for all $t\in [-1,1]$.
	  
	We show $D_K^{(2)}(p,q;x,y)$ is best polynomial approximation of $[p-e^\varepsilon q]^+$. We have
	\begin{eqnarray*}
		\left|D^{(2)}_K\left(p,q; x, y\right) - [p-e^\varepsilon q]^+\right|&\;\;=\;\;& 
		\frac{1}{2}\left|\sum_{j=0}^Kr_j W^{-j+1}(e^{\varepsilon}q-p)^j - |p-e^{\varepsilon}q|\right|\\
		&\;\;=\;\;& \frac{ W}{2}\left|R_K\Big(\frac{e^{\varepsilon}q-p}{W}\Big) - \Big|\frac{e^{\varepsilon}q-p}{W}\Big|\right|\\
		&\;\;\lesssim\;\;& \frac{ W}{K}\\
		&\;\;\lesssim\;\;&  \frac{1}{K} \sqrt{\frac{c_1 \ln n}{n}}(\sqrt{x}+\sqrt{e^\varepsilon y})\;.
	\end{eqnarray*}
	
	Now we analyze the variance. 
	
	Let $a_j = r_j$ for $j=0,2,3,\ldots, K$ and $a_1 = r_1-1$ we can write $D^{(2)}_K\left(p,q; x, y\right)$ as
	\begin{eqnarray}
		D^{(2)}_K\left(p,q; x, y\right) \;\;=\;\; \frac{1}{2}\sum_{j=0}^Ka_j W^{-j+1}(e^{\varepsilon}q-p)^j\;.
	\end{eqnarray}
	It was shown in \cite[Lemma~2]{cai2011testing} that $r_j\leq 2^{3K}$, $0\leq j\leq K$. So we have $|a_j|\leq 2\cdot 2^{3K}$. Denote the unique uniformly minimum unbiased estimator (MVUE) of $(e^{\varepsilon}q-p)^j$ by $\hat{A}_j(\hp,\hq)$. Then the unbiased estimator $\tD_K^{(2)}$ of polynomial function $D_K^{(2)}$ is
	\begin{eqnarray}
		\tD_K^{(2)}(\hp,\hq; x,y)\;\;=\;\; \frac{1}{2}\sum_{j=0}^Ka_j W^{-j+1}\hat{A}_j(\hp,\hq)\;.\label{eq:opt2_estimator}
	\end{eqnarray}
	
	Denote $\|X\|_2 = \sqrt{\E(X-\E X)^2}$ for random variable $X$. It follows from triangle inequality of $\|\cdot\|_2$ and Lemma~\ref{lem:uumvue} that
	
	\begin{eqnarray}
		\|\tD_K^{(2)}(\hp,\hq; x,y)\|_2 &\;\; \leq \;\;& \frac{1}{2}\sum_{j=1}^K|a_j|W^{-j+1}\|\hat{A}_j\|_2\\
		&\;\; \leq \;\;&  2^{3K} W \sum_{j=1}^K\left(\frac{\sqrt{2}|e^{\varepsilon}q-p|}{W} \vee \frac{\sqrt{8 j(e^{2\varepsilon}q \vee p)}}{\sqrt{n} W}\right)^{j}\;,
	\end{eqnarray}
	where
	\begin{eqnarray*}
		&&\frac{\sqrt{2}|e^{\varepsilon}q-p|}{W} \vee \frac{\sqrt{8 j(e^{2\varepsilon}q \vee p)}}{\sqrt{n} W} \\
		\;\;\leq\;\; && \frac{\sqrt{2} \sqrt{\frac{2 c_{1} \ln n}{n}}(\sqrt{e^{\varepsilon}q}+\sqrt{p})}{\sqrt{\frac{8 c_{1} \ln n}{n}} \sqrt{e^{\varepsilon}y+x}} \vee \frac{\sqrt{8 K(e^{2\varepsilon}q+p)}}{\sqrt{n} \sqrt{\frac{8 c_{1} \ln n}{n}} \sqrt{e^{\varepsilon}y+x}}\\
		\;\;\leq\;\; && \frac{\sqrt{e^{\varepsilon}q}+\sqrt{p}}{\sqrt{2} \sqrt{e^{\varepsilon}y+x}} \vee \sqrt{\frac{c_{3}}{c_{1}}} \frac{\sqrt{e^{2\varepsilon}q+p}}{\sqrt{e^{\varepsilon}y+x}}\\
		\;\;\leq\;\; && \sqrt{2} \vee \sqrt{\frac{2 c_{3}e^{\varepsilon}}{c_{1}}}\\
		\;\;\leq\;\; && \sqrt 2\;.
	\end{eqnarray*}
	Consequently, 
	\begin{eqnarray*}
		\|\tD_K^{(2)}(\hp,\hq; x,y)\|_2 &\;\; \leq \;\;&   2^{3K} W K (\sqrt{2})^K\\
		&\;\; \leq \;\;&  2^{3K} \sqrt{\frac{8 c_{1} \ln n}{n}} \sqrt{x+e^{\varepsilon}y} K (\sqrt{2})^K\\
		&\;\; \leq \;\;& \sqrt{B^{K} \frac{(x+e^{\varepsilon}y) c_{1} \ln n}{n}}\;,
	\end{eqnarray*}
	where $B>0$ is some constant.

\section{Conclusion}

We investigate the fundamental trade-off between accuracy and sample size in estimating differential privacy guarantees from 
a black-box access to a purportedly private mechanism. 
Such a data-driven approach to verifying privacy guarantees will allow us to 
hold accountable the mechanisms in the wild that are not faithful to the claimed privacy guarantees, and 
help find and fix bugs in either the design or the implementation. 
To this end, we propose a polynomial approximation based approach to estimate the differential privacy guarantees. 
We show that in the high-dimensional regime, the proposed estimator achieves 
{\em sample size amplification} effect. 
Compared to the parametric rate achieved by the plug-in estimator, 
we achieve a factor of $\ln n$ gain in the sample size. 
A matching lower bound proves the minimax optimality of our approach. 
Here, we list important remaining challenges that are outside the scope of this paper. 

Since the introduction of differential privacy, there have been 
several innovative notions of privacy, 
such as pufferfish, concentrated DP, zCDP, and Renyi DP, 
proposed in  \cite{KM14,DR16,SWC17,KOV17}. 
Our estimator builds upon the fact that differential privacy guarantee is a divergence between two random outputs.
This is no longer true for the other notions of privacy, which makes it more challenging. 

Characterizing the fundamental tradeoff for continuous mechanisms is an important problem, as 
several popular mechanisms output continuous random variables, 
such as Laplacian and Gaussian mechanisms. 
One could use non-parametric estimators such as $k$-nearest neighbor methods and 
kernel methods, popular for estimating  information theoretic quantities and divergences \cite{BSY19,GOV18,GOV16,JGH18}. 
Further, when the output is a mixture of discrete and continuous variables, 
recent advances in estimating mutual information for mixed variables provide a 
guideline for such complex estimation process \cite{GKOV17}. 

There is a fundamental connection between differential privacy and 
ROC curves, as investigated in \cite{KOV17,KOV14,KOV15}. 
Binary hypothesis testing and ROC curves provide an important 
measure of performance in generative adversarial networks (GAN) \cite{SBL18}. 
This fundamental connection between differential privacy and GAN was first investigated in
\cite{LKFO18}, where it was used to provide an implicit bias for mitigating mode collapse, 
a fundamental challenge in training GANs. 
A DP estimator, like the one we proposed, provides valuable tools to measure performance of GANs. 
The main challenge is that GAN outputs are extremely high-dimensional (popular examples being $1,024\times 1,024,\times 3$ 
dimensional images). 
Non-parametric methods have exponential dependence in the dimension, rendering them useless. 
Even some recent DP approaches have output dimensions that are equally large \cite{HKC18}. 
We need fundamentally different approach to deal with such high dimensional continuous mechanisms. 

We considered a setting where we create synthetic databases $\cD$ and $\cD'$ and 
test the guarantees of a mechanism of interest. 
Instead, \cite{GM18} assumes we do not have such a control, and 
the privacy of the real databases  used in the testing needs to also be preserved. 
It is proven that one cannot test the privacy guarantee of a mechanism without 
revealing the contents of the test databases.  
Such fundamental limits suggest that 
the samples used in estimating DP needs to be destroyed after the estimation. 
However, the estimated $d_\varepsilon(P_{\cQ,\cD}\|P_{\cQ,\cD'})$ 
still leaks some information about the databases used, although limited.  
This is related to a challenging task of 
designing mechanisms with $(\varepsilon,\delta)$-DP guarantees when 
$(\varepsilon,\delta)$ also depends on the databases. 
Without answering any queries, just publishing the guarantee of the mechanism on a set of databases reveal something about the database. Detection and estimation under such complicated constraints is a challenging open question. 

%
%

%
%
\bibliographystyle{alpha}

\bibliography{ref.bib}

\end{document}